\documentclass[12pt]{amsart}
\usepackage{amssymb,amsthm,amsmath,enumerate}
\usepackage[numbers,sort&compress]{natbib}
\usepackage{color,hyperref,url}
\usepackage{graphicx}
\usepackage{tikz,tikz-qtree,float}

\title[]{Optimal bounds for embedded eigenvalues  of one-dimensional discrete Schr\"odinger operators with decaying potentials}
\author[]
{Jifeng Chu$^{1}$, \quad Wencai Liu$^2$, \quad Kang Lyu$^{1}$}
\address{$^1$ School of Mathematics, Hangzhou Normal University, Hangzhou 311121, China}
\address{$^2$ Department of Mathematics, Texas A\&M University, College Station, TX 77843-3368, USA}

\email{jifengchu@126.com (J. Chu)}
\email{liuwencai1226@gmail.com; wencail@tamu.edu (W. Liu)}
\email{lvkang201905@outlook.com (K. Lyu)}


\newcommand{\abs}[1]{\left\lvert #1 \right\rvert}

\theoremstyle{plain}
\newtheorem{theorem}{Theorem}[section]
\newtheorem{corollary}[theorem]{Corollary}
\newtheorem{lemma}[theorem]{Lemma}
\newtheorem{proposition}[theorem]{Proposition}

\newcommand{\R}{\mathbb{R}}
\newcommand{\Z}{\mathbb{Z}}
\theoremstyle{definition}

\newtheorem{example}[theorem]{Example}
\newtheorem{remark}[theorem]{Remark}

\begin{document}


	\begin{abstract}
	In this paper, we consider one-dimensional discrete Schr\"odinger operators
	\begin{align}
		Hu(n)=(\Delta+V)u(n)\nonumber
	\end{align}
	on  $\ell^2(\mathbb{N})$ with a self-adjoint boundary condition at $n=0$, where $\Delta$ denotes the discrete Laplacian and $V(n)$ is a real-valued perturbation satisfying 
	$$V(n)=\frac{O(1)}{1+n}.$$ 
We determine the sharp transition for the asymptotic coefficient of \(V\) governing the existence and nonexistence of embedded eigenvalues.
\end{abstract}
\maketitle

\tableofcontents

\section{Introduction}\label{sec1}

Consider the one-dimensional discrete Schr\"odinger operator acting on $\ell^2(\mathbb N)$, given by
\begin{align}\label{maineq}
	(Hu)(n)
	=(\Delta +V)u(n)
	=u(n+1)+u(n-1)+V(n)u(n),
	\qquad n\in\mathbb N,
\end{align}
where
\[
(\Delta u)(n)=u(n+1)+u(n-1)
\]
denotes the discrete Laplacian and $V$ is a real-valued potential. We impose a self-adjoint
boundary condition at the origin, written as
\begin{align}\label{bc}
	u(1)\cos\varphi-u(0)\sin\varphi=0.
\end{align}
The essential and absolutely continuous spectra of $\Delta$ are both
equal to $[-2,2]$. Throughout this paper, we consider potentials at the
critical Coulomb scale
\begin{align}
	V(n)=\frac{O(1)}{1+n},
	\qquad n\to\infty.
	\nonumber
\end{align}
Such perturbations preserve the absolutely continuous spectrum $[-2,2]$,
but may produce eigenvalues embedded in its interior \cite{RemlingCMP,LiuDiscrete,DeiftKillip,KillipSimon}.
Our first objective is to determine the sharp asymptotic constant required to
create the prescribed embedded eigenvalue. For a potential $V$, set
\begin{align}
	a(V):=\limsup_{n\to\infty}|nV(n)|.
	\nonumber
\end{align}
Given $E\in(-2,2)$, we seek the critical value $\mathcal C(E)$ characterized by
the following two properties:
\begin{description}
	\item[(i)] if $a(V)<\mathcal C(E)$, then $E$ is not an eigenvalue of
	$\Delta+V$;
	
	\item[(ii)] for every $\varepsilon>0$ and every prescribed boundary
	condition \eqref{bc}, there exists a potential $V$ satisfying
	\[
	a(V)<\mathcal C(E)+\varepsilon
	\]
	such that $E$ is an eigenvalue of $\Delta+V$.
\end{description}
Thus, $\mathcal C(E)$ is the sharp threshold for the asymptotic Coulomb coefficient governing the existence and nonexistence of an embedded eigenvalue at the energy $E$.

The corresponding problem for continuous Schr\"odinger operators is
classical. Consider
\begin{align}
	H^c=-\frac{d^2}{dx^2}+V^c(x)
	\nonumber
\end{align}
on $L^2(\mathbb R_+)$, and set
\[
a_c(V^c):=\limsup_{x\to\infty}|xV^c(x)|.
\]
The Wigner--von Neumann construction \cite{vonNeumannWigner} first showed that
an oscillatory potential at the critical scale $x^{-1}$ can produce a
positive embedded eigenvalue. Kato's classical estimate implies, in
particular, that $V^c(x)=o(x^{-1})$ rules out positive embedded
eigenvalues \cite{Kato}. At the critical scale, Atkinson and
Everitt proved the sharp bound
\begin{align}
	E\leq \left(\frac{2a_c(V^c)}{\pi}\right)^2
	\nonumber
\end{align}
for every positive embedded eigenvalue and showed that the constant $\frac{2}{\pi}$ is
optimal \cite{AtkinsonEveritt}, where Wigner--von Neumann  type oscillatory potentials play a crucial role. Equivalently,
the continuous critical coefficient is
\begin{align}
	\mathcal C_c(E)=\frac{\pi}{2}\sqrt{E},
	\qquad E>0.
	\nonumber
\end{align}
Several directions have subsequently developed from the Wigner--von Neumann mechanism.  Simon constructed
examples with dense embedded point spectrum
\cite{SimonDense}.  For individual oscillatory terms, one can
obtain eigenfunction
asymptotics, spectral density information, and a
description of exceptional resonant energies; see
\cite{Behncke,HintonKlausShaw,LukicJST,LukicCMP,SimonovZeros}. Wigner--von Neumann and sharp transition questions have also been
considered for periodic Schr\"odinger operators
\cite{NabokoSimonov,LukicOng,LyuYang},
for Stark type operators
\cite{LiuStarkJFA}, and for Dirac operators
\cite{KhapreLyuYu}.  Related work includes 
embedded eigenvalues for periodic graph operators \cite{Shipman},
and absence or uniqueness results for higher-dimensional
operators \cite{LiuFermi,LiYangNLS}. See \cite{SimonJacobiDense,JudgeNabokoWood,Kruger,RemlingDiscrete,LiuLyuResonant} for more results.

The critical decay scale also separates several other spectral
phenomena.  Christ and Kiselev proved optimal stability results for
the absolutely continuous spectrum under slowly decaying
one-dimensional perturbations \cite{ChristKiselev}, and related
results were obtained by Kiselev \cite{KiselevDuke}, Remling
\cite{RemlingCMP}, and Deift and Killip \cite{DeiftKillip}.  Modified
Pr\"ufer and EFGP variables provide a common framework for many of
these arguments \cite{KiselevLastSimonPruefer,LastSimonDense}; effective perturbation
methods and periodic backgrounds were treated in
\cite{KiselevRemlingSimon,Stolz}.  Transfer matrix and subordinacy
criteria give complementary descriptions of the absolutely
continuous and singular parts of the spectrum
\cite{GilbertPearson,JitomirskayaLast,LastSimonInvent}.
At slower or more irregular decay, singular continuous spectrum becomes
possible.  Pearson constructed decaying potentials with singular
continuous spectrum \cite{Pearson}, and Kiselev constructed singular
continuous spectrum embedded in the positive spectrum
\cite{KiselevJAMS}.  Quantitative restrictions and counterexamples
were studied in \cite{RemlingProc,RemlingDuke}.  More generally, sharp constants and extremal eigenvalue problems
arise in a variety of spectral problems; see, for example,
\cite{AshbaughBenguria,FeffermanSecoEnergy,FeffermanSecoODE,ChuMeng,ChuMengZhang}.

The discrete problem is more delicate because of its dependence on the arithmetic properties of the energy. To describe this dependence, for any $E\in (-2,2)$, write
\begin{align}
	E=2\cos\pi k,
	\qquad k=k(E)\in(0,1),
	\nonumber
\end{align}
and, when convenient, write $\mathcal C(k)$ for $\mathcal C(E)$. We say that $E$ is of irrational type if $k\notin\mathbb Q$ and of rational type if $k\in\mathbb Q$. In the rational case, we write
\[
k=\frac pq,
\qquad 1\leq p\leq q-1,
\qquad (p,q)=1.
\]
A striking feature of the discrete problem is that the critical coefficient $\mathcal C(k)$ is continuous at every irrational $k$ and discontinuous at every rational $k$; see \cite{LiuDiscrete}. We explain below how this phenomenon arises from the different averaging mechanisms associated with irrational and rational quasi-momenta.

The continuous problem is closely related to the irrational discrete problem. In the continuous setting, the leading logarithmic contribution, for example on the interval $(x_1,x_2)$, to the Pr\"ufer radius is governed by
\begin{align}
	\int_{x_1}^{x_2}
	\frac{|\sin(2\sqrt E\,x+\varphi)|}{x}\,dx,
	\nonumber
\end{align}
and the long scale average of its oscillatory factor is
\[
\frac{1}{2\pi}\int_0^{2\pi}|\sin t|\,dt
=\frac{2}{\pi},
\]
independently of the initial phase $\varphi$. In the discrete problem, the corresponding expression is
\begin{align}
	\sum_{n=n_1}^{n_2}
	\frac{|\sin(2\pi nk+2\pi\varphi)|}{n}.
	\nonumber
\end{align}
When $k$ is irrational, the rotation
\[
\varphi\longmapsto\varphi+k\pmod{\mathbb Z}
\]
is uniquely ergodic. Consequently,
\begin{align}
	\lim_{N\to\infty}\frac1N
	\sum_{n=0}^{N-1}
	|\sin(2\pi nk+2\pi\varphi)|
	=\frac{2}{\pi}
	\nonumber
\end{align}
uniformly in $\varphi$. Thus, the irrational discrete problem has the same phase-independent averaging mechanism as the continuous problem. An adaptation of Remling's argument gives the corresponding nonexistence estimate \cite{RemlingCMP}, and its sharpness was established in \cite{LiuDiscrete}.

The situation changes fundamentally when $k=p/q$ is rational. The rotation then has a finite orbit, and the average over one period is
$P(\varphi)/q$, where
\begin{align}
	P(\varphi):=
	\sum_{j=0}^{q-1}
	\left|
	\sin\left(\frac{2\pi}{q}j+2\pi\varphi\right)
	\right|.
	\nonumber
\end{align}
Unlike the irrational average, this quantity depends on both the phase $\varphi$ and the denominator $q$. Moreover, in the spectral problem, $\varphi$ is the evolving Pr\"ufer angle rather than a fixed external parameter. The {radial} decay and angular evolution are therefore coupled, making the rational problem substantially more difficult.

The parity of $q$ introduces a further distinction. When $q$ is even, the points in the finite orbit can be paired through
\[
j\longleftrightarrow j+\frac q2,
\]
using the identity
\[
\sin(t+\pi)=-\sin t.
\]
This half-period symmetry allows the leading angular displacements to cancel while the {radial} contributions add. In \cite{LiuDiscrete}, this observation was combined with an almost sign-type potential and a finite-dimensional inverse argument to keep the Pr\"ufer phase near a maximizer of $P(\varphi)$. This yields the sharp critical coefficient in the even denominator case.

For odd $q$, no analogous half-period pairing is available. The potential values that improve the decay of the Pr\"ufer radius generally also move the Pr\"ufer angle, so a favorable phase cannot simply be preserved from one period to the next. Resolving this interaction between radial decay and angular motion is the main difficulty of the odd denominator problem and the principal contribution of the present paper.

Our main theorem gives the complete sharp transition in the odd denominator case.
For odd $q\geq3$, define
\begin{align}\label{bqb}
	B_q:=\frac{1+\cos\frac{\pi}{q}}{q\sin\frac{\pi}{q}}.
\end{align}
\begin{theorem}\label{main}	Suppose that
	\[
	k=\frac pq\in(0,1),
	\qquad (p,q)=1,
	\]
	where $q\geq3$ is odd, and let $E=2\cos\pi k$. Then the following statements hold:
	\begin{enumerate}
		\item[(i)]  If
		\[
		\limsup_{n\to\infty}|nV(n)|
		<
		\frac{\sin\pi k}{B_q},
		\]
		then $E$ is not an eigenvalue of $\Delta+V$.
		
		\item[(ii)]  For every $\varepsilon>0$ and any prescribed boundary condition \eqref{bc}, there exists a potential $V$ satisfying
		\[
		\limsup_{n\to\infty}|nV(n)|
		<
		\frac{\sin\pi k}{B_q}+\varepsilon
		\]
		such that $E$ is an eigenvalue of $\Delta+V$.
	\end{enumerate}
	Consequently,
	\begin{align}
		\mathcal C(k)
		=
		\frac{\sin\pi k}{B_q},
		\qquad q\geq3\ \text{odd}.
		\nonumber
	\end{align}
\end{theorem}

{\begin{remark}
Before the present work, the odd denominator estimates obtained in \cite{LiuDiscrete} were
\begin{align}
\frac{\sin\pi k}{A_q}
\leq \mathcal C(k)
\leq
\frac{\sin\pi k}{B_q},
\qquad q\geq3\ \text{odd},
\nonumber
\end{align}
where $A_q$ is defined by \eqref{aqa} below.
Since $A_q\neq B_q$ for odd $q$, these estimates did not determine the sharp coefficient. Theorem \ref{main} replaces the lower bound by the matching value $\frac{\sin\pi k}{B_q}$ and closes this gap.
Assertion $(ii)$ of Theorem \ref{main} was proved by Liu in \cite{LiuDiscrete} and is included here for completeness. Assertion $(i)$ of Theorem \ref{main}, which provides the matching nonexistence bound and therefore completes the sharp transition, is the new contribution of the present paper.
\end{remark}

Together with \cite{LiuDiscrete}, this completes the optimal
single eigenvalue coefficient problem for every $E\in(-2,2)$: 
\begin{align}
	\mathcal C(k)=
	\begin{cases}
		\dfrac{\sin\pi k}{A_0},
		& \quad k\notin\mathbb Q,\\[12pt]
		\dfrac{\sin\pi k}{A_q},
		& \quad q\geq2\ \text{even},\\[12pt]
		\dfrac{\sin\pi k}{B_q},
		& \quad q\geq3\ \text{odd},
	\end{cases}
	\label{demck}
\end{align}
where
\begin{align}\label{da0}
	A_0=\frac{2}{\pi}
\end{align}
and
\begin{align}\label{aqa}
	A_q:=
	\begin{cases}
		\dfrac{2}{q\sin\frac{\pi}{q}},
		& q\geq2\ \text{even},\\[15pt]
		\dfrac{2\cos\frac{\pi}{2q}}{q\sin\frac{\pi}{q}},
		& q\geq3\ \text{odd}.
	\end{cases}
\end{align}
This formula makes the continuity contrast precise. At an irrational
$k$, rational approximants necessarily have denominators tending to
infinity, and both rational coefficients converge to $A_0$. At a
rational $k=\frac{p}{q}$, however, we have
\[
\lim_{\substack{k'\to p/q\\ k'\neq p/q}}\mathcal C(k')
=\frac{\sin\pi k}{A_0},
\]
whereas
\[
A_q>A_0
\quad(q\ \text{even}),
\qquad
B_q<A_0
\quad(q\ \text{odd}),
\]
from which we can see that $\mathcal C(k)$ is continuous precisely at irrational
quasi-momenta and discontinuous at every rational quasi-momentum.
The same statement holds for $\mathcal C(E)$ at energies of
irrational and rational type, respectively.}

Theorem \ref{main} cannot be obtained by optimizing the radial Pr\"ufer equation pointwise. The potential values that improve the
decay of the Pr\"ufer radius also alter the angle and therefore change
the phases encountered in later periods. The proof must control this
feedback for every admissible potential and for every possible
long time behavior of the angle. We describe the main
ingredients.

\medskip
\noindent\textbf{1. A dynamical substitute for the missing
	half-period symmetry.}
Let $R(n)$ and $\theta(n)$ be the modified Pr\"ufer variables of a
solution of \eqref{maineq}.  For a sign-type potential, the leading
decrease of $\ln R^2$ over a rational period is governed by
$P(\varphi)$, while the same potential also produces a leading
displacement of the angle.  The distinction between even and odd
denominators first appears in the relation between these two effects.

When $q$ is even, the points in one orbit can be paired by
$j\leftrightarrow j+q/2$.  The corresponding sine terms have opposite
signs, whereas the corresponding sine squares agree.  Consequently,
the two halves of a period can be balanced so that their leading
angular displacements cancel while their radial contributions add.
After one initial adjustment, the angle can be returned after every
period to the position 
$$\theta(nq)=\frac{1}{2q}\pmod\Z,$$
which is an
optimal phase for the even denominator trigonometric sum.  Thus, in
the even case, the phase giving the best decay can also be made
dynamically stable.  The same favorable configuration can then be
used period after period.

For odd $q$, there is no half-period pairing.  The sign choice that
improves the decay of $R$ generally moves $\theta$ away from its
current position, so the optimal phase cannot simply be fixed and
reused.  This is the basic instability of the odd denominator
problem: the potential that is favorable for the radius is not
neutral for the angle.
To quantify this interaction, we introduce, beside the radial
function $P(\varphi)$, a second trigonometric function $T(\varphi)$ measuring the
first order angular drift during one period.  Thus $P(\varphi)$ records how
much decay is gained at a given phase, while $T(\varphi)$ records how far that
phase is displaced by the same choice of potential.  The pointwise
bound $P(\varphi)\geq qB_q$ (see \eqref{Bq}) may at first look unfavorable for a
nonexistence theorem with the constant $B_q$: away from the
minimizers, the radial term is strictly larger.  The angle, however,
is not a freely chosen parameter.  Exactly where this apparent extra
decay occurs, the drift encoded by $T(\varphi)$ forces the orbit to move.
The key trigonometric identity can be summarized as
$P(\varphi)-F(\varphi)T(\varphi)=qB_q$ for an explicit coefficient
$F(\varphi)$ (see \eqref{ghs0} for details).  It says that the excess of the instantaneous radial
term over the sharp constant is precisely tied to angular motion.
Moreover, reflection preserves $P(\varphi)$ and reverses the sign of $T(\varphi)$.
Hence periods with opposite drifts may be paired without losing their
radial information.  The exact half-period symmetry used for even
$q$ is therefore replaced, when $q$ is odd, by a dynamical
compensation principle along the orbit.  This principle explains at
a conceptual level why the sharp odd denominator constant is $B_q$
rather than $A_q$.

This changes the nature of the optimization problem.  In the even
case, one can optimize at a single phase and then preserve that phase.
In the odd case, the relevant object is an entire trajectory: a
temporary gain at one phase must be evaluated together with the
future motion it creates.  The sharp estimate emerges only after the
radial gain and the accumulated drift are considered simultaneously.
This viewpoint is also the reason that a direct period-by-period
estimate gives the wrong impression.  Such an estimate sees the value
of $P(\varphi)$ at the current phase but does not yet see the compensating
motion forced upon the subsequent phases.

\medskip
\noindent\textbf{2. Uniform local estimates for arbitrary
	Coulomb scale potentials.}
To pass from the extremal sign choice to an arbitrary potential, write
\[
V(n)=\frac{a_n\,a}{1+n},
\qquad |a_n|\leq1,
\]
and introduce the defect quantities $P_V$ and $T_V$. The period
expansions show that the radial equation involves $P-P_V$, whereas the
angular equation involves $T+T_V$. Since both defects are generated by
the same coefficients $a_n$, they cannot vary independently.
Coefficientwise trigonometric inequalities convert this dependence
into sharp local compensation estimates.

Near the special phases $\frac{j}{2q},j=0,1,\cdots,2q$, one has
$P(\varphi)=qB_q+O(\delta)$, so a direct estimate is sufficient. Away
from these phases, a sharp estimate follows on
sets of periods having bounded net angular displacement \eqref{bqjia0}. When a one sided balance is unavailable, the estimate
is applied simultaneously near reflected phases $\varphi$ and
$1-\varphi$, whose drift terms have opposite signs  \eqref{lemmaduichenfinal}. These estimates
are uniform over all sequences $|a_n|\leq1$.

\medskip
\noindent\textbf{3. A partition result for slowly varying sequences.}
We consider the slowly varying sequence which satisfies, as $n\to\infty,$
\[
x_{n+1}-x_n=\frac{O(1)}{1+n}.
\]
Its orbit need not converge or be monotone. We prove that the
accumulation set of $[x_n]$ on $\mathbb R/\mathbb Z$ has one of four
forms: a point, an ordinary closed arc, a closed arc crossing the
identification point $0$, or the whole circle, where 
\begin{align}
	[x_n]:=x_n-\lfloor x_n\rfloor
\end{align}
denotes the fractional part of $x_n$ and $\lfloor \cdot\rfloor$ is the floor function.
These four cases
classify possible angle dynamics; they should not be confused with
the two arithmetic regimes for $k$ described above.
In the first three cases, long sets of indices can be partitioned into
small phase windows so that each useful piece has bounded total
increment. In the whole circle
case, the indices are instead paired in windows near $\varphi$ and
$1-\varphi$, and the difference of the two accumulated increments is
bounded. These are exactly the hypotheses required by the local
one-window and reflected-window compensation estimates.

The discrete nature of this partition is essential. A sequence may
jump across a level, producing an overshoot of order $\frac{1}{n}$ at every
crossing, and the indices belonging to different phase windows may be
interlaced. We remove crossings caused only by shallow oscillations
and retain the endpoints of genuine excursions. The remaining
entrance--exit pattern produces telescoping cancellation, while slow
variation forces successive unmatched deep excursions to occur at
geometrically separated times, making the residual overshoot errors
summable. Applied to the Pr\"ufer angle, this partition reduces every
possible orbit to one of the local compensation mechanisms and yields
the global bound in Theorem \ref{main}.

We also study the \textbf{critical problem }$a(V)=\mathcal C(E)(=\mathcal{C} (k))$, which is not contained in the above
 noncritical
constructions. Away from the
threshold, the Pr\"ufer radius decays with a power strictly stronger
than the harmonic rate, and square summability follows directly. At
the critical coefficient, the leading decay has exactly harmonic
order, so lower order terms become decisive. At the same time, a
construction may not use any fixed excess in the asymptotic size of
the potential. A separate argument is therefore required to obtain an
$\ell^2$ solution while retaining the exact limiting coefficient.

For reference, we recall the result of \cite{LiuDiscrete} for
irrational quasi-momenta and for rational quasi-momenta with even
denominator  in the language of the present paper. Define
\begin{align}
	S_0:=
	\left\{
	E=2\cos\pi k:
	k\in(0,1)\setminus\mathbb Q
	\right\},
	\nonumber
\end{align}
and, for $q\geq2$, define
\begin{align}
	S_q:=
	\left\{
	E=2\cos\pi k:
	k=\frac pq,\ 
	1\leq p\leq q-1,\ 
	(p,q)=1
	\right\}.
	\nonumber
\end{align}
\begin{theorem}\cite{LiuDiscrete}\label{ccdd}
	Suppose that $q\geq 0$ is even and $E\in S_q$. Then the following statements hold:
	\begin{enumerate}
		\item[(i)] If
		\[
		\limsup_{n\to\infty}|nV(n)|
		<
		\frac{\sin\pi k}{A_q},
		\]
		then $E$ is not an eigenvalue of $\Delta+V$;
		
		\item[(ii)]  For every $\varepsilon>0$ and any  prescribed boundary condition
		\eqref{bc}, there exists a potential $V$ satisfying
		\[
		\limsup_{n\to\infty}|nV(n)|
		<
		\frac{\sin\pi k}{A_q}+\varepsilon,
		\]
		such that $E$ is an eigenvalue of $\Delta+V$.
	\end{enumerate}
\end{theorem}

Our second main result proves that the critical coefficient is attained
at every energy.

\begin{theorem}\label{main2}
	Suppose that $q=0$ or $q\geq 2$ and $E\in S_q$. Let $\mathcal{C}(E)=\mathcal{C}(k)$ be defined by \eqref{demck}.
	Then for any prescribed boundary condition
	\eqref{bc}, there exists a potential $V$ satisfying
	\[
	\limsup_{n\to\infty}|nV(n)|
	=\mathcal C(E),
	\]
	such that $E$ is an eigenvalue of $\Delta+V$.
	
\end{theorem}

The proof of Theorem \ref{main2} turns the borderline decay described
above into a summable one without introducing a fixed excess in the numerator of the
potential.  The half line is divided into rapidly growing blocks, and
the coefficient on the $m$th block is changed only by a quantity
tending to zero.  This correction makes the local decay exponent
slightly larger than one, with the excess itself tending to zero.
Meanwhile, the block endpoints grow fast enough to make all
transition losses summable.  A separate comparison lemma then
upgrades these vanishing blockwise gains to the global conclusion
$R\in\ell^2(\mathbb N)$.  Thus square summability is obtained not from
a uniform power gain, but from an infinite sequence of increasingly
small gains placed on increasingly long scales.
Three scales must consequently be coordinated.  Inside each basic
period or averaging block, the potential has to produce the desired
Pr\"ufer decay.  Across a much longer construction block, the small
gain in the exponent has to accumulate to a visible improvement.  At
the global scale, the blocks must grow fast enough that the errors
from their endpoints do not rebuild a divergent harmonic tail.  A
choice that works on any one of these scales may fail on another:
short blocks do not accumulate enough gain, while poorly separated
long blocks can make the transition errors uncontrollable.  The
comparison lemma is designed to encode the precise balance among
these competing requirements.

The rational and irrational cases require different implementations of this principle.  For odd $q$, the phase remains unstable, but the uniform
inequality $P(\varphi)\geq qB_q$ holds at every phase.  This makes it
possible to choose the vanishing coefficient corrections uniformly
along a drifting orbit, provided all intra-period errors are retained.
For even $q$, the optimal coefficient $A_q$ is attained at the phase
$\frac{1}{2q}$, so the construction must keep the angle at that exact phase
after every period.  A second, still smaller parameter is inserted
into one half of the period and adjusted to cancel the terminal angle
error without changing the limiting size of the potential.  For
irrational $k$, no finite period is available.  One must instead
choose block lengths on which the oscillatory average of the actual
Pr\"ufer phase is close to $A_0=\frac{2}{\pi}$, while simultaneously keeping
the block length small relative to its starting point and satisfying
the growth conditions of the comparison lemma.  Coordinating these
three requirements is the main additional difficulty in the critical
part of the paper.

\subsection*{Notation}

For $s,t\in\mathbb R$ with $s<t$, let
\begin{align}
	[s]:=s-\lfloor s\rfloor
	\label{demodx}
\end{align}denote the fractional part of $s$,
and
\begin{align}
	(s,t)_{\mathrm{mod}}:=\{[x]:x\in(s,t)\}.
	\label{demodint}
\end{align}
We identify $\mathbb R/\mathbb Z$ with the unit circle $S^1$ via
$[x]\mapsto e^{2\pi i x}$ and use the same notation $[x]$ in both
spaces. For $x\in[0,1)$, one has $x\in(s,t)_{\mathrm{mod}}$ if and only
if there exists $n\in\mathbb Z$ such that
\begin{align}
	x+n\in(s,t).
	\label{equstmod}
\end{align}
Equivalently, for $x\in\mathbb R$, one has
$[x]\in(s,t)_{\mathrm{mod}}$ if and only if there exists
$n\in\mathbb Z$ such that
\begin{align}
	x+n\in(s,t).
	\label{equmodr}
\end{align}
From the viewpoint of $S^1$, then $(s,t)_{\rm{mod}}$ covers the entire $S^1$ if $t-s>1$, $(s,t)_{\rm{mod}}$ is strictly contained in $S^1$ if $t-s\leq 1$.
\begin{figure}[H]
	\centering  
	\begin{tikzpicture}[x=0.75pt,y=0.75pt,yscale=-1,xscale=1]
		
		\draw   (255,129.5) .. controls (255,104.37) and (275.37,84) .. (300.5,84) .. controls (325.63,84) and (346,104.37) .. (346,129.5) .. controls (346,154.63) and (325.63,175) .. (300.5,175) .. controls (275.37,175) and (255,154.63) .. (255,129.5) -- cycle ;
		\draw [color={rgb, 255:red, 208; green, 2; blue, 27 }  ,draw opacity=1 ][line width=3]    (276,168.5) .. controls (246,150.5) and (252,102.5) .. (275,91.5) ;
		\draw   (426,127.5) .. controls (426,102.37) and (446.37,82) .. (471.5,82) .. controls (496.63,82) and (517,102.37) .. (517,127.5) .. controls (517,152.63) and (496.63,173) .. (471.5,173) .. controls (446.37,173) and (426,152.63) .. (426,127.5) -- cycle ;
		\draw  [color={rgb, 255:red, 208; green, 2; blue, 27 }  ,draw opacity=1 ][line width=2.25]  (426,127.5) .. controls (426,102.37) and (446.37,82) .. (471.5,82) .. controls (496.63,82) and (517,102.37) .. (517,127.5) .. controls (517,152.63) and (496.63,173) .. (471.5,173) .. controls (446.37,173) and (426,152.63) .. (426,127.5) -- cycle ;
		\draw  [color={rgb, 255:red, 255; green, 255; blue, 255 }  ,draw opacity=1 ][line width=0.75] [line join = round][line cap = round] (514,129) .. controls (517,129) and (520,129) .. (523,129) ;
		\draw   (64,134.5) .. controls (64,109.37) and (84.37,89) .. (109.5,89) .. controls (134.63,89) and (155,109.37) .. (155,134.5) .. controls (155,159.63) and (134.63,180) .. (109.5,180) .. controls (84.37,180) and (64,159.63) .. (64,134.5) -- cycle ;
		\draw  [color={rgb, 255:red, 208; green, 2; blue, 27 }  ,draw opacity=1 ][line width=3]  (63,133.63) .. controls (63,108.01) and (83.76,87.25) .. (109.38,87.25) .. controls (134.99,87.25) and (155.75,108.01) .. (155.75,133.63) .. controls (155.75,159.24) and (134.99,180) .. (109.38,180) .. controls (83.76,180) and (63,159.24) .. (63,133.63) -- cycle ;
		
		\draw (255,173) node [anchor=north west][inner sep=0.75pt]   [align=left] {$[s]$};
		\draw (255,70) node [anchor=north west][inner sep=0.75pt]   [align=left] {$[t]$};
		\draw   (513.75,128.5) .. controls (513.75,126.71) and (515.21,125.25) .. (517,125.25) .. controls (518.79,125.25) and (520.25,126.71) .. (520.25,128.5) .. controls (520.25,130.29) and (518.79,131.75) .. (517,131.75) .. controls (515.21,131.75) and (513.75,130.29) .. (513.75,128.5) -- cycle ;
		\draw (445,210) node [anchor=north west][inner sep=0.75pt]   [align=left] {$t-s=1$};
		\draw (80,210) node [anchor=north west][inner sep=0.75pt]   [align=left] {$t-s>1$};
		\draw (335,80) node [anchor=north west][inner sep=0.75pt]   [align=left] {$S^1$};
		\draw (510,80) node [anchor=north west][inner sep=0.75pt]   [align=left] {$S^1$};
		\draw (140,80) node [anchor=north west][inner sep=0.75pt]   [align=left] {$S^1$};
		\draw (270,210) node [anchor=north west][inner sep=0.75pt]   [align=left] {$t-s<1$};
	\end{tikzpicture}
	\caption{}
	\label{f1k}
\end{figure}

The remainder of the paper is organized as follows. We begin in Section \ref{sec2} by introducing the modified Prüfer transformation and the associated trigonometric functions that govern both the radial and angular evolution. Section \ref{sec3} then turns to the proof of the local one-window and symmetric cancellation estimates. In Section \ref{sec4}, we develop the general partition theory for slowly varying sequences, with the two most intricate partition constructions deferred to Sections \ref{pcb} and \ref{pcd}. The implications of these results for the Prüfer angle are established in Section \ref{secparpru}. The proof of the main result, Theorem \ref{main}, is presented in Section \ref{sec8}, while the critical and boundary constructions for Theorem \ref{main2} are treated separately in Section \ref{sec9}. Finally, the appendices supply the continuous analogue of the partition theorem and a discussion of the behavior under sign-type potentials.

\vspace{2mm}\noindent \textbf{Acknowledgments} The authors thank ChatGPT for its assistance with English-language
editing in the Introduction. Jifeng Chu was supported by the National Natural Science Foundation of China (No. 12571168), the Science and Technology Innovation Plan of Shanghai (No. 23JC1403200) and
Zhejiang Provincial Natural Science Foundation of China (No. LZ26A010006).
Part of this work was completed while Wencai Liu was visiting the
Simons Institute for the Theory of Computing.
	\section{Some Basic Lemmas}\label{sec2}

	We first introduce the modified Pr\"ufer transformation (see \cite{KiselevLastSimonPruefer} and \cite{RemlingCMP} for more details). For any $E\in (-2,2)$, write $E=2\cos \pi k$. Suppose that $u(n)=u(n,E)$ is a real solution of \eqref{maineq}. Let
	\begin{align}
		Y(n,E)=\frac{1}{\sin\pi k}\begin{pmatrix}
			\sin \pi k &  0\\
			-\cos \pi k & 1
		\end{pmatrix}
		\begin{pmatrix}
			u(n-1,E)\\
			u(n,E)
		\end{pmatrix}.\label{yne}
	\end{align}
	Define the Pr\"ufer variables $R(n,E)$ and $\theta(n,E)$ as 
	\begin{align}
		Y(n,E)=R(n,E)\begin{pmatrix}
			\sin (\pi\theta(n,E)-\pi k)\\
			\cos (\pi\theta(n,E)-\pi k)
		\end{pmatrix}.\label{rne}
	\end{align}
In the following arguments, we leave the dependence on $E$ implicit if there is no confusion.	Then $R(n)=R(n,E)$ and $\theta(n)=\theta(n,E)$ obey
	\begin{eqnarray}\label{lr}
		\frac{R(n+1)^2}{ R(n)^2}=1-\frac{V(n)}{\sin\pi k}\sin2\pi\theta(n)+\frac{V(n)^2}{\sin^2\pi k}\sin^2\pi\theta(n)
	\end{eqnarray}
	and
	\begin{eqnarray}\label{ctt}
		\cot (\pi \theta(n+1)-\pi k)=\cot \pi \theta(n)-\frac{V(n)}{\sin\pi k}.
	\end{eqnarray}
In this paper, all potentials $V$ satisfy
\begin{align}
	V(n)=\frac{O(1)}{1+n},\ n\to\infty.\label{vno1}
\end{align}

{\begin{lemma}\cite[Lemma 2.2]{LiuLyuResonant}\label{nn}
	If $\abs{\frac{V(n)}{\sin\pi k}}<\frac{1}{10}$, then one has
	\begin{eqnarray}\label{n1n}
		\theta(n+1)=\theta(n)+k+\sin^2\pi\theta(n)\frac{V(n)}{\pi\sin\pi k}+O\left(\frac{V(n)^2}{\sin^2\pi k}\right),\ \ n\to\infty.
	\end{eqnarray}
	 
\end{lemma}}
By \eqref{n1n}, for any fixed $j\in\mathbb{N}$, one has
\begin{eqnarray}\label{bhgx}
	\theta(n+j)=\theta(n)+kj+\frac{O(1)}{1+n},\ n\to\infty.\label{thetan+jn}
\end{eqnarray}

In the following part of this section, we assume that
$k=\frac{p}{q}$
with coprime $p,q$, where $p<q$. Since $$\left\{0,\frac{1}{q},\frac{2}{q},\cdots,\frac{q-1}{q}\right\}=\left\{0,\frac{p}{q},\frac{2p}{q},\cdots,\frac{q-1}{q}p\right\} \pmod \Z,$$
we may choose $(l_0,\cdots, l_{q-1})$ to be a permutation of $(0,1,\cdots,q-1)$ so that for any $j=0,\cdots,q-1,$ we have
\begin{align}
	\left[\frac{p}{q}l_j\right]=\frac{j}{q}.\label{permupqlj}
\end{align}

Substituting \eqref{thetan+jn} into \eqref{n1n}, we obtain
\begin{align}
	\theta(n+q)&=\theta(n)+kq+\sum_{j=0}^{q-1}\sin^2\pi\left(\theta(n)+\frac{p}{q}j\right)\frac{V(n+j)}{\pi\sin\pi k}+\frac{O(1)}{1+n^2}\nonumber\\
	&=\theta(n)+kq+\sum_{j=0}^{q-1}\sin^2\pi\left(\theta(n)+\frac{p}{q}l_j\right)\frac{V(n+l_j)}{\pi\sin\pi k}+\frac{O(1)}{1+n^2}\nonumber\\
	\label{bhgx2}&=\theta(n)+kq+\sum_{j=0}^{q-1}\sin^2\pi\left(\theta(n)+\frac{j}{q}\right)\frac{V(n+l_j)}{\pi\sin\pi k}+\frac{O(1)}{1+n^2},\ n\to\infty.
\end{align}
Define
\begin{align}
	P(\varphi)=\sum_{j=0}^{q-1}\abs{\sin\left(\frac{2\pi}{q}j+2\pi\varphi\right)}\label{definitionofP_qtheta}
\end{align}
and
\begin{align}
	T(\varphi)=\sum_{j=0}^{q-1}\sin^2\pi\left(\varphi+\frac{j}{q}\right){\rm{sgn}}\left(\sin\left(2\pi\varphi+\frac{2\pi}{q}j\right)\right).\label{definitionofT_qtheta}
\end{align}
It is obvious that for any $\varphi\in\R$,
\begin{align}
	P(\varphi)=P\left(\varphi+\frac{1}{q}\right)\label{perioPvarphi}
\end{align}
and
\begin{align}
		T(\varphi)=T\left(\varphi+\frac{1}{q}\right)\label{perioTvarphi}.
\end{align}
Therefore,  we only study them on the interval $\left[0,\frac{1}{q}\right]$.
	
	\begin{lemma}\label{lsscc}
		Assume that $q\geq 3$ is odd. Then for any $\varphi\in \mathbb{R}$, one has
		\begin{align}\label{pqduichen}
			P(\varphi)=P\left(\frac{1}{q}-\varphi\right).
		\end{align} 
		Specifically, we have
		\begin{align}
			\label{sscc}P(\varphi)=\begin{cases}
				\frac{\cos2\pi\varphi+\cos\left(\frac{\pi}{q}-2\pi\varphi\right)}{\sin\frac{\pi}{q}},&\varphi\in\Big[0,\frac{1}{2q}\Big),\\
				\frac{\cos\left(\frac{\pi}{q}-2\pi\varphi\right)+\cos\left(\frac{2\pi}{q}-2\pi\varphi\right)}{\sin\frac{\pi}{q}},& \varphi\in\left[\frac{1}{2q},\frac{1}{q}\right].
			\end{cases}
		\end{align}
		In particular,
\begin{align}
\min_{\varphi\in\left[0,\frac{1}{q}\right]}P(\varphi)=P(0)=P\left(\frac{1}{2q}\right)=P\left(\frac{1}{q}\right)=qB_q,\label{Bq}\\
P(\varphi)>qB_q,\ \mathrm{for\ any\ } \varphi\in \left(0,\frac{1}{2q}\right)\cup\left(\frac{1}{2q},\frac{1}{q}\right),\label{Bqleq1toqPqtheta}
\end{align}
		where  $B_q$ is defined by \eqref{bqb}. 
	\end{lemma} 
	\begin{proof}
		Once we have \eqref{sscc}, direct computations will show  \eqref{Bq}, \eqref{Bqleq1toqPqtheta}, and \eqref{pqduichen} for any $\varphi\in\left[0,\frac{1}{q}\right]$, then by \eqref{perioPvarphi}
one can obtain		
		 \eqref{pqduichen} for any $\varphi\in\mathbb{R}$.
		We only prove \eqref{sscc}.
		
		Since 
		\begin{align}
			\sum_{i=0}^n\sin ix=\frac{\cos\frac{x}{2}-\cos\left(nx+\frac{x}{2}\right)}{2\sin\frac{x}{2}},\ \sum_{i=0}^n\cos ix=\frac{\sin\frac{x}{2}+\sin\left(nx+\frac{x}{2}\right)}{2\sin\frac{x}{2}},\label{sinixcosix}
		\end{align}
		one can obtain that for any $\varphi\in\mathbb{R}$,
		\begin{align}
			\sum_{i=0}^n\sin (\varphi+ix)=\frac{\cos\left(\varphi-\frac{x}{2}\right)-\cos\left(\varphi+\left(n+\frac{1}{2}\right)x\right)}{2\sin\frac{x}{2}}.\label{sintheta+ix}
		\end{align}
		Hence, one has
		\begin{align}
			\sum_{j=0}^{q-1}\sin\left(\frac{2\pi}{q}j+2\pi\varphi\right)=\frac{\cos\left(2\pi\varphi-\frac{\pi}{q}\right)-\cos\left(2\pi\varphi+(q-\frac{1}{2})\frac{2\pi}{q}\right)}{2\sin\frac{\pi}{q}}=0.\nonumber
		\end{align}
		Therefore, for any $\varphi\in\Big[0,\frac{1}{2q}\Big)$, by \eqref{definitionofP_qtheta} we have
		\begin{align}
			P(\varphi)=\sum_{j=0}^{\frac{q-1}{2}}\sin\left(\frac{2\pi}{q}j+2\pi\varphi\right)-\sum_{j=\frac{q+1}{2}}^{q-1}\sin\left(\frac{2\pi}{q}j+2\pi\varphi\right)
			=2\sum_{j=0}^{\frac{q-1}{2}}\sin\left(\frac{2\pi}{q}j+2\pi\varphi\right).\nonumber
		\end{align}
		Then applying \eqref{sintheta+ix} one can obtain \eqref{sscc} for $\varphi\in\Big[0,\frac{1}{2q}\Big)$.
		
		For any $\varphi\in\left[\frac{1}{2q},\frac{1}{q}\right]$, by \eqref{definitionofP_qtheta} we have
		\begin{align}
			P(\varphi)=&\sum_{j=0}^{\frac{q-3}{2}}\sin\left(\frac{2\pi}{q}j+2\pi\varphi\right)-\sum_{j=\frac{q-1}{2}}^{q-1}\sin\left(\frac{2\pi}{q}j+2\pi\varphi\right)
			=2\sum_{j=0}^{\frac{q-3}{2}}\sin\left(\frac{2\pi}{q}j+2\pi\varphi\right).\nonumber
		\end{align}
		Then applying \eqref{sintheta+ix} we can obtain \eqref{sscc} for $\varphi\in\left[\frac{1}{2q},\frac{1}{q}\right]$.
	\end{proof}

	\begin{lemma}\label{lssccs}Assume that $q\geq3$ is odd. Then for any $\varphi\in \mathbb{R}$, one has
		\begin{align}\label{tqduichen}
			T(\varphi)=-T\left(\frac{1}{q}-\varphi\right).
		\end{align}
		Specifically,	we have
		\begin{align}
			T(\varphi)=\begin{cases}
				\label{ssccs}\frac{1}{2}+\frac{\sin2\pi\varphi-\sin\left(\frac{\pi}{q}-2\pi\varphi\right)}{2\sin\frac{\pi}{q}}, & \varphi\in[0,\frac{1}{2q}),\\
				0, & \varphi=\frac{1}{2q},\\
				-\frac{1}{2}-\frac{\sin\left(\frac{2\pi}{q}-2\pi\varphi\right)+\sin\left(\frac{\pi}{q}-2\pi\varphi\right)}{2\sin\frac{\pi}{q}}, & \varphi\in  (\frac{1}{2q}, \frac{1}{q}],
			\end{cases}
		\end{align}
		and
		\begin{align}
			\label{mp1}&\lim_{\varphi\to\frac{1}{2q}\pm0}T(\varphi)=\mp1,\\
			\label{tqpdy0}T^{\prime}(\varphi)>0&, \mathrm{\ for\ any\ }\varphi\in\left(0,\frac{1}{2q}\right)\cup\left(\frac{1}{2q},\frac{1}{q}\right).
		\end{align}
	\end{lemma}
	\begin{proof}
		Once we have \eqref{ssccs}, direct computations will show \eqref{mp1}, \eqref{tqpdy0}, and \eqref{tqduichen} for any $\varphi\in\left[0,\frac{1}{q}\right]$. Then by \eqref{perioTvarphi} one can obtain \eqref{tqduichen} for any $\varphi\in\mathbb{R}$.
		We only prove \eqref{ssccs}.
		
		Applying \eqref{sintheta+ix}, one can obtain that for any $\varphi\in\mathbb{R}$,
		\begin{align}
			\sum_{i=0}^n\cos(\varphi+ix)=\frac{\sin\left(\varphi+\left(n+\frac{1}{2}\right)x\right)-\sin\left(\varphi-\frac{x}{2}\right)}{2\sin\frac{x}{2}}.\label{sumcostheta+ix}
		\end{align}
		Hence, one has
		\begin{align}
			\sum_{j=0}^{q-1}\cos\left(2\pi\varphi+\frac{2\pi}{q}j\right)=0.\label{sumj=0q-1cos=0}
		\end{align}
	By \eqref{definitionofT_qtheta}, for any $\varphi\in\Big[0,\frac{1}{2q}\Big)$, 
		\begin{align}
			T(\varphi)
			=&\sum_{j=0}^{\frac{q-1}{2}}\sin^2\pi\left(\frac{j}{q}+\varphi\right )-\sum_{j=\frac{q+1}{2}}^{q-1}\sin^2\pi\left(\frac{j}{q}+\varphi\right)\nonumber\\
			=&\sum_{j=0}^{\frac{q-1}{2}}\frac{1-\cos\left(2\pi\varphi+\frac{2\pi}{q}j\right)}{2}-\sum_{j=\frac{q+1}{2}}^{{q-1}}\frac{1-\cos\left(2\pi\varphi+\frac{2\pi}{q}j\right)}{2}.\nonumber
		\end{align}
Using \eqref{sumcostheta+ix} and \eqref{sumj=0q-1cos=0}, one has
\begin{align}
	T(\varphi)=\frac{1}{2}-\sum_{j=0}^{\frac{q-1}{2}}\cos\left(2\pi\varphi+\frac{2\pi}{q}j\right)=\frac{1}{2}+\frac{\sin2\pi\varphi-\sin\left(\frac{\pi}{q}-2\pi\varphi\right)}{2\sin\frac{\pi}{q}}.\nonumber
\end{align}
Thus, we obtain \eqref{ssccs} for  $\varphi\in [0,\frac{1}{2q})$.

By \eqref{definitionofT_qtheta},  for $\varphi=\frac{1}{2q}$,
\begin{align}
	T\left(\frac{1}{2q}\right)
	=&\sum_{j=0}^{\frac{q-3}{2}}\frac{1-\cos\left(\frac{\pi}{q}+\frac{2\pi}{q}j\right)}{2}-\sum_{j=\frac{q+1}{2}}^{{q-1}}\frac{1-\cos\left(\frac{\pi}{q}+\frac{2\pi}{q}j\right)}{2}\nonumber\\
	=&-\frac{1}{2}\left(\sum_{j=0}^{\frac{q-3}{2}}\cos\left(\frac{\pi}{q}+\frac{2\pi}{q}j\right)-\sum_{j=\frac{q+1}{2}}^{{q-1}}\cos\left(\frac{\pi}{q}+\frac{2\pi}{q}j\right)\right).\nonumber
\end{align}
Since for any $j\in\left\{0,1,\cdots,\frac{q-3}{2}\right\}$ and any $l\in\left\{\frac{q+1}{2},\cdots,q-1\right\}$ with $j+l=q-1$, 
\begin{align}
\cos\left(\frac{\pi}{q}+\frac{2\pi}{q}j\right)=\cos\left(\frac{\pi}{q}+\frac{2\pi}{q}l\right),\nonumber
\end{align}
we have $T\left(\frac{1}{2q}\right)=0$.

		 By \eqref{definitionofT_qtheta}, for any $\varphi\in \Big(\frac{1}{2q},\frac{1}{q}\Big]$,
		\begin{align}
			T(\varphi)=&\sum_{j=0}^{\frac{q-3}{2}}\sin^2\pi\left(\frac{j}{q}+\varphi\right )-\sum_{j=\frac{q-1}{2}}^{q-1}\sin^2\pi\left(\frac{j}{q}+\varphi\right)\nonumber\\
			=&\sum_{j=0}^{\frac{q-3}{2}}\frac{1-\cos\left(2\pi\varphi+\frac{2\pi}{q}j\right)}{2}-\sum_{j=\frac{q-1}{2}}^{{q-1}}\frac{1-\cos\left(2\pi\varphi+\frac{2\pi}{q}j\right)}{2}\nonumber.
		\end{align}
		Hence, applying \eqref{sumcostheta+ix} and \eqref{sumj=0q-1cos=0}, one can obtain 
		\begin{align}
			T(\varphi)=-\frac{1}{2}-\sum_{j=0}^{\frac{q-3}{2}}\cos\left(2\pi\varphi+\frac{2\pi}{q}j\right)=-\frac{1}{2}-\frac{\sin\left(\frac{2\pi}{q}-2\pi\varphi\right)+\sin\left(\frac{\pi}{q}-2\pi\varphi\right)}{2\sin\frac{\pi}{q}}.\nonumber
		\end{align}
		Thus, we obtain \eqref{ssccs} for $\varphi\in\Big(\frac{1}{2q},\frac{1}{q}\Big].$
	\end{proof}
	
		\begin{lemma} Assume that $q\geq3$ is odd. For any $\varphi\in\left[0,\frac{1}{q}\right]$, we have
		\begin{align}\label{ghs0}
			P(\varphi)-\frac{2\sin\left(\frac{\pi}{q}-2\pi\varphi\right)}{1+\cos\left(\frac{\pi}{q}-2\pi\varphi\right)}T(\varphi)=\frac{1+\cos\frac{\pi}{q}}{\sin\frac{\pi}{q}}.
		\end{align}
	\end{lemma}
	\begin{proof} Since we have Lemma \ref{lsscc} and Lemma \ref{lssccs}, the result follows from direct computations.
	\end{proof}

	\section{Technical Preparations}\label{sec3}
	In this section, we assume that $k=\frac{p}{q}$ with odd $q\geq 3$. 
	
	\begin{lemma}
		Let $\varphi\in\left(0,\frac{1}{q}\right)$ and $j\in\mathbb{Z}$. If
		\begin{align}
			\sin\left(2\pi\varphi+\frac{2\pi}{q}j\right)\geq 0,\label{lem3.1sin>0}
		\end{align}
then
			\begin{align}\label{qiuhe}
				\frac{\sin\left(2\pi\varphi+\frac{2\pi}{q}j\right)}{\sin^2\pi\left(\varphi+\frac{j}{q}\right)}\geq \frac{2\sin\left(\frac{\pi}{q}-2\pi\varphi\right)}{1+\cos\left(\frac{\pi}{q}-2\pi\varphi\right)}.
			\end{align}
			If 
			\begin{align}
				\sin\left(2\pi\varphi+\frac{2\pi}{q}j\right)\leq 0,\label{lem3.1sin<0}
			\end{align}
		then
			\begin{align}\label{qiuhe2}
				\frac{\sin\left(2\pi\varphi+\frac{2\pi}{q}j\right)}{\sin^2\pi\left(\varphi+\frac{j}{q}\right)}\leq \frac{2\sin\left(\frac{\pi}{q}-2\pi\varphi\right)}{1+\cos\left(\frac{\pi}{q}-2\pi\varphi\right)}.
			\end{align}
	\end{lemma}
\begin{proof} Since $\sin\left(2\pi\varphi+\frac{2\pi}{q}j\right)$ and $\sin^2\pi\left(\varphi+\frac{j}{q}\right)$ are both $q$-periodic in $j$, we only need to prove the result for $j\in \{0,1,\cdots,q-1\}$. 
	
	For any $\varphi\in \left(0,\frac{1}{q}\right)$, if we have \eqref{lem3.1sin>0},
	then  $0\leq j\leq \frac{q-1}{2}$, if we have \eqref{lem3.1sin<0},
	then $\frac{q-1}{2}\leq j \leq q-1.$
	For any $x\in (0,2\pi)$, one has
	\begin{align}
		\left(\frac{\sin x}{\sin^2\frac{x}{2}}\right)^{\prime}=\frac{2(\cos x-1)}{(1-\cos x)^2}<0.\nonumber
	\end{align}
Therefore, for any $0\leq j\leq \frac{q-1}{2}\leq \tilde{j}\leq q-1$,
\begin{align}
	\frac{\sin\left(2\pi\varphi+\frac{2\pi}{q}j\right)}{\sin^2\pi\left(\varphi+\frac{j}{q}\right)}\geq \frac{\sin\left(
		2\pi\varphi+\frac{2\pi}{q}\frac{q-1}{2}\right)}{\sin^2\pi\left(\varphi+\frac{q-1}{2q}\right)}=\frac{2\sin\left(\frac{\pi}{q}-2\pi\varphi\right)}{1+\cos \left(\frac{\pi}{q}-2\pi\varphi\right)}\geq \frac{\sin\left(2\pi\varphi+\frac{2\pi}{q}\tilde{j}\right)}{\sin^2\pi\left(\varphi+\frac{\tilde{j}}{q}\right)}.\nonumber
\end{align}
	Hence, we can obtain \eqref{qiuhe} and \eqref{qiuhe2}. 
\end{proof}

	To simplify the following formulae, we introduce some notation here. Recall that $(l_0,\cdots, l_{q-1})$ is a permutation of $(0,1,\cdots,q-1)$ so that for any $j=0,\cdots,q-1,$ we have
	\begin{align}
		\left[\frac{p}{q}l_j\right]=\frac{j}{q}.\label{finalmod}
	\end{align}
	Let $a>0$ be fixed. For any potential $V=V(n)$, define $a_n$ by $V(n)=\frac{a_na}{1+n}$ for each $n\geq 1$. Define
	\begin{align}
			P_V(\varphi,n)=P(\varphi)-\sum_{j=0}^{q-1}a_{n+l_j}\sin\left(2\pi\varphi+\frac{2\pi}{q}j\right)\label{definitionofPVtheta}
	\end{align}
	and 
	\begin{align}
		T_V(\varphi,n)=\sum_{j=0}^{q-1}a_{n+l_j}\sin^2\pi\left(\varphi+\frac{j}{q}\right)-T(\varphi),\label{definitionofTVtheta}
	\end{align}
	where $P(\varphi)$ and $T(\varphi)$ are defined by \eqref{definitionofP_qtheta} and \eqref{definitionofT_qtheta}, respectively.
	Clearly, definitions of $P_V(\varphi,n)$ and $T_V(\varphi,n)$
 depend on $a$. However, since $a$
 will always be fixed in what follows, we suppress this dependence for notational simplicity.
	By \eqref{perioPvarphi} and \eqref{perioTvarphi}, $P(\varphi),T(\varphi)$ are $\frac{1}{q}$-periodic. Then we have that $P_V(\varphi,n)$ and $T_V(\varphi,n)$ are both $1$-periodic in $\varphi$. Namely,
	for any $\varphi\in\R$, one has
	\begin{align}
		P_V(\varphi,n)=P_V\left(\varphi+1,n\right)\label{periodPV}
	\end{align}
	and 
	\begin{align}
		T_V(\varphi,n)=T_V\left(\varphi+1,n\right)\label{periodTV}.
	\end{align}
	\begin{lemma}
	Let $a>0$ be fixed.	Suppose that $\theta(n)$ is determined by \eqref{ctt} with $V(n)=\frac{a_na}{1+n}$, where $\abs{a_n}\leq 1$ for each $n\geq 1$. Then one has 
		\begin{align}\label{relationoftheta}
			\theta(n+q)=\theta(n)+kq+\frac{aT(\theta(n))+aT_V(\theta(n),n)}{(1+n)\pi\sin\pi k}+\frac{O(1)}{1+n^2},\ n\to\infty
		\end{align}
		and 
		\begin{align}\label{pqthetapathetan}
			\sum_{j=0}^{q-1}V(n+j)\sin2\pi\theta(n+j)=\frac{aP(\theta(n))-aP_V(\theta(n),n)}{1+n}+\frac{O(1)}{1+n^2},\ n\to\infty.
		\end{align}
	\end{lemma}
	\begin{proof}
		One can obtain \eqref{relationoftheta} by \eqref{bhgx2} and \eqref{definitionofTVtheta}, and obtain \eqref{pqthetapathetan} by \eqref{bhgx}, \eqref{finalmod} and \eqref{definitionofPVtheta}.
	\end{proof}
	
	Recall that $[s]$ and $(s,t)_{\mathrm{mod}}$ are defined by \eqref{demodx} and \eqref{demodint}, respectively.
	
	\begin{lemma}\label{lemmafixedpoints}	Let $a>0$ be fixed.
			Suppose that $\theta(n)$ is determined by \eqref{ctt} with $\abs{V(n)}\leq\frac{a}{1+n}$.  Assume $0<\delta\ll 1$. Let $\{n_m\}_{m=0}^{d}\subset\mathbb{N}$ satisfy
			\begin{align}
				\frac{1}{\delta}<n_0<n_1<n_2<\cdots<n_d,
			\end{align} 
			and be such that
			\begin{align}
				\left[\theta(n_m)\right]\in \cup_{j=0}^{2q}{\left(\frac{j}{2q}-\delta,\frac{j}{2q}+\delta\right)_{\mathrm{mod}}},\ m=0,1,\cdots,d.\label{lemmafixedpointsthetarange}
			\end{align}
		Then
		we have 
		\begin{align}\label{lemmafixedpointsqBq}
			\frac{\sum_{m=0}^{d}\sum_{j=0}^{q-1}V(n_{m}+j)\sin2\pi\theta(n_{m}+j)}{\sum_{m=0}^{d}\frac{a}{1+n_{m}}}\leq qB_q+O(\delta),\ d\to\infty.
		\end{align}
	\end{lemma}
	\begin{proof}
		Since $\abs{V(n)}\leq \frac{a}{1+n}$, for any $m=0,1,\cdots,d,$ 
		\begin{align}
	\frac{\sum_{j=0}^{q-1}V(n_{m}+j)\sin2\pi\theta(n_{m}+j)}{\frac{a}{1+n_{m}}}\leq \frac{\sum_{j=0}^{q-1}\frac{a\abs{\sin2\pi\theta(n_{m}+j)}}{1+n_m+j}}{\frac{a}{1+n_{m}}}
	\leq \sum_{j=0}^{q-1}\abs{\sin2\pi\theta(n_{m}+j)}.\label{lemma3.4sum}
			\end{align}
		By \eqref{thetan+jn}, \eqref{definitionofP_qtheta} and $n_0>\frac{1}{\delta}$, 
		\begin{align}
			\sum_{j=0}^{q-1}\abs{\sin2\pi\theta(n_{m}+j)}=&\sum_{j=0}^{q-1}\abs{\sin\left(2\pi\theta(n_m)+\frac{2\pi}{q}pj\right)}+O(\delta)\nonumber\\
			=&\sum_{j=0}^{q-1}\abs{\sin\left(2\pi\theta(n_m)+\frac{2\pi}{q}j\right)}+O(\delta)\nonumber\\
			=&P(\theta(n_m))+O(\delta), \ m\to\infty.\label{lemma3.4pq}
		\end{align}
		By Lemma \ref{lsscc}, we know that $P(\varphi)$ is piecewise continuously differentiable. Then by \eqref{perioPvarphi}, \eqref{Bq} and \eqref{lemmafixedpointsthetarange}, one has that 
		\begin{align}
			P(\theta(n_m))\leq qB_q+O(\delta),\ m\to\infty.\label{lemma3.4ptheta}
		\end{align}
		Then one can obtain \eqref{lemmafixedpointsqBq} by \eqref{lemma3.4sum}, \eqref{lemma3.4pq} and \eqref{lemma3.4ptheta}.
	\end{proof}
	
		\begin{lemma}\label{lemmaduichencancel}	Let $a>0$ be fixed.
		Suppose that $\theta(n)$ is determined by \eqref{ctt} with $\abs{V(n)}\leq\frac{a}{1+n}$. 
		Assume that $\varphi_0\in (0,1)\setminus \left\{\frac{j}{2q}\right\}_{j=0}^{2q}$,  $\tilde{\varphi}_0=1-\varphi_0$ and
		\begin{align}
			0<\delta\ll \min\left\{\abs{\varphi_0-\frac{j}{2q}}:j=0,1,\cdots,2q\right\}.\label{dlld1}
		\end{align} 
	Let $\{n_m\}_{m=0}^d\subset\mathbb{N},\{\tilde{n}_m\}_{m=0}^{\tilde{d}}\subset\mathbb{N}$ satisfy
		\begin{align}
			\frac{1}{\delta}<n_0<n_1<n_2<\cdots<n_d,\nonumber\\
			\frac{1}{\delta}<\tilde{n}_0<\tilde{n}_1<\tilde{n}_2<\cdots<\tilde{n}_{\tilde{d}},\nonumber
		\end{align}
		and be such that
		\begin{align}
			\label{theta0neighbour}\left[\theta(n_m)\right]\in (\varphi_0-3\delta,\varphi_0+3\delta),\ m=0,1,\cdots,d,\\
			\label{thetatilde0neighbour}\left[\theta(\tilde{n}_m)\right]\in (\tilde{\varphi}_0-3\delta,\tilde{\varphi}_0+3\delta),\ m=0,1,\cdots,\tilde{d}.
		\end{align}
		Then if
		\begin{align}
			\abs{\sum_{m=0}^d(\theta(n_m+q)-\theta(n_m)-kq)-\sum_{m=0}^{\tilde{d}}(\theta(\tilde{n}_m+q)-\theta(\tilde{n}_m)-kq)}\leq 10,\label{thetanmtildethetanmleq10}
		\end{align}
		and 
		\begin{align}\label{tiaohejishugeq}
		\sum_{m=0}^{d}\frac{1}{1+n_{m}}+\sum_{m=0}^{\tilde{d}}\frac{1}{1+\tilde{n}_{m}}\geq \frac{1}{\delta},
		\end{align}		
		we have
		\begin{align}
			\frac{\sum_{m=0}^{d}\sum_{j=0}^{q-1}V(n_{m}+j)\sin2\pi\theta(n_{m}+j)+\sum_{m=0}^{\tilde{d}}\sum_{j=0}^{q-1}V(\tilde{n}_{m}+j)\sin2\pi\theta(\tilde{n}_{m}+j)}{\sum_{m=0}^{d}\frac{a}{1+n_{m}}+\sum_{m=0}^{\tilde{d}}\frac{a}{1+\tilde{n}_{m}}}\label{lemmaduichenfinal}\\
			\leq qB_q+O(\delta),\ d+{\tilde{d}}\to\infty.\nonumber
		\end{align}
	\end{lemma}
	
	\begin{proof}
		Suppose that $\varphi_0\in \left(\frac{j_0}{q},\frac{j_0+1}{q}\right)$ for some $j_0=0,1,\cdots,q-1$. We only need to prove the result for the case 
		\begin{align}
			\varphi_0\in \left(\frac{j_0}{q},\frac{j_0}{q}+\frac{1}{2q}\right),\label{ficase}
		\end{align}
		then one can obtain the result by interchanging $\varphi_0$ and $\tilde{\varphi}_0$ if $\varphi_0\in \left(\frac{j_0}{q}+\frac{1}{2q},\frac{j_0+1}{q}\right)$.
		
			Let us first consider the denominator of \eqref{lemmaduichenfinal}. Denote by 
$V(n)=\frac{a_na}{1+n}$
		with $\abs{a_n}\leq 1$, combining \eqref{relationoftheta} with \eqref{thetanmtildethetanmleq10}, as $ d+{\tilde{d}}\to\infty$,
		\begin{align}
			\sum_{m=0}^{d}\frac{T(\theta(n_{m}))+T_V(\theta(n_{m}),n_{m})}{1+n_{m}}-\sum_{m=0}^{\tilde{d}}\frac{T(\theta(\tilde{n}_{m}))+T_V(\theta(\tilde{n}_{m}),\tilde{n}_{m})}{1+\tilde{n}_{m}}=O(1).\label{smallthanOvarepsilonsum}
		\end{align}
			By \eqref{perioTvarphi} and \eqref{ssccs}, we know that $T(\varphi)$ is smooth near $\varphi_0, \tilde{\varphi}_0$. Then applying \eqref{definitionofTVtheta}, $T_V(\varphi,n)$ is smooth near $\varphi_0,\tilde{\varphi}_0$  with respect to $\varphi$. Therefore, by
		\eqref{theta0neighbour} and \eqref{thetatilde0neighbour}, and that $T(\varphi),T_V(\varphi,n)$ are both $1$-periodic in $\varphi$, one has that as $m\to\infty$, 
		\begin{align}
			T(\theta(n_m))=T(\varphi_0)+O(\delta),\ \  T_V(\theta(n_m),n_m)=T_V(\varphi_0,n_m)+O(\delta),\nonumber
		\end{align}
		and
		\begin{align}
			T(\theta(\tilde{n}_m))=T(\tilde{\varphi}_0)+O(\delta),\ \  T_V(\theta(\tilde{n}_m),\tilde{n}_m)=T_V(\tilde{\varphi}_0,\tilde{n}_m)+O(\delta). \nonumber
		\end{align}
		Therefore, by \eqref{smallthanOvarepsilonsum}, as $d+\tilde{d}\to\infty$,
		\begin{align}
			\sum_{m=0}^{d}\frac{T(\varphi_0)+T_V(\varphi_0,n_{m})}{1+n_{m}}-\sum_{m=0}^{\tilde{d}}\frac{T(\tilde{\varphi}_0)+T_V(\tilde{\varphi}_0,\tilde{n}_{m})}{1+\tilde{n}_{m}}\nonumber\\
			=O(1)+ \sum_{m=0}^{d}\frac{O(\delta)}{1+n_{m}}+\sum_{m=0}^{\tilde{d}}\frac{O(\delta)}{1+\tilde{n}_{m}}.\nonumber
		\end{align}
		Thus, by \eqref{tiaohejishugeq}, as $d+\tilde{d}\to\infty$,
		\begin{align}
		\sum_{m=0}^{d}\frac{T(\varphi_0)+T_V(\varphi_0,n_{m})}{1+n_{m}}-\sum_{m=0}^{\tilde{d}}\frac{T(\tilde{\varphi}_0)+T_V(\tilde{\varphi}_0,\tilde{n}_{m})}{1+\tilde{n}_{m}}\nonumber\\
			=O(\delta)\left( \sum_{m=0}^{d}\frac{1}{1+n_{m}}+\sum_{m=0}^{\tilde{d}}\frac{1}{1+\tilde{n}_{m}}\right),\nonumber
		\end{align}
		namely,
		\begin{align}
			\label{l3.5sumOvarepsilon}\sum_{m=0}^{\tilde{d}}\frac{T_V(\tilde{\varphi}_0,\tilde{n}_{m})}{1+\tilde{n}_{m}}-\sum_{m=0}^{d}\frac{T_V(\varphi_0,n_{m})}{1+n_{m}}=\sum_{m=0}^{d}\frac{T(\varphi_0)+O(\delta)}{1+n_{m}}-\sum_{m=0}^{\tilde{d}}\frac{T(\tilde{\varphi}_0)+O(\delta)}{1+\tilde{n}_{m}}.
		\end{align}
		By \eqref{perioTvarphi}, \eqref{tqduichen}, \eqref{ssccs} and \eqref{ficase}, we know that $T(\varphi_0)=-T(\tilde{\varphi}_0)>0$. By \eqref{ssccs} and \eqref{dlld1} we have $\delta\ll T(\varphi_0)$,  by \eqref{l3.5sumOvarepsilon}, we conclude that
		\begin{align}
			 \sum_{m=0}^{d}\frac{1}{1+n_{m}}+\sum_{m=0}^{\tilde{d}}\frac{1}{1+\tilde{n}_{m}}=\frac{\sum_{m=0}^{\tilde{d}}\frac{T_V(\tilde{\varphi}_0,\tilde{n}_{m})}{1+\tilde{n}_{m}}-\sum_{m=0}^{d}\frac{T_V(\varphi_0,n_{m})}{1+n_{m}}}{T(\varphi_0)+O(\delta)}.\label{l3.5denominator}
		\end{align}
Then we have
\begin{align}
	{\sum_{m=0}^{\tilde{d}}\frac{T_V(\tilde{\varphi}_0,\tilde{n}_{m})}{1+\tilde{n}_{m}}-\sum_{m=0}^{d}\frac{T_V(\varphi_0,n_{m})}{1+n_{m}}}>0.\label{denominator>0}
\end{align}
		
		Now let us consider the numerator of \eqref{lemmaduichenfinal} (denoted by $S$).
		By \eqref{pqthetapathetan} one has
		\begin{align}
			S=&\sum_{m=0}^{d}\left(\frac{aP(\theta(n_{m}))-aP_V(\theta(n_{m}),n_{m})}{1+n_{m}}+\frac{O(1)}{1+n_{m}^2}\right)\label{desfi}\\
			&+\sum_{m=0}^{\tilde{d}}\left(\frac{aP(\theta(\tilde{n}_{m}))-aP_V(\theta(\tilde{n}_{m}),\tilde{n}_{m})}{1+\tilde{n}_{m}}+\frac{O(1)}{1+\tilde{n}_{m}^2}\right).\nonumber
		\end{align}
		Since $n_0>\frac{1}{\delta}$ and $\tilde{n}_0>\frac{1}{\delta}$, we have
		\begin{align}
		\label{nmsquaredOvarepsilon}	\frac{1}{1+n_m^2}=\frac{O(\delta)}{1+n_m},\ m\to\infty,\\
		\label{nmtildesquaredOvarepsilon}	\frac{1}{1+\tilde{n}_m^2}=\frac{O(\delta)}{1+\tilde{n}_m},\ m\to\infty.
		\end{align}
	By \eqref{perioPvarphi} and \eqref{sscc}, we know that $P(\varphi)$ is smooth near $\varphi_0,\tilde{\varphi}_0$. Then by \eqref{definitionofPVtheta}, $P_V(\varphi,n)$ is smooth near $\varphi_0,\tilde{\varphi}_0$  with respect to $\varphi$. Therefore, from
	\eqref{theta0neighbour} and \eqref{thetatilde0neighbour}, and that $P(\varphi),P_V(\varphi,n)$ are both $1$-periodic in $\varphi$, one has
\begin{align}
	P(\theta(n_m))=P(\varphi_0)+O(\delta),\ \  P_V(\theta(n_m),n_m)=P_V(\varphi_0,n_m)+O(\delta),\ m\to\infty,\nonumber
\end{align}
and
\begin{align}
	P(\theta(\tilde{n}_m))=P(\tilde{\varphi}_0)+O(\delta),\ \  P_V(\theta(\tilde{n}_m),\tilde{n}_m)=P_V(\tilde{\varphi}_0,\tilde{n}_m)+O(\delta),\ m\to\infty. \nonumber
\end{align}		
Then by \eqref{desfi},
		\begin{align}
			S=\sum_{m=0}^{d}\frac{aP(\varphi_0)-aP_V(\varphi_0,n_{m})+O(\delta)}{1+n_{m}}+\sum_{m=0}^{\tilde{d}}\frac{aP(\tilde{\varphi}_0)-aP_V(\tilde{\varphi}_0,\tilde{n}_{m})+O(\delta)}{1+\tilde{n}_{m}},\ d+\tilde{d}\to\infty.\nonumber
		\end{align}
		By \eqref{perioPvarphi} and \eqref{pqduichen}, we know that $P(\varphi_0)=P(\tilde{\varphi}_0),$ then
		\begin{align}
			S=&\left(aP(\varphi_0)+O(\delta)\right)\left(\sum_{m=0}^{d}\frac{1}{1+n_{m}}+\sum_{m=0}^{\tilde{d}}\frac{1}{1+\tilde{n}_{m}}\right)-\sum_{m=0}^{d}\frac{aP_V(\varphi_0,n_{m})}{1+n_{m}}\label{l3.5numerator}\\
			&-\sum_{m=0}^{\tilde{d}}\frac{aP_V(\tilde{\varphi}_0,\tilde{n}_{m})}{1+\tilde{n}_{m}},\ d+\tilde{d}\to\infty.\nonumber
		\end{align}
		By \eqref{l3.5denominator} and \eqref{l3.5numerator}, we can obtain 
		\begin{align}
			&\frac{\sum_{m=0}^{d}\sum_{j=0}^{q-1}V(n_{m}+j)\sin2\pi\theta(n_{m}+j)+\sum_{m=0}^{\tilde{d}}\sum_{j=0}^{q-1}V(\tilde{n}_{m}+j)\sin2\pi\theta(\tilde{n}_{m}+j)}{\sum_{m=0}^{d}\frac{a}{1+n_{m}}+\sum_{m=0}^{\tilde{d}}\frac{a}{1+\tilde{n}_{m}}}\nonumber\\
			=&\frac{\left(aP(\varphi_0)+O(\delta)\right)\left(\sum_{m=0}^{d}\frac{1}{1+n_{m}}+\sum_{m=0}^{\tilde{d}}\frac{1}{1+\tilde{n}_{m}}\right)-\sum_{m=0}^{d}\frac{aP_V(\varphi_0,n_{m})}{1+n_{m}}
				-\sum_{m=0}^{\tilde{d}}\frac{aP_V(\tilde{\varphi}_0,\tilde{n}_{m})}{1+\tilde{n}_{m}}}{\sum_{m=0}^{d}\frac{a}{1+n_{m}}+\sum_{m=0}^{\tilde{d}}\frac{a}{1+\tilde{n}_{m}}}\nonumber\\
			=&P(\varphi_0)+O(\delta)-\frac{\sum_{m=0}^{d}\frac{P_V(\varphi_0,n_{m})}{1+n_{m}}+\sum_{m=0}^{\tilde{d}}\frac{P_V(\tilde{\varphi}_0,\tilde{n}_{m})}{1+\tilde{n}_{m}}}{{\sum_{m=0}^{\tilde{d}}\frac{T_V(\tilde{\varphi}_0,\tilde{n}_{m})}{1+\tilde{n}_{m}}-\sum_{m=0}^{d}\frac{T_V(\varphi_0,n_{m})}{1+n_{m}}}}(T(\varphi_0)+O(\delta)),\ d+\tilde{d}\to\infty.\label{l3.5middlefinal}
		\end{align}

		Since $\varphi_0-\frac{j_0}{q}\in \left(0,\frac{1}{q}\right)$, by \eqref{bqb} and \eqref{ghs0}, one has
		\begin{align}
		P\left(\varphi_0-\frac{j_0}{q}\right)-\frac{2\sin\left(\frac{\pi}{q}-2\pi\left(\varphi_0-\frac{j_0}{q}\right)\right)}{1+\cos\left(\frac{\pi}{q}-2\pi\left(\varphi_0-\frac{j_0}{q}\right)\right)}T\left(\varphi_0-\frac{j_0}{q}\right)=qB_q.\nonumber
		\end{align}
	    By \eqref{perioPvarphi} and \eqref{perioTvarphi}, one has $P(\varphi_0)=	P\left(\varphi_0-\frac{j_0}{q}\right), T(\varphi_0)=T\left(\varphi_0-\frac{j_0}{q}\right)$, then
	    \begin{align}
	    	P\left(\varphi_0\right)-\frac{2\sin\left(\frac{\pi}{q}-2\pi\left(\varphi_0-\frac{j_0}{q}\right)\right)}{1+\cos\left(\frac{\pi}{q}-2\pi\left(\varphi_0-\frac{j_0}{q}\right)\right)}T\left(\varphi_0\right)=qB_q.\label{l3.5ghs0}
	    \end{align}
	    Therefore, to prove \eqref{lemmaduichenfinal}, by \eqref{l3.5middlefinal}, \eqref{l3.5ghs0} and the fact $T(\varphi_0)>0$, we only need to show
	    		\begin{align}
	    		\inf_{\substack{\abs{a_n}\leq1\\ \text{subject to}\  \eqref{denominator>0}}}\frac{\sum_{m=0}^{d}\frac{P_V(\varphi_0,n_{m})}{1+n_{m}}+\sum_{m=0}^{\tilde{d}}\frac{P_V(\tilde{\varphi}_0,\tilde{n}_{m})}{1+\tilde{n}_{m}}}{{\sum_{m=0}^{\tilde{d}}\frac{T_V(\tilde{\varphi}_0,\tilde{n}_{m})}{1+\tilde{n}_{m}}-\sum_{m=0}^{d}\frac{T_V(\varphi_0,n_{m})}{1+n_{m}}}}\geq \frac{2\sin\left(\frac{\pi}{q}-2\pi\left(\varphi_0-\frac{j_0}{q}\right)\right)}{1+\cos\left(\frac{\pi}{q}-2\pi\left(\varphi_0-\frac{j_0}{q}\right)\right)}.\label{lem3.5final}
	    	\end{align}
Since $\varphi_0\in \left(\frac{j_0}{q},\frac{j_0}{q}+\frac{1}{2q}\right)$, and $\tilde{\varphi}_0=1-\varphi_0,$ we have that for any $j=0,1,\cdots,q-1$,
\begin{align}
	\sin \left(2\pi\varphi_0+\frac{2\pi}{q}j\right)\neq0,\ 	\sin \left(2\pi\tilde{\varphi}_0+\frac{2\pi}{q}j\right)\neq0.\nonumber
\end{align}
Denote by $J_0,\tilde{J}_0\subset\{0,\cdots,q-1\}$ so that 
\begin{align}
	\sin \left(2\pi\varphi_0+\frac{2\pi}{q}j\right)>0,\ \mathrm{for\ any}\ j\in J_0,\label{lem3.5sintheta0>0}\\
	\sin \left(2\pi\tilde{\varphi}_0+\frac{2\pi}{q}j\right)>0,\ \mathrm{for\ any}\ j\in \tilde{J}_0,\label{lem3.5sinthetatilde0>0}
\end{align}
and
\begin{align}
		\sin \left(2\pi\varphi_0+\frac{2\pi}{q}j\right)<0,\ \mathrm{for\ any}\ j\in J_0^c:= \{0,\cdots,q-1\}\setminus J_0,\label{lem3.5sintheta0<0}\\
	\sin \left(2\pi\tilde{\varphi}_0+\frac{2\pi}{q}j\right)<0,\ \mathrm{for\ any}\ j\in \tilde{J}_0^c:= \{0,\cdots,q-1\}\setminus\tilde{J}_0.\label{lem3.5sinthetatilde0<0}
\end{align}
By \eqref{definitionofP_qtheta} and \eqref{definitionofPVtheta}, for any $m$ one has
\begin{align}
	P_V(\varphi_0,n_m)=&\sum_{j=0}^{q-1}\abs{\sin\left(2\pi\varphi_0+\frac{2\pi}{q}j\right)}-\sum_{j=0}^{q-1}a_{n_m+l_j}\sin\left(2\pi\varphi_0+\frac{2\pi}{q}j\right)\nonumber\\
	=&\sum_{j\in J_0}(1-a_{n_m+l_j})\sin\left(2\pi\varphi_0+\frac{2\pi}{q}j\right)-\sum_{j\in J_0^c}(1+a_{n_m+l_j})\sin\left(2\pi\varphi_0+\frac{2\pi}{q}j\right)\label{PVexpansion}
\end{align}
and
\begin{align}
		P_V(\tilde{\varphi}_0,\tilde{n}_m)
	=\sum_{j\in \tilde{J}_0}(1-a_{\tilde{n}_m+l_j})\sin\left(2\pi\tilde{\varphi}_0+\frac{2\pi}{q}j\right)-\sum_{j\in \tilde{J}_0^c}(1+a_{\tilde{n}_m+l_j})\sin\left(2\pi\tilde{\varphi}_0+\frac{2\pi}{q}j\right).\label{PVtildeexpansion}
\end{align}
By $\abs{a_n}\leq 1$, one obtains that $1\pm a_{n_m+l_j}\geq 0$, $1\pm a_{\tilde{n}_m+l_j}\geq 0$. Therefore, \begin{align}
	P_V(\varphi_0,n_m)\geq 0,\ P_V(\tilde{\varphi}_0,\tilde{n}_m)\geq 0.\label{PVl0}
\end{align} 
Moreover, $P_V(\varphi_0,n_m)> 0$ if and only if there exists $j\in J_0$ such that $a_{n_m+l_j}<1$ or there exists $j\in J_0^c$ such that $a_{n_m+l_j}>-1$,  and $	P_V(\tilde{\varphi}_0,\tilde{n}_m)>0$ if and only if  there exists $j\in \tilde{J}_0$ such that $a_{\tilde{n}_m+l_j}<1$ or there exists $j\in \tilde{J}_0^c$ such that $a_{\tilde{n}_m+l_j}>-1$.

By \eqref{definitionofT_qtheta} and \eqref{definitionofTVtheta}, for any $m$ one has 
\begin{equation}
T_V(\varphi_0,n_m)
	=\sum_{j\in J_0}\left(a_{n_m+l_j}-1\right)\sin^2\pi\left(\varphi_0+\frac{j}{q}\right)+\sum_{j\in J_0^c}\left(1+a_{n_m+l_j}\right)\sin^2\pi\left(\varphi_0+\frac{j}{q}\right)\label{TVexpansion}
\end{equation}
and
\begin{align}
	T_V(\tilde{\varphi}_0,\tilde{n}_m)
	=\sum_{j\in \tilde{J}_0}(a_{\tilde{n}_m+l_j}-1)\sin^2\pi\left(\tilde{\varphi}_0+\frac{j}{q}\right)+\sum_{j\in \tilde{J}_0^c}(1+a_{\tilde{n}_m+l_j})\sin^2\pi\left(\tilde{\varphi}_0+\frac{j}{q}\right).\label{TVtildeexpansion}
\end{align}
Hence, we obtain
\begin{align}
	&\sum_{m=0}^{\tilde{d}}\frac{T_V(\tilde{\varphi}_0,\tilde{n}_{m})}{1+\tilde{n}_{m}}-\sum_{m=0}^{d}\frac{T_V(\varphi_0,n_{m})}{1+n_{m}}\nonumber\\
	=&\sum_{m=0}^{\tilde{d}}\frac{\sum_{j\in \tilde{J}_0^c}(1+a_{\tilde{n}_m+l_j})\sin^2\pi\left(\tilde{\varphi}_0+\frac{j}{q}\right)}{1+\tilde{n}_{m}}+\sum_{m=0}^{d}\frac{\sum_{j\in J_0}\left(1-a_{n_m+l_j}\right)\sin^2\pi\left(\varphi_0+\frac{j}{q}\right)}{1+n_{m}}\nonumber\\
	&-\sum_{m=0}^{\tilde{d}}\frac{\sum_{j\in \tilde{J}_0}(1-a_{\tilde{n}_m+l_j})\sin^2\pi\left(\tilde{\varphi}_0+\frac{j}{q}\right)}{1+\tilde{n}_{m}}-\sum_{m=0}^{d}\frac{\sum_{j\in J_0^c}\left(1+a_{n_m+l_j}\right)\sin^2\pi\left(\varphi_0+\frac{j}{q}\right)}{1+n_{m}}.\nonumber
\end{align}

Therefore, by \eqref{denominator>0}, $ 1\pm a_{n_m+l_j}\geq0$ and $1\pm a_{\tilde{n}_m+l_j}\geq0$, one has
\begin{align}
	&\sum_{m=0}^{\tilde{d}}\frac{\sum_{j\in \tilde{J}_0^c}(1+a_{\tilde{n}_m+l_j})\sin^2\pi\left(\tilde{\varphi}_0+\frac{j}{q}\right)}{1+\tilde{n}_{m}}+\sum_{m=0}^{d}\frac{\sum_{j\in J_0}\left(1-a_{n_m+l_j}\right)\sin^2\pi\left(\varphi_0+\frac{j}{q}\right)}{1+n_{m}}\nonumber\\
	&>\sum_{m=0}^{\tilde{d}}\frac{\sum_{j\in \tilde{J}_0}(1-a_{\tilde{n}_m+l_j})\sin^2\pi\left(\tilde{\varphi}_0+\frac{j}{q}\right)}{1+\tilde{n}_{m}}+\sum_{m=0}^{d}\frac{\sum_{j\in J_0^c}\left(1+a_{n_m+l_j}\right)\sin^2\pi\left(\varphi_0+\frac{j}{q}\right)}{1+n_{m}}\nonumber\\
	&\geq 0.\nonumber
\end{align}
Thus we can conclude that 
\begin{align}
	\mathrm{there\ exists\ }m\in\{0,1,\cdots,\tilde{d}\}\ &\mathrm{and}\ j\in \tilde{J}_0^c,\ \mathrm{such\ that}\ a_{\tilde{n}_m+l_j}>-1,\label{fis1}\\
	&\ \mathrm{or}\nonumber\\
	\mathrm{there\ exists\ }m\in\{0,1,\cdots,{d}\}\ &\mathrm{and}\ j\in {J}_0,\ \mathrm{such\ that}\ a_{{n}_m+l_j}<1.\label{fis2}
\end{align}
By \eqref{PVexpansion} and \eqref{PVtildeexpansion}, there exists $m\in\{0,1,\cdots,\tilde{d}\}$ such that $	P_V(\tilde{\varphi}_0,\tilde{n}_m)>0$, or there exists $m\in\{0,1,\cdots,{d}\}$ such that $P_V(\varphi_0,n_m)>0$. Therefore, by \eqref{PVl0},
\begin{align}
	\sum_{m=0}^{d}\frac{P_V(\varphi_0,n_{m})}{1+n_{m}}+\sum_{m=0}^{\tilde{d}}\frac{P_V(\tilde{\varphi}_0,\tilde{n}_{m})}{1+\tilde{n}_{m}}>0.\label{numerator>0}
\end{align}
By \eqref{denominator>0} and \eqref{numerator>0} we know that the numerator and denominator of 
\begin{align}
	\frac{\sum_{m=0}^{d}\frac{P_V(\varphi_0,n_{m})}{1+n_{m}}+\sum_{m=0}^{\tilde{d}}\frac{P_V(\tilde{\varphi}_0,\tilde{n}_{m})}{1+\tilde{n}_{m}}}{{\sum_{m=0}^{\tilde{d}}\frac{T_V(\tilde{\varphi}_0,\tilde{n}_{m})}{1+\tilde{n}_{m}}-\sum_{m=0}^{d}\frac{T_V(\varphi_0,n_{m})}{1+n_{m}}}}\nonumber
\end{align}
are both positive. To obtain the infimum of a fraction with positive numerator and denominator, we should let the denominator be large and the numerator be small as much as possible. Then by \eqref{PVexpansion} and \eqref{TVexpansion}, for any $m\in\{0,1,\cdots,d\}$, we should let  $a_{n_m+l_j}=-1$ for any $j\in J_0^c$. By \eqref{PVtildeexpansion} and \eqref{TVtildeexpansion},  for any $m\in\{0,1,\cdots,\tilde{d}\}$, we should let $a_{\tilde{n}_m+l_j}=1$ for  any $j\in \tilde{J}_0$. Therefore, by \eqref{PVexpansion}, \eqref{PVtildeexpansion}, \eqref{TVexpansion} and \eqref{TVtildeexpansion},
\begin{align}
	&\inf_{\substack{\abs{a_n}\leq1\\ \text{subject to}\  \eqref{denominator>0}}}\frac{\sum_{m=0}^{d}\frac{P_V(\varphi_0,n_{m})}{1+n_{m}}+\sum_{m=0}^{\tilde{d}}\frac{P_V(\tilde{\varphi}_0,\tilde{n}_{m})}{1+\tilde{n}_{m}}}{{\sum_{m=0}^{\tilde{d}}\frac{T_V(\tilde{\varphi}_0,\tilde{n}_{m})}{1+\tilde{n}_{m}}-\sum_{m=0}^{d}\frac{T_V(\varphi_0,n_{m})}{1+n_{m}}}}\label{lem3.5inf}\\
=&\inf_{\substack{\abs{a_n}\leq1\\ \text{subject to}\  \eqref{denominator>0}}}\frac{\sum_{m=0}^{d}\frac{\sum_{j\in J_0}(1-a_{n_m+l_j})\sin\left(2\pi\varphi_0+\frac{2\pi}{q}j\right)}{1+n_{m}}+\sum_{m=0}^{\tilde{d}}\frac{-\sum_{j\in \tilde{J}_0^c}(1+a_{\tilde{n}_m+l_j})\sin\left(2\pi\tilde{\varphi}_0+\frac{2\pi}{q}j\right)}{1+\tilde{n}_{m}}}{\sum_{m=0}^d\frac{\sum_{j\in J_0}\left(1-a_{n_m+l_j}\right)\sin^2\pi\left(\varphi_0+\frac{j}{q}\right)}{1+n_m} +\sum_{m=0}^{\tilde{d}}\frac{\sum_{j\in \tilde{J}_0^c}(1+a_{\tilde{n}_m+l_j})\sin^2\pi\left(\tilde{\varphi}_0+\frac{j}{q}\right)}{1+\tilde{n}_m}}.\nonumber
\end{align}
By \eqref{qiuhe}, \eqref{ficase} and \eqref{lem3.5sintheta0>0}, we can conclude that for any $j\in J_0$,
\begin{align}
	\frac{\sin\left(2\pi\varphi_0+\frac{2\pi}{q}j\right)}{\sin^2\pi\left(\varphi_0+\frac{j}{q}\right)}=	\frac{\sin\left(2\pi\left(\varphi_0-\frac{j_0}{q}\right)+\frac{2\pi}{q}(j+j_0)\right)}{\sin^2\pi\left(\varphi_0-\frac{j_0}{q}+\frac{j+j_0}{q}\right)}\geq \frac{2\sin\left(\frac{\pi}{q}-2\pi\left(\varphi_0-\frac{j_0}{q}\right)\right)}{1+\cos\left(\frac{\pi}{q}-2\pi\left(\varphi_0-\frac{j_0}{q}\right)\right)}.\nonumber
\end{align}
Hence, if there exists $m\in \{0,1,\cdots,d\}$ and $j\in J_0$ with $a_{n_m+l_j}<1$, then
\begin{align}
	\frac{\sum_{m=0}^{d}\frac{\sum_{j\in J_0}(1-a_{n_m+l_j})\sin\left(2\pi\varphi_0+\frac{2\pi}{q}j\right)}{1+n_{m}}}{\sum_{m=0}^d\frac{\sum_{j\in J_0}\left(1-a_{n_m+l_j}\right)\sin^2\pi\left(\varphi_0+\frac{j}{q}\right)}{1+n_m} }\geq \frac{2\sin\left(\frac{\pi}{q}-2\pi\left(\varphi_0-\frac{j_0}{q}\right)\right)}{1+\cos\left(\frac{\pi}{q}-2\pi\left(\varphi_0-\frac{j_0}{q}\right)\right)}.\label{lem3.5halflarge}
\end{align}
By \eqref{ficase},  $\tilde{\varphi}_0=1-\varphi_0\in\left(\frac{q-j_0-1}{q}+\frac{1}{2q},\frac{q-j_0}{q}\right)$. Then by \eqref{qiuhe2} and \eqref{lem3.5sinthetatilde0<0}, we can conclude that for any $j\in \tilde{J}_0^c$,
\begin{align}
	\frac{-\sin\left(2\pi\tilde{\varphi}_0+\frac{2\pi}{q}j\right)}{\sin^2\pi\left(\tilde{\varphi}_0+\frac{j}{q}\right)}=&	\frac{-\sin\left(2\pi\left(\tilde{\varphi}_0-\frac{q-j_0-1}{q}\right)+\frac{2\pi}{q}(j+q-j_0-1)\right)}{\sin^2\pi\left(\tilde{\varphi}_0-\frac{q-j_0-1}{q}+\frac{j+q-j_0-1}{q}\right)}\nonumber\\
	\geq& \frac{-2\sin\left(\frac{\pi}{q}-2\pi\left(\tilde{\varphi}_0-\frac{q-j_0-1}{q}\right)\right)}{1+\cos\left(\frac{\pi}{q}-2\pi\left(\tilde{\varphi}_0-\frac{q-j_0-1}{q}\right)\right)}.\nonumber
\end{align}
Hence, if there exists $m\in \{0,1,\cdots,\tilde{d}\}$ and $j\in \tilde{J}_0^c$ with $a_{\tilde{n}_m+l_j}>-1$, then
\begin{align}
	\frac{\sum_{m=0}^{\tilde{d}}\frac{-\sum_{j\in \tilde{J}_0^c}(1+a_{\tilde{n}_m+l_j})\sin\left(2\pi\tilde{\varphi}_0+\frac{2\pi}{q}j\right)}{1+\tilde{n}_{m}}}{\sum_{m=0}^{\tilde{d}}\frac{\sum_{j\in \tilde{J}_0^c}(1+a_{\tilde{n}_m+l_j})\sin^2\pi\left(\tilde{\varphi}_0+\frac{j}{q}\right)}{1+\tilde{n}_m}}\geq \frac{-2\sin\left(\frac{\pi}{q}-2\pi\left(\tilde{\varphi}_0-\frac{q-j_0-1}{q}\right)\right)}{1+\cos\left(\frac{\pi}{q}-2\pi\left(\tilde{\varphi}_0-\frac{q-j_0-1}{q}\right)\right)}.\label{lem3.5halfless}
\end{align}
By \eqref{ficase} and $\tilde{\varphi}_0=1-\varphi_0$,
\begin{align}
	\frac{\sin\left(\frac{\pi}{q}-2\pi\left(\varphi_0-\frac{j_0}{q}\right)\right)}{1+\cos\left(\frac{\pi}{q}-2\pi\left(\varphi_0-\frac{j_0}{q}\right)\right)}=\frac{-\sin\left(\frac{\pi}{q}-2\pi\left(\tilde{\varphi}_0-\frac{q-j_0-1}{q}\right)\right)}{1+\cos\left(\frac{\pi}{q}-2\pi\left(\tilde{\varphi}_0-\frac{q-j_0-1}{q}\right)\right)}.\label{lem3.5equal}
\end{align}
Combining with \eqref{fis1}  and \eqref{fis2}, by \eqref{lem3.5inf}-\eqref{lem3.5equal}, we can obtain \eqref{lem3.5final}.	\end{proof}
	
	\begin{remark}Conditions \eqref{theta0neighbour} and \eqref{thetatilde0neighbour}
	are mainly imposed to emphasize that $[\theta(n_m)]$ and $[\theta(\tilde{n}_m)]$ remain in the neighborhood of $\varphi_0$ and the neighborhood of $\tilde{\varphi}_0$, respectively. The factor $3\delta$ is chosen to be consistent with the following lemmas in Section \ref{secparpru}.
	In the proof, one can notice that one of the sets $\{n_{m}\}_{m=0}^{d}$, $\{\tilde{n}_{m}\}_{m=0}^{\tilde{d}}$ could be empty. For later reference, we state this result separately as Corollary \ref{zhongyaoyinli}.
\end{remark}

	\begin{corollary} \label{zhongyaoyinli}
		Let $a>0$ be fixed.	Suppose that $\theta(n)$ is determined by \eqref{ctt} with $\abs{V(n)}\leq\frac{a}{1+n}$. 
		Assume that $\varphi_0\in (0,1)\setminus \left\{\frac{j}{2q}\right\}_{j=0}^{2q}$ and 
	\begin{align}
	0<\delta\ll \min\left\{\abs{\varphi_0-\frac{j}{2q}}:j=0,1,\cdots,2q\right\}.\nonumber
\end{align} 
Let $\{n_m\}_{m=0}^d\subset\mathbb{N}$ satisfy
		\begin{align}
			\frac{1}{\delta}<n_0<n_1<n_2<\cdots<n_d,\nonumber
		\end{align}
		and be such that
		\begin{align}
			\label{theta0pmvarepsilon}\left[\theta(n_m)\right]\in (\varphi_0-3\delta,\varphi_0+3\delta),\ m=0,1,\cdots,d.
		\end{align}
		Then if
		\begin{align}
			\abs{\sum_{m=0}^d(\theta(n_m+q)-\theta(n_m)-kq)}\leq 10,\label{xydy1}
		\end{align}
		and 
		\begin{align}\label{dayuN}
			\sum_{m=0}^{d}\frac{1}{1+n_{m}}\geq \frac{1}{\delta},
		\end{align}		
		we have 
		\begin{align}\label{bqjia0}
			\frac{\sum_{m=0}^{d}\sum_{j=0}^{q-1}V(n_m+j)\sin2\pi\theta(n_m+j)}{\sum_{m=0}^{d}\frac{a}{1+n_m}}\leq qB_q+O(\delta) ,\ d\to\infty.
		\end{align}
	\end{corollary}

\section{Partitions of Slowly Varying Sequences}\label{sec4}
	Recall that $[s]$ and $(s,t)_{\mathrm{mod}}$ are defined by \eqref{demodx} and \eqref{demodint}, respectively.
Let $(x_n)_{n\in\mathbb{N}}$ be a real sequence with
\begin{align}
		x_{n+1}=x_n+\frac{O(1)}{1+n}.\label{xnbigo1}
	\end{align}
Denote the set of accumulation points of the sequence $(\left[x_n\right])_{n\in\mathbb{N}}$ by $A$. 

\begin{lemma}\label{casesxn}
	Suppose that the real sequence $(x_n)_{n\in\mathbb{N}}$
	 satisfies \eqref{xnbigo1}.
	Then for any $s,t\in A$ with $s\leq t$, we have either
	\begin{align}
		[s,t]\subset A,\nonumber
	\end{align}
	or
	\begin{align}
		[0,s]\cup[t,1]\subset A,\nonumber
	\end{align}
	where $[s,s]:=\{s\}$.
\end{lemma}
\begin{proof}
	Without loss of generality, we assume that $0<s<t<1.$
	
	Otherwise, there exist $y_1\in (s,t)$, and $y_2\in [0,s)\cup (t,1]$ such that $y_1\notin A, y_2\notin A$. Namely,  there exists some small $\varepsilon_0>0$ such that for any large $n$,
	\begin{align}
		\left[x_n\right]\notin (y_1-\varepsilon_0,y_1+\varepsilon_0),\label{y1lem41}
	\end{align}
	and 
	\begin{align}
		\left[x_n\right]\notin (y_2-\varepsilon_0,y_2+\varepsilon_0).\label{y2lem41}
	\end{align}
	Here we replace \eqref{y2lem41} with
	\begin{align}
		&\left[x_n\right]\notin (0,\varepsilon_0),&\mathrm{if}\ y_2=0,\nonumber\\
&\left[x_n\right]\notin (1-\varepsilon_0,1), &\mathrm{if}\ y_2=1.\nonumber
	\end{align}
	However, since $s\in A$, there exists some large enough $n_0$ with $\left[x_{n_0}\right]\in (s-\varepsilon_0,s+\varepsilon_0)$. Then by \eqref{xnbigo1}, \eqref{y1lem41} and \eqref{y2lem41}, for any $n>n_0$, $\left[x_n\right]$ will not approach $t$, which leads to a contradiction. 
\end{proof}

\begin{corollary}
		Suppose that the real sequence $(x_n)_{n\in\mathbb{N}}$
	satisfies \eqref{xnbigo1}, then there are four cases of $A$:
	
	\textbf{Case a.} There exists some $\varphi_0\in [0,1]$ such that $A=\{\varphi_0\}$ or $A=\{0,1\}$.

	\textbf{Case b.} 
$A=[b,c],\ 0\leq b<c\leq 1, b^2+(c-1)^2\neq 0.$
	
	\textbf{Case c.}
	$A=[0,d]\cup [e,1], \ 0\leq d<e\leq 1,d^2+(e-1)^2\neq0.$
	
		\textbf{Case d.}
	$A=[0,1].$
\end{corollary}
\begin{proof}
If the cardinality of $A$, $\# A=2$, then $A=\{0,1\}$. Otherwise, there exist $s,t\in A$ with $s<t$ and $s^2+(t-1)^2\neq 0$. By Lemma \ref{casesxn}, one has $[0,s]\cup [t,1]\subset A$ or $[s,t]\subset A$, which leads to $\#A=\infty$.

Assume that $\#A >2$ and $A\neq [0,1]$. Let $b=\inf A$ and $c=\sup A$.
If $b^2+(c-1)^2\neq 0$, then by Lemma \ref{casesxn}, $A=[b,c]$.
 If $b=0,c=1$, since $A\neq [0,1]$, there exists $y\in (0,1)$ such that $y\notin A$. Let $d=\sup\{x\in A:x<y\}$ and $e=\inf\{x\in A:x>y\}$, where $d^2+(e-1)^2\neq 0$ because of $\#A>2$. Then by Lemma \ref{casesxn}, $A=[0,d]\cup[e,1]$.
\end{proof}

We prove the following result to show that, in some sense, \textbf{Case b} and \textbf{Case c} are equivalent.

\begin{lemma}\label{lemequcd}
		Assume that the real sequence $(x_n)_{n\in\mathbb{N}}$ satisfies \eqref{xnbigo1}. If $A=[0,d]\cup [e,1], \ 0\leq d<e\leq 1,d^2+(e-1)^2\neq0,$ then the set of accumulation points of $([y_n])_{n\in\mathbb{N}}$ is $A'=[e+t-1,d+t]$, where $y_n=x_n+t, n\geq 1$ and $t\in (1-e,1-d)$ is fixed.
		If $A=[b,c],\ 0\leq b<c\leq 1,b^2+(c-1)^2\neq 0$, then the set of accumulation points of $([y_n])_{n\in\mathbb{N}}$ is $A'=[0,c-s]\cup [b-s+1,1]$, where $y_n=x_n-s, n\geq 1$ and $s\in (b,c)$ is fixed.
\end{lemma}
\begin{proof}
	Suppose $A=[0,d]\cup [e,1]$, then for any positive $\varepsilon\ll \min\{1-(d+t),e+t-1\}$, there exists $N\in\mathbb{N}$ such that for any $n\geq N$,
	\begin{align}
		[x_n]\in [0,d+\varepsilon)\cup (e-\varepsilon,1).\nonumber
	\end{align}
	Therefore, for any $n\geq N$,
	\begin{align}
		[y_n]=[x_n+t]\in [t,d+t+\varepsilon)\cup(e+t-1-\varepsilon,t)=(e+t-1-\varepsilon,d+t+\varepsilon).\nonumber
	\end{align}
	Then $A'\subset [e+t-1,d+t]$.
	On the other hand, assume $y\in (e+t-1,d+t)$. If $y\in (t,d+t)$, namely, $y-t\in (0,d)\subset A$. Then for any small $\varepsilon>0$, there exists $m\in\mathbb{N}$ such that
	\begin{align}
		[x_m]\in (y-t-\varepsilon,y-t+\varepsilon).\nonumber
	\end{align}
	Therefore,
	\begin{align}
		[y_m]=[x_m+t]\in (y-\varepsilon,y+\varepsilon),\nonumber
	\end{align}
	from which we have $y\in A'$. Namely, $(t,d+t)\subset A'$. Similarly, $(e+t-1,t]\subset A'$.
  Since $A'$ is closed, we can obtain the desired result.
	
	For the second assertion, suppose $A=[b,c]$. Without loss of generality, assume $0\leq b<c<1$, then for any positive $\varepsilon\ll\min\{s-b,c-s\}$, there exists $N\in\mathbb{N}$ such that for any $n\geq N$,
	\begin{align}
		&[x_n]\in [0,c+\varepsilon),\ &\mathrm{if}\ b=0,\nonumber\\
		&[x_n]\in (b-\varepsilon,c+\varepsilon),\ &\mathrm{if}\ b>0.\nonumber
	\end{align}
	Therefore, for any $n\geq N$,
	\begin{align}
		&[y_n]=[x_n-s]\in [0,c-s+\varepsilon)\cup [1-s,1],\ &\mathrm{if}\ b=0,\nonumber\\
		&[y_n]=[x_n-s]\in [0,c-s+\varepsilon)\cup (b-s+1-\varepsilon,1],\ &\mathrm{if}\ b>0,\nonumber
	\end{align}
which leads to $A'\subset [0,c-s]\cup [b-s+1,1]$. On the other hand, assume $y\in (0,c-s)\cup (b-s+1,1)$. If $y\in (0,c-s)$, then $y+s\in (s,c)\subset A$. Therefore, for any small $\varepsilon>0$, there exists $m\in \mathbb{N}$ such that
\begin{align}
	[x_m]\in (y+s-\varepsilon,y+s+\varepsilon).\nonumber
\end{align}
Thus
\begin{align}
	[y_m]=[x_m-s]\in (y-\varepsilon,y+\varepsilon).\nonumber
\end{align}
This leads to $(0,c-s)\subset A'$. In a similar way, one can show $(b-s+1,1)\subset A'$. Combined with the fact that $A'$ is closed, one can obtain the desired result.
\end{proof}

Recall that $\R/\Z$ is homeomorphic to the unit circle $S^1$ via the homeomorphism $[x]\to e^{2\pi i x}$, we do not distinguish between elements in the two spaces and denote them uniformly by $[x]$. From this point of view, one can see more clearly that \textbf{Case b} and \textbf{Case c} are equivalent.
\begin{figure}[H]
	\centering

	\tikzset{every picture/.style={line width=0.75pt}} 
	
	\begin{tikzpicture}[x=0.75pt,y=0.75pt,yscale=-1,xscale=1]
		
		\draw   (136,133.5) .. controls (136,108.37) and (156.37,88) .. (181.5,88) .. controls (206.63,88) and (227,108.37) .. (227,133.5) .. controls (227,158.63) and (206.63,179) .. (181.5,179) .. controls (156.37,179) and (136,158.63) .. (136,133.5) -- cycle ;
		\draw  [fill={rgb, 255:red, 208; green, 2; blue, 27 }  ,fill opacity=1 ] (133,132) .. controls (133,130.62) and (134.12,129.5) .. (135.5,129.5) .. controls (136.88,129.5) and (138,130.62) .. (138,132) .. controls (138,133.38) and (136.88,134.5) .. (135.5,134.5) .. controls (134.12,134.5) and (133,133.38) .. (133,132) -- cycle ;
		\draw   (326,130.5) .. controls (326,105.37) and (346.37,85) .. (371.5,85) .. controls (396.63,85) and (417,105.37) .. (417,130.5) .. controls (417,155.63) and (396.63,176) .. (371.5,176) .. controls (346.37,176) and (326,155.63) .. (326,130.5) -- cycle ;
		\draw  [fill={rgb, 255:red, 208; green, 2; blue, 27 }  ,fill opacity=1 ] (414,129) .. controls (414,127.62) and (415.12,126.5) .. (416.5,126.5) .. controls (417.88,126.5) and (419,127.62) .. (419,129) .. controls (419,130.38) and (417.88,131.5) .. (416.5,131.5) .. controls (415.12,131.5) and (414,130.38) .. (414,129) -- cycle ;
		
		\draw (110,205) node [anchor=north west][inner sep=0.75pt]   [align=left] {$A=\{\varphi_0\}$\ with $\varphi_0\in [0,1]$};
		\draw (335,205) node [anchor=north west][inner sep=0.75pt]   [align=left] {$A=\{0,1\}$};
		\draw (231,115) node [anchor=north west][inner sep=0.75pt]   [align=left] {$[0]$};
		\draw (424,115) node [anchor=north west][inner sep=0.75pt]   [align=left] {$[0]$};
		\draw (208,75) node [anchor=north west][inner sep=0.75pt]   [align=left] {$S^1$};
		\draw (409,78) node [anchor=north west][inner sep=0.75pt]   [align=left] {$S^1$};
		\draw (100,120) node [anchor=north west][inner sep=0.75pt]   [align=left] {$[\varphi_0]$};
		
	\end{tikzpicture}
	\caption{}
	\label{fcasea}
\end{figure}

\begin{figure}[H]
	\centering

\tikzset{every picture/.style={line width=0.75pt}} 

\begin{tikzpicture}[x=0.75pt,y=0.75pt,yscale=-1,xscale=1]
	
	\draw   (264,98.5) .. controls (264,73.37) and (284.37,53) .. (309.5,53) .. controls (334.63,53) and (355,73.37) .. (355,98.5) .. controls (355,123.63) and (334.63,144) .. (309.5,144) .. controls (284.37,144) and (264,123.63) .. (264,98.5) -- cycle ;
	\draw   (117,101.5) .. controls (117,76.37) and (137.37,56) .. (162.5,56) .. controls (187.63,56) and (208,76.37) .. (208,101.5) .. controls (208,126.63) and (187.63,147) .. (162.5,147) .. controls (137.37,147) and (117,126.63) .. (117,101.5) -- cycle ;
	\draw [color={rgb, 255:red, 208; green, 2; blue, 27 }  ,draw opacity=1 ][line width=3]    (138,140.5) .. controls (108,122.5) and (114,74.5) .. (137,63.5) ;
	\draw [color={rgb, 255:red, 208; green, 2; blue, 27 }  ,draw opacity=1 ][line width=3]    (333,137) .. controls (355,124) and (370,87) .. (330,58) ;
	\draw   (420,97.5) .. controls (420,72.37) and (440.37,52) .. (465.5,52) .. controls (490.63,52) and (511,72.37) .. (511,97.5) .. controls (511,122.63) and (490.63,143) .. (465.5,143) .. controls (440.37,143) and (420,122.63) .. (420,97.5) -- cycle ;
	\draw  [color={rgb, 255:red, 208; green, 2; blue, 27 }  ,draw opacity=1 ][line width=2.25]  (420,97.5) .. controls (420,72.37) and (440.37,52) .. (465.5,52) .. controls (490.61,54) and (511,72.37) .. (511,97.5) .. controls (511,122.63) and (490.63,143) .. (465.5,143) .. controls (440.37,143) and (420,122.63) .. (420,97.5) -- cycle ;
	
	\draw (130,180) node [anchor=north west][inner sep=0.75pt]   [align=left] {$A=[b,c]$};
	\draw (260,180) node [anchor=north west][inner sep=0.75pt]   [align=left] {$A=[0,d]\cup [e,1]$};
	\draw (212,88) node [anchor=north west][inner sep=0.75pt]   [align=left] {$[0]$};
	\draw (361,88) node [anchor=north west][inner sep=0.75pt]   [align=left] {$[0]$};
	\draw (199,50) node [anchor=north west][inner sep=0.75pt]   [align=left] {$S^1$};
	\draw (351,50) node [anchor=north west][inner sep=0.75pt]   [align=left] {$S^1$};
	\draw (435,180) node [anchor=north west][inner sep=0.75pt]   [align=left] {$A=[0,1]$};
	\draw (524,88) node [anchor=north west][inner sep=0.75pt]   [align=left] {$[0]$};
	\draw (518,50) node [anchor=north west][inner sep=0.75pt]   [align=left] {$S^1$};
		\draw (125,35) node [anchor=north west][inner sep=0.75pt]   [align=left] {$[c]$};
		\draw (318,35) node [anchor=north west][inner sep=0.75pt]   [align=left] {$[e]$};
		\draw (125,145) node [anchor=north west][inner sep=0.75pt]   [align=left] {$[b]$};
		\draw (318,145) node [anchor=north west][inner sep=0.75pt]   [align=left] {$[d]$};
	
\end{tikzpicture}
	\caption{}
	\label{fcasebc}
\end{figure}

\begin{proposition}\label{propconfi}
In \textbf{Cases a}, \textbf{b} and \textbf{c},  $(x_n)_{n\geq N}$
is confined to an interval with length less than $1$ for large $N$.
\end{proposition}
\begin{proof}
	If $A=\{\varphi_0\}$ for some $\varphi_0\in (0,1)$, by the definition, there exist  $0<\varepsilon\ll \min\{\varphi_0,1-\varphi_0\}$ and $N\in\mathbb{N}$, such that for any $n\geq N$,
	\begin{align}
		[x_n]\in (\varphi_0-\varepsilon,\varphi_0+\varepsilon).\nonumber
	\end{align}
Combined with \eqref{xnbigo1}, by letting $N$ be large enough, there exists some $m\in\Z$ such that for any $n\geq N$,
\begin{align}
	x_n\in (\varphi_0+m-\varepsilon,\varphi_0+m+\varepsilon),\nonumber
\end{align}
which leads to the result. If $A=\{0\}$ or $A=\{1\}$ or $A=\{0,1\}$, the proofs are similar and omitted.
By Lemma \ref{lemequcd},  \textbf{Case b} and \textbf{Case c} are essentially the same, therefore we only need to prove the result for \textbf{Case b}, in which $A=[b,c],0\leq b<c\leq 1$ with $b^2+(c-1)^2\neq 0$. Without loss of generality, we assume $0\leq b<c<1$. Then by the definition, there exists positive $\varepsilon\ll \min \{b,1-c\}$ if $b>0$ or $\varepsilon\ll 1-c$ if $b=0$, and $N\in\mathbb{N}$, such that for any $n\geq N$,
\begin{align}
	[x_n]&\in (b-\varepsilon,c+\varepsilon),\ &\mathrm{if}\ b>0,\nonumber\\
	[x_n]&\in [0,c+\varepsilon),\ &\mathrm{if}\ b=0\nonumber.
\end{align}
Combined with \eqref{xnbigo1}, by letting $N$ be large enough, there exists some $m\in\Z$ such that for any $n\geq N$,
\begin{align}
	x_n&\in (b+m-\varepsilon,c+m+\varepsilon),\ &\mathrm{if}\ b>0,\nonumber\\
	x_n&\in [m,c+m+\varepsilon),\ &\mathrm{if}\ b=0,\nonumber
\end{align}
which leads to the result.
\end{proof}

\begin{theorem}[\textbf{Case a}]\label{casea}
Assume that the real sequence $(x_n)_{n\in\mathbb{N}}$ satisfies \eqref{xnbigo1} with $A=\{\varphi_0\},\varphi_0\neq 0,1$. Let $\delta>0$ be small.  Then there exists large $n_\delta\in\mathbb{N}$ with $n_\delta\gg \frac{1}{\delta}$, such that 
for any $n\geq n_\delta$, one has
\begin{align}
	\left[x_n\right]\in (\varphi_0-3\delta,\varphi_0+3\delta),\label{caseaaiinuitilde}
\end{align}
and for any  $N\in\mathbb{N}$ large enough,
\begin{align}
	\abs{\sum_{n=n_\delta}^{n_\delta+N}(x_{n+1}-x_n)}\leq 10,\label{caseaaiminusbileq10}
\end{align}
\begin{align}
\sum_{n=n_\delta}^{n_\delta+N}\frac{1}{1+n}\gg \frac{1}{\delta}.\label{caseaIifractionlarge}
\end{align}
\end{theorem}
\begin{proof}
The existence of $n_\delta$ for \eqref{caseaaiinuitilde} is obvious. Let $n_\delta$ be large enough, by \eqref{xnbigo1} and \eqref{caseaaiinuitilde}, there exists $m\in \mathbb{Z}$ such that for any $n\geq n_\delta$, 
\begin{align}
x_n\in (\varphi_0+m-3\delta,\varphi_0+m+3\delta).\nonumber
\end{align}
Then by
$\sum_{n=n_\delta}^{n_\delta+N}(x_{n+1}-x_n)=x_{n_\delta+N+1}-x_{n_\delta}$ one can obtain \eqref{caseaaiminusbileq10}. Let $N$ be large, we have \eqref{caseaIifractionlarge}.
\end{proof}

Let us consider the latter three cases. Since the proofs of \textbf{Case b} and \textbf{Case d} are extremely hard and technical, we state only the results here and provide the proofs separately in the next two sections.

We first consider \textbf{Case b}. Recall
\begin{align}
	A=[b,c],\ 0\leq b<c\leq 1, b^2+(c-1)^2\neq 0.\nonumber
\end{align}
For any small $\delta>0$, let $K_1=K_1(b,c,\delta):=\left \lfloor \frac{c-b}{\delta} \right \rfloor$, where $\left \lfloor \cdot \right \rfloor$ is the floor function. For any $i=1,\cdots,K_1-1$, let
\begin{align}
	\alpha_i=\alpha_i(\delta)=c-i\delta,\label{dealphai}
\end{align}
and
\begin{align}
	Q_i(\delta)=(\alpha_i-3\delta,\alpha_i+3\delta)_{\mathrm{mod}}.\label{deofVi}
\end{align}

\begin{theorem}[\textbf{Case b}]\label{caseb}
	Assume that the real sequence $(x_n)_{n\in\mathbb{N}}$ satisfies \eqref{xnbigo1} with $A=[b,c],0\leq b<c\leq 1, b^2+(c-1)^2\neq 0$. Let $\delta>0$ be small and $K_1=\left \lfloor \frac{c-b}{\delta} \right \rfloor$. Then there exists large $n_\delta\in\mathbb{N}$ with $n_\delta\gg \frac{1}{\delta}$, such that for any $N\in\mathbb{N}$ large enough, we can divide $I:=\{n\}_{n=n_\delta}^{n_\delta+N}$ into $(K_1-1)$ disjoint subsets $I_i, i=1,2,\cdots,K_1-1,$ and for any $i=1,2,\cdots,K_1-1$, one has
	\begin{align}
		{[x_n]}&\in Q_i(\delta),\ \mathrm{for\ any}\ n\in I_i,\label{aiinvi}\\
		&\abs{\sum_{n\in I_i}(x_{n+1}-x_n)}\leq 10,\label{Iileq10}
	\end{align}
	and 
	\begin{align}
		\sum_{n\in I_i}\frac{1}{1+n}\gg \frac{1}{\delta}.\label{Iifractionlargecasec}
	\end{align}
\end{theorem}
\begin{proof}
	See Section \ref{pcb}.
\end{proof}

Let us next consider \textbf{Case c}, in which 
$$A=[0,d]\cup [e,1], \ 0\leq d<e\leq 1,d^2+(e-1)^2\neq0.$$ 
For any small $\delta>0$, let $K_2=K_2(d,e,\delta):=\left \lfloor \frac{1+d-e}{\delta} \right \rfloor$. 
For any $i=1,\cdots,K_2-1$, let
\begin{align}
	\beta_i=\beta_i(\delta)=d-i\delta,
\end{align}
and
\begin{align}
	W_i(\delta)=(\beta_i-3\delta,\beta_i+3\delta)_{\mathrm{mod}}.\label{deofWi}
\end{align}
\begin{theorem}[\textbf{Case c}]\label{casec}
	Assume that the real sequence $(x_n)_{n\in\mathbb{N}}$ satisfies \eqref{xnbigo1} with $A=[0,d]\cup [e,1], \ 0\leq d<e\leq 1,d^2+(e-1)^2\neq0.$ Let $\delta>0$ be small and $K_2=\left \lfloor \frac{1+d-e}{\delta} \right \rfloor$. Then there exists large $n_\delta\in\mathbb{N}$ with $n_\delta\gg \frac{1}{\delta}$, such that for any $N\in\mathbb{N}$ large enough, we can divide $I:=\{n\}_{n=n_\delta}^{n_\delta+N}$ into $(K_2-1)$ disjoint subsets $I_i, i=1,2,\cdots,K_2-1,$  and for any $i=1,2,\cdots,K_2-1$, one has
	\begin{align}
		[x_n]&\in W_i(\delta),\ \mathrm{for\ any}\ n\in I_i,\label{aiinwi}\\
		&\abs{\sum_{n\in I_i}(x_{n+1}-x_n)}\leq 10,\label{Iileq10cased}
	\end{align}
	and 
	\begin{align}
		\sum_{n\in I_i}\frac{1}{1+n}\gg \frac{1}{\delta}.\label{Iifractionlargecased}
	\end{align}
\end{theorem}
\begin{proof}
	Let $t\in (1-e,1-d)$ be fixed and $y_n=x_n+t$. By Lemma \ref{lemequcd}, the set of accumulation points of $([y_n])_{n\in\mathbb{N}}$ is $[e+t-1,d+t]$. By Theorem \ref{caseb}, we can divide $I:=\{n\}_{n=n_\delta}^{n_\delta+N}$ into $(K_2-1)$ disjoint subsets $I_i, i=1,2,\cdots,K_2-1,$  and for any $i=1,2,\cdots,K_2-1$,
	\begin{align}
		[y_n]\in (d+t-&i\delta-3\delta,d+t-i\delta+3\delta)_{\mathrm{mod}},\ \mathrm{for\ any}\ n\in I_i,\label{fiynind}\\
		&\abs{\sum_{n\in I_i}(y_{n+1}-y_n)}\leq 10,\nonumber
	\end{align}
	and 
	\begin{align}
		\sum_{n\in I_i}\frac{1}{1+n}\gg \frac{1}{\delta}.\nonumber
	\end{align}
	Therefore, we have \eqref{Iileq10cased} and \eqref{Iifractionlargecased}. 
	By \eqref{equmodr} and \eqref{fiynind}, for each $n\in I_i$, there exists $m_n\in\Z$ such that
	\begin{align}
		x_n+t+m_n\in(d+t-i\delta-3\delta,d+t-i\delta+3\delta).\nonumber
	\end{align}
	Therefore, $[x_n]+x_n-[x_n]+m_n\in (d-i\delta-3\delta,d-i\delta+3\delta)$. By \eqref{equstmod}, one has $[x_n]\in W_i(\delta)$.
\end{proof}
We finally consider \textbf{Case d}. Namely, $A=[0,1]$. For any small $\delta>0$, let $K_3=K_3(\delta):=\left \lfloor \frac{1}{2\delta} \right \rfloor$. For any $i=1,\cdots,K_3-1$, let
\begin{align}
	\varphi_i=\varphi_i(\delta)=\frac{1}{2}-i\delta, \ \tilde{\varphi}_i=\tilde{\varphi}_i(\delta)=1-\varphi_i(\delta)=\frac{1}{2}+i\delta,\label{definitionofthetai}
\end{align}
and
\begin{align}
	U_i(\delta)=(\varphi_i-3\delta,\varphi_i+3\delta)_{\mathrm{mod}},\ \tilde{U}_i(\delta)=(\tilde{\varphi}_i-3\delta,\tilde{\varphi}_i+3\delta)_{\mathrm{mod}}.\label{definitionofui}
\end{align}
\begin{theorem}[\textbf{Case d}]\label{cased}
 Assume that the real sequence $(x_n)_{n\in\mathbb{N}}$ satisfies \eqref{xnbigo1} with $A=[0,1]$. Let $\delta>0$ be small and $K_3=\left \lfloor \frac{1}{2\delta} \right \rfloor$. Then there exists large $n_\delta\in\mathbb{N}$ with $n_\delta\gg \frac{1}{\delta}$, such that for any $N\in\mathbb{N}$ large enough, we can divide $I:=\{n\}_{n=n_\delta}^{n_\delta+N}$ into $2(K_3-1)$ disjoint subsets: $A_i, \tilde{{A}}_i,\ i=1,2,\cdots,K_3-1,$ 
and
\begin{enumerate}
\item For any $i=1,\cdots,K_3-1,$ one has
\begin{align}
	\left[x_n\right]\in U_i(\delta),\ \mathrm{for\ any\ } n\in A_i,\label{aiinui}
\end{align}
and
\begin{align}
	\left[x_n\right]\in \tilde{U}_i(\delta),\ \mathrm{for\ any\ } n\in \tilde{A}_i.\label{aiinuitilde}
\end{align}
\item For any $i=2,\cdots,K_3-2$, one has
\begin{align}
	\abs{\sum_{n\in A_i}(x_{n+1}-x_n)-\sum_{n\in \tilde{A}_i}(x_{n+1}-x_n)}\leq 10,\label{aiminusbileq10}
\end{align}
and 
\begin{align}
\sum_{n\in I_i}\frac{1}{1+n}\gg \frac{1}{\delta},\label{Iifractionlarge}
\end{align}
where $I_i=A_i\cup \tilde{{A}}_i.$ 
\end{enumerate}
\end{theorem}
\begin{proof}
	See Section \ref{pcd}.
\end{proof}

To better understand the above theorems, we present a similar result to Theorem \ref{caseb} in the continuous setting.

\begin{theorem}[Continuous analogue of Theorem \ref{caseb}]\label{thmcontinuous} Assume that the real function $f(\cdot)\in C^1(\R)$ with
	\begin{align}
		f'(t)=\frac{O(1)}{1+\abs{t}},\ t\to\infty,\label{apfprO}
	\end{align}
	 and $0\leq b<c\leq 1,b^2+(c-1)^2\neq 0$, where
	\begin{align}
		b=\liminf_{t\to\infty}f(t),\ c=\limsup_{t\to\infty}f(t).\nonumber
	\end{align} Let $\delta>0$ be small and $K_1=\left \lfloor \frac{c-b}{\delta} \right \rfloor$. Then there exists $t_\delta\gg \frac{1}{\delta}$, such that for any large enough $M$, we can divide $[t_{\delta},t_{\delta}+M]$ into $(K_1-1)$ disjoint subsets $I_i, i=1,2,\cdots,K_1-1,$ and for any $i=1,2,\cdots,K_1-1$, one has
	\begin{align}
		f(t)&\in Q_i(\delta),\ \mathrm{for\ any}\ t\in I_i,\label{appeniinvia}\\
		&\abs{\int_{t\in I_i}f'(t)dt}\leq 10,\label{appenee}
	\end{align}
	and 
	\begin{align}
	\int_{t\in I_i}\frac{1}{1+t}dt\gg \frac{1}{\delta},\label{appenIifractionlargecasec}
	\end{align}
	where $Q_i(\delta)$ is defined by \eqref{deofVi}.
\end{theorem}
\begin{proof}
To avoid repetition, we include a proof in Appendix \ref{appena}.
\end{proof}

\begin{remark}
	Although the discrete setting involves extremely hard analysis and technical difficulties, for some special cases, the corresponding results in the continuous setting are more transparent and easier to interpret. 
	For example, let $$f(t)=\frac{c}{2}(1+\sin (\ln t)),\ \ \ \  t>0,\ \ \ 0<c<1.$$
	Then 
	\begin{align}
		\liminf_{t\to\infty}f(t)=0,\ \ \ \ \ \limsup_{t\to\infty}f(t)=c.\nonumber
	\end{align}
	
	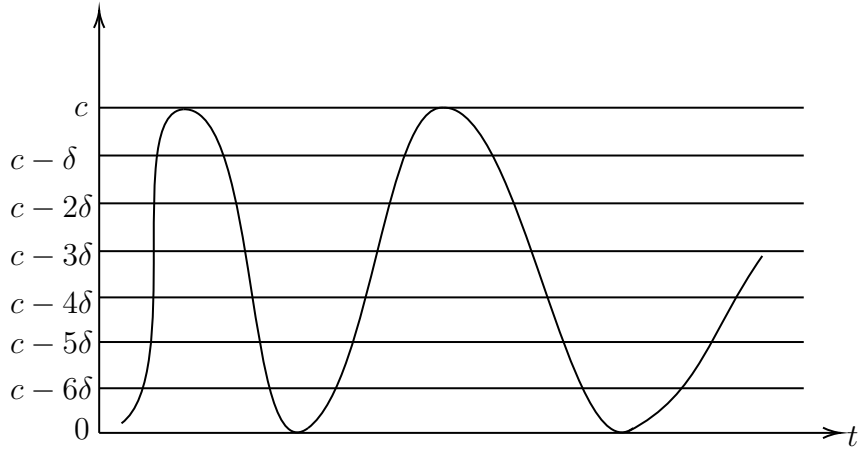
\begin{figure}[H]

		\tikzset{every picture/.style={line width=0.75pt}} 
		
		\begin{tikzpicture}[x=0.75pt,y=0.75pt,yscale=-0.8,xscale=0.8]
			
			\draw    (117.67,316) -- (580.67,316) ;
			\draw [shift={(582.67,316)}, rotate = 180] [color={rgb, 255:red, 0; green, 0; blue, 0 }  ][line width=0.75]    (10.93,-3.29) .. controls (6.95,-1.4) and (3.31,-0.3) .. (0,0) .. controls (3.31,0.3) and (6.95,1.4) .. (10.93,3.29)   ;
			\draw    (117.67,316) -- (117.67,53) ;
			\draw [shift={(117.67,51)}, rotate = 90] [color={rgb, 255:red, 0; green, 0; blue, 0 }  ][line width=0.75]    (10.93,-3.29) .. controls (6.95,-1.4) and (3.31,-0.3) .. (0,0) .. controls (3.31,0.3) and (6.95,1.4) .. (10.93,3.29)   ;
			\draw    (117.67,112) -- (559.67,112) ;
			\draw    (117.67,142) -- (559.67,142) ;
			\draw    (117.67,172) -- (559.67,172) ;
			\draw    (117.67,202) -- (559.67,202) ;
			\draw    (117.67,231) -- (559.67,231) ;
			\draw    (117.67,259) -- (559.67,259) ;
			\draw    (117.67,288) -- (559.67,288) ;
		
		\draw    (131.67,310) .. controls (171.67,280) and (132.67,115) .. (170.67,113) ;
		\draw    (170.67,113) .. controls (218.67,111) and (208.67,343) .. (248.67,313) ;
		\draw    (248.67,313) .. controls (288.67,283) and (297.67,109) .. (333.67,112) ;
		\draw    (333.67,112) .. controls (381.67,110) and (412.67,343) .. (452.67,313) ;
		\draw    (449.67,315) .. controls (494.67,292) and (500.67,252) .. (533.67,205) ;

			\draw (100,107) node [anchor=north west][inner sep=0.75pt]   [align=left] {$c$};
			\draw (60,135) node [anchor=north west][inner sep=0.75pt]   [align=left] {$c-\delta$};
			\draw (60,165) node [anchor=north west][inner sep=0.75pt]   [align=left] {$c-2\delta$};
			\draw (60,195) node [anchor=north west][inner sep=0.75pt]   [align=left] {$c-3\delta$};
			\draw (60,225) node [anchor=north west][inner sep=0.75pt]   [align=left] {$c-4\delta$};
			\draw (60,250) node [anchor=north west][inner sep=0.75pt]   [align=left] {$c-5\delta$};
			\draw (60,279) node [anchor=north west][inner sep=0.75pt]   [align=left] {$c-6\delta$};
			\draw (100,305) node [anchor=north west][inner sep=0.75pt]   [align=left] {$0$};
			\draw (585,310) node [anchor=north west][inner sep=0.75pt]   [align=left] {$t$};
		\end{tikzpicture}
		\caption{$f(t)=\frac{c}{2}(1+\sin (\ln t))$}
		\label{figcontinuousbands}
	\end{figure}
Figure \ref{figcontinuousbands} illustrates the level band partition used below.
	Let $t_\delta$ be large with $\sin(\ln t_{\delta})=-1$, so that $f(t_\delta)=0$. Let
	\[
	I_i=\{t_{\delta}\leq t\leq t_{\delta}+M:c-i\delta<f(t)\leq c-(i-1)\delta\},
	\quad i=1,2,\ldots,K_1-2,
	\]
	and
	\[
	I_{K_1-1}=\{t_{\delta}\leq t\leq t_{\delta}+M:0\leq f(t)\leq c-(K_1-2)\delta\},
	\]
	so that these sets form a partition of $[t_\delta,t_\delta+M]$. Since $\int_{x}^{y}f'(t)dt=f(y)-f(x)$, on each $I_i$, the contributions of complete upward and downward traversals of the same level band cancel; only the two boundary traversals can remain. Hence \eqref{appenee} follows.
	
	Since 
	\begin{align}
		\abs{f'(t)}=\abs{\frac{c\cos(\ln t)}{2t}}\leq \frac{c}{2t},\ t>0,\nonumber
	\end{align}
	we know that each $I_i$ consists of intervals $I$ with
	\begin{align}
		\abs{\int_I f'(t)dt}=\delta\leq \int_I\frac{c}{2t}dt.\nonumber
	\end{align}
	Then  by letting $M$ be large enough, each $I_i$ contains enough such intervals $I$, from which we can prove \eqref{appenIifractionlargecasec}.
\end{remark}

\section{Partitions in the Proper-Arc Case}\label{pcb}
\begin{proof}[\textbf{Proof of Theorem \ref{caseb}}] 
Since $c-b<1$,	by an appropriate rotation of the interval 
$[b,c]$ on $S^1$, we can always assume $0<b<c<1$. 

Let $0<\delta\ll\min\{b,1-c\}$. By Proposition \ref{propconfi}, there exists  $m\in \Z$ such that the sequence $(x_n)_{n\in\mathbb{N}}$ will eventually oscillate in the interval $(b-\delta+m,c+\delta+m)$. Without loss of generality, we assume that $m=0$ and
\begin{align}
	x_n\in (b-\delta,c+\delta),\ \mathrm{for\ any\ }n\in\mathbb{N}.\label{xninbcmpe}
\end{align}
Thus for any $n\in\mathbb{N}$, $$
	x_n=[x_n].$$
	Define
	\begin{align}
		&\mathcal{I}_1=\left\{n\in I:x_n\in[\alpha_1,c+\delta)\right\},\label{cbdmi1}\\
		\mathcal{I}_i=\{n&\in I:x_n\in[\alpha_i,\alpha_{i-1}]\},\ i=2,\cdots,K_1-1,\label{cbdmii}
	\end{align}
	where $\alpha_i$ is defined by \eqref{dealphai}. For any $i=1,2,\cdots,K_1-1$, define
	\begin{align}
		\frac{1}{2}\mathcal{I}_i=\left\{n\in I:x_n\in\left[\alpha_{i}+\frac{\delta}{2},\alpha_{i-1}\right] \right\},\label{cbdhalfmii}
	\end{align}
	where $\alpha_0:=c$.
	Let $n_\delta\gg \frac{1}{\delta}$ and by \eqref{xnbigo1} we have
	\begin{align}
		\abs{x_{n+1}-x_n}\ll \delta,\ \mathrm{for\ any}\ n\geq n_\delta.\label{cbxn1xnlle}
	\end{align}
	Since $A=[b,c]$ and $\alpha_1=c-\delta$, by \eqref{xninbcmpe} we can choose $n_\delta$ so that 
	\begin{align}
		x_{n_\delta}\in \left(\alpha_1+\frac{\delta}{2},c+\delta\right).\label{cbxnein}
	\end{align}
	Let $i\in\{1,2,\cdots,K_1-2\}$ be fixed and define
	\begin{align}
		X_i=\{n:n_\delta< n<n_\delta+N, x_{n+1}<\alpha_i\leq x_n\ \mathrm{or}\ x_n< \alpha_i\leq x_{n+1}\}.\label{dxicb}
	\end{align}
By the definition, for each $n\in X_i$,  $x_n$ must either exit
$[\alpha_i,c+\delta)$ from $\alpha_i$, as ensured by $x_{n+1}<\alpha_i\leq x_n$, or enter
$[\alpha_i,c+\delta)$ from $\alpha_i$, as guaranteed by $x_n<\alpha_i\leq x_{n+1}$.

 Let 
	\begin{align}
		\{n_m\}_{m=1}^d:=\{n_\delta,n_\delta+N\}\cup X_i\label{dnmcb}
	\end{align}
	with 
	\begin{align}
		n_1:=n_\delta<n_2<n_3<\cdots<n_d:=n_\delta+N.\label{cbdn}
	\end{align}
	
	Then by \eqref{cbxnein}-\eqref{cbdn}, one has that for any $n\in (n_m,n_{m+1}]$ with odd $m\leq d-1$,
	\begin{align}
		x_n\in [\alpha_i,c+{\delta}),\nonumber
	\end{align}
	and for any $n\in (n_m,n_{m+1}]$ with even $m\leq d-1$, 
	\begin{align}
		x_n<\alpha_i.\nonumber
	\end{align}
	Roughly speaking, $x_{n_m}$  exits $[\alpha_i,c+{\delta})$ when $m$ is even, and enters $[\alpha_i,c+{\delta})$ when $m$ is odd. A natural idea is to let $I_i=\mathcal{I}_i$. However,
	\begin{align}
		\abs{\sum_{n\in X_i\cap \mathcal{I}_i}(x_{n+1}-x_n)}\nonumber
	\end{align}
	may be very large so that \eqref{Iileq10} fails. Therefore, we need to carefully choose the points. 
	We remove the redundant elements from $\{n_m\}_{m=1}^d$. For any $m_1,m_2\in \{1,2,\cdots,d\}$ with $m_1<m_2$, if for any $n\in [n_{m_1},n_{m_2}]$, one has 
	$$x_n> \alpha_i-\frac{\delta}{2},$$
	then we remove all $n_m\in (n_{m_1},n_{m_2})$ from $\{n_m\}_{m=1}^d$.
	
\begin{figure}[H]

	\tikzset{every picture/.style={line width=0.75pt}} 
	
	\begin{tikzpicture}[x=0.75pt,y=0.75pt,yscale=-1,xscale=1]
		
		\draw    (37.27,280.65) -- (36.61,35) ;
		\draw [shift={(36.6,33)}, rotate = 89.84] [color={rgb, 255:red, 0; green, 0; blue, 0 }  ][line width=0.75]    (10.93,-3.29) .. controls (6.95,-1.4) and (3.31,-0.3) .. (0,0) .. controls (3.31,0.3) and (6.95,1.4) .. (10.93,3.29)   ;
		\draw    (37.13,223.49) -- (280.99,223.95) ;
		\draw    (36.73,69.05) -- (280.6,69.51) ;
		\draw    (37.27,280.65) -- (279.14,281.1) ;
		\draw [shift={(281.14,281.1)}, rotate = 180.11] [color={rgb, 255:red, 0; green, 0; blue, 0 }  ][line width=0.75]    (10.93,-3.29) .. controls (6.95,-1.4) and (3.31,-0.3) .. (0,0) .. controls (3.31,0.3) and (6.95,1.4) .. (10.93,3.29)   ;
		\draw    (37.02,182.14) -- (280.89,182.6) ;
		\draw [line width=1.5]  [dash pattern={on 1.69pt off 2.76pt}]  (43,195) -- (66,162) ;
		\draw [line width=1.5]  [dash pattern={on 1.69pt off 2.76pt}]  (83,196) -- (70,166) ;
		\draw [line width=1.5]  [dash pattern={on 1.69pt off 2.76pt}]  (91.2,195.39) -- (109,166) ;
		\draw [line width=1.5]  [dash pattern={on 1.69pt off 2.76pt}]  (113.85,168.08) -- (126,194) ;
		\draw [line width=1.5]  [dash pattern={on 1.69pt off 2.76pt}]  (131,196) -- (145,171) ;
		\draw [line width=1.5]  [dash pattern={on 1.69pt off 2.76pt}]  (148.65,173.13) -- (159,201) ;
		\draw [line width=1.5]  [dash pattern={on 1.69pt off 2.76pt}]  (164.37,197.23) -- (176,171) ;
		\draw [line width=1.5]  [dash pattern={on 1.69pt off 2.76pt}]  (194,197) -- (181.46,173.36) ;
		\draw [line width=1.5]  [dash pattern={on 1.69pt off 2.76pt}]  (199,195) -- (214,166) ;
		\draw    (372.22,280.65) -- (371.55,35) ;
		\draw [shift={(371.55,33)}, rotate = 89.84] [color={rgb, 255:red, 0; green, 0; blue, 0 }  ][line width=0.75]    (10.93,-3.29) .. controls (6.95,-1.4) and (3.31,-0.3) .. (0,0) .. controls (3.31,0.3) and (6.95,1.4) .. (10.93,3.29)   ;
		\draw    (372.07,223.49) -- (615.94,223.95) ;
		\draw    (371.68,69.05) -- (615.54,69.51) ;
		\draw    (372.22,280.65) -- (614.08,281.1) ;
		\draw [shift={(616.08,281.1)}, rotate = 180.11] [color={rgb, 255:red, 0; green, 0; blue, 0 }  ][line width=0.75]    (10.93,-3.29) .. controls (6.95,-1.4) and (3.31,-0.3) .. (0,0) .. controls (3.31,0.3) and (6.95,1.4) .. (10.93,3.29)   ;
		\draw    (293,163) -- (344,163) ;
		\draw [shift={(346,163)}, rotate = 180] [color={rgb, 255:red, 0; green, 0; blue, 0 }  ][line width=0.75]    (10.93,-3.29) .. controls (6.95,-1.4) and (3.31,-0.3) .. (0,0) .. controls (3.31,0.3) and (6.95,1.4) .. (10.93,3.29)   ;
		\draw    (372.02,182.14) -- (615.89,182.6) ;
		\draw [line width=1.5]  [dash pattern={on 1.69pt off 2.76pt}]  (378,195) -- (401,162) ;
		\draw [line width=1.5]  [dash pattern={on 1.69pt off 2.76pt}]  (418,196) -- (405,166) ;
		\draw [line width=1.5]  [dash pattern={on 1.69pt off 2.76pt}]  (426.2,196.39) -- (444,167) ;
		\draw [line width=1.5]  [dash pattern={on 1.69pt off 2.76pt}]  (448.85,168.08) -- (461,194) ;
		\draw [line width=1.5]  [dash pattern={on 1.69pt off 2.76pt}]  (466,196) -- (480,171) ;
		\draw [line width=1.5]  [dash pattern={on 1.69pt off 2.76pt}]  (483.65,173.13) -- (494,201) ;
		\draw [line width=1.5]  [dash pattern={on 1.69pt off 2.76pt}]  (499.37,197.23) -- (511,171) ;
		\draw [line width=1.5]  [dash pattern={on 1.69pt off 2.76pt}]  (529,197) -- (516.46,173.36) ;
		\draw [line width=1.5]  [dash pattern={on 1.69pt off 2.76pt}]  (534,195) -- (549,166) ;
		\draw  [dash pattern={on 4.5pt off 4.5pt}]  (51,183) -- (50,281) ;
		\draw  [dash pattern={on 4.5pt off 4.5pt}]  (75,180) -- (74,280) ;
		\draw  [dash pattern={on 4.5pt off 4.5pt}]  (98,183) -- (97,281) ;
		\draw  [dash pattern={on 4.5pt off 4.5pt}]  (119,181) -- (118,280) ;
		\draw  [dash pattern={on 4.5pt off 4.5pt}]  (138,183) -- (137,281) ;
		\draw  [dash pattern={on 4.5pt off 4.5pt}]  (150.82,181.07) -- (149.82,281.07) ;
		\draw  [dash pattern={on 4.5pt off 4.5pt}]  (170,184) -- (169,281) ;
		\draw  [dash pattern={on 4.5pt off 4.5pt}]  (185,181) -- (184,281) ;
		\draw  [dash pattern={on 4.5pt off 4.5pt}]  (205,183) -- (204,281) ;
		\draw  [dash pattern={on 4.5pt off 4.5pt}]  (386,183) -- (385,281) ;
		\draw  [dash pattern={on 4.5pt off 4.5pt}]  (540,183) -- (539,281) ;
		
		\draw (15.05,178) node [anchor=north west][inner sep=0.75pt]  [rotate=-359.98] [align=left] {\scriptsize$\alpha_i$};
		\draw (0,218) node [anchor=north west][inner sep=0.75pt]  [rotate=-359.98] [align=left] {\scriptsize$\alpha_i-\frac{\delta}{2}$};
		\draw (25,271.97) node [anchor=north west][inner sep=0.75pt]  [rotate=-359.98] [align=left] {\scriptsize0};
		\draw (285.71,275) node [anchor=north west][inner sep=0.75pt]  [rotate=-359.98] [align=left] {$n$};
		\draw (535,283) node [anchor=north west][inner sep=0.75pt]  [rotate=-359.98] [align=left] {\scriptsize$n_9$};
		\draw (350,178) node [anchor=north west][inner sep=0.75pt]  [rotate=-359.98] [align=left] {\scriptsize$\alpha_i$};
		\draw (334,218) node [anchor=north west][inner sep=0.75pt]  [rotate=-359.98] [align=left] {\scriptsize$\alpha_i-\frac{\delta}{2}$};
		\draw (360,271.97) node [anchor=north west][inner sep=0.75pt]  [rotate=-359.98] [align=left] {\scriptsize0};
		\draw (377.46,283) node [anchor=north west][inner sep=0.75pt]  [rotate=-359.98] [align=left] {\scriptsize$n_1$};
		\draw (620.66,275) node [anchor=north west][inner sep=0.75pt]  [rotate=-359.98] [align=left] {$n$};
		\draw (41.27,283) node [anchor=north west][inner sep=0.75pt]   [align=left] {\scriptsize$n_1$};
		\draw (66.27,283) node [anchor=north west][inner sep=0.75pt]   [align=left] {\scriptsize$n_2$};
		\draw (90.27,283) node [anchor=north west][inner sep=0.75pt]   [align=left] {\scriptsize$n_3$};
		\draw (110.27,283) node [anchor=north west][inner sep=0.75pt]   [align=left] {\scriptsize$n_4$};
		\draw (130.27,283) node [anchor=north west][inner sep=0.75pt]   [align=left] {\scriptsize$n_5$};
		\draw (145.82,283) node [anchor=north west][inner sep=0.75pt]   [align=left] {\scriptsize$n_6$};
		\draw (162,283) node [anchor=north west][inner sep=0.75pt]   [align=left] {\scriptsize$n_7$};
		\draw (176,283) node [anchor=north west][inner sep=0.75pt]   [align=left] {\scriptsize$n_8$};
		\draw (197,283) node [anchor=north west][inner sep=0.75pt]   [align=left] {\scriptsize$n_9$};
		\draw (5,61) node [anchor=north west][inner sep=0.75pt]  [rotate=-359.98] [align=left] {\scriptsize$c+\delta$};
		\draw (340,61) node [anchor=north west][inner sep=0.75pt]  [rotate=-359.98] [align=left] {\scriptsize$c+\delta$};
		\draw (43.6,32) node [anchor=north west][inner sep=0.75pt]  [rotate=-359.98] [align=left] {$x_n$};
		\draw (377.55,32) node [anchor=north west][inner sep=0.75pt]  [rotate=-359.98] [align=left] {$x_n$};

	\end{tikzpicture}
	\caption{}
\end{figure}
	
	 After this deletion procedure, the remaining crossing points still alternate between entering and exiting the strip $[\alpha_i,c+\delta)$, but each retained gap now contains a genuine excursion down to the lower level $\alpha_i-\frac{\delta}{2}$.
	 To avoid introducing additional notation, we continue to denote the remaining elements by $\mathcal{X}_i=\{n_m\}_{m=1}^d$ ($n_1=n_\delta$ and $n_d=n_\delta+N$ will not be removed, see Figure \ref{cbremove} for example).

\begin{figure}[H]
	\centering
	
	\tikzset{every picture/.style={line width=0.75pt}}
	\scalebox{1.05}{
		\rotatebox[origin=c]{90.2}{%
			\begin{tikzpicture}[x=0.75pt,y=0.75pt,yscale=-1,xscale=1]
				
				\draw (162.77,63.83) -- (475.47,63);
				\draw [shift={(477.47,62.99)},rotate=179.85]
				[color={rgb,255:red,0;green,0;blue,0}][line width=0.75]
				(10.93,-3.29) .. controls (6.95,-1.4) and (3.31,-0.3)
				.. (0,0) .. controls (3.31,0.3) and (6.95,1.4)
				.. (10.93,3.29);
				
				\draw (231.24,63.65) -- (230.64,386.88);
				\draw (280.78,63.52) -- (280.18,386.75);
				\draw (416.28,63.16) -- (415.68,386.38);
				
				\draw (162.77,63.83) -- (162.17,400.06);
				\draw [shift={(162.17,402.06)},rotate=270.11]
				[color={rgb,255:red,0;green,0;blue,0}][line width=0.75]
				(10.93,-3.29) .. controls (6.95,-1.4) and (3.31,-0.3)
				.. (0,0) .. controls (3.31,0.3) and (6.95,1.4)
				.. (10.93,3.29);
				
				\draw [line width=1.5,dash pattern={on 1.69pt off 2.76pt}]
				(406.61,69.69) -- (259.4,101.53);
				\draw [line width=1.5,dash pattern={on 1.69pt off 2.76pt}]
				(254.11,102.42) -- (332.74,127.51);
				\draw [line width=1.5,dash pattern={on 1.69pt off 2.76pt}]
				(336.02,129.3) -- (219.41,160.52);
				\draw [line width=1.5,dash pattern={on 1.69pt off 2.76pt}]
				(330.16,193.16) -- (219.49,163.94);
				\draw [line width=1.5,dash pattern={on 1.69pt off 2.76pt}]
				(335.9,194.54) -- (251.35,221.46);
				\draw [line width=1.5,dash pattern={on 1.69pt off 2.76pt}]
				(295.57,238.62) -- (247.52,223.29);
				\draw [line width=1.5,dash pattern={on 1.69pt off 2.76pt}]
				(224.59,263.29) -- (304.22,241.4);
				\draw [line width=1.5,dash pattern={on 1.69pt off 2.76pt}]
				(223.84,266.73) -- (314.14,286.16);
				\draw [line width=1.5,dash pattern={on 1.69pt off 2.76pt}]
				(318.96,287.55) -- (248.99,312);
				\draw [line width=1.5,dash pattern={on 1.69pt off 2.76pt}]
				(249.72,316.19) -- (370.58,375.38);
				
				\draw [dash pattern={on 4.5pt off 4.5pt}]
				(285.18,95.9) -- (163.08,94.8);
				\draw [dash pattern={on 4.5pt off 4.5pt}]
				(276.86,110.79) -- (162.59,109.67);
				\draw [dash pattern={on 4.5pt off 4.5pt}]
				(284.09,143.09) -- (162.62,142);
				\draw [dash pattern={on 4.5pt off 4.5pt}]
				(276.28,180.25) -- (163.01,179.13);
				\draw [dash pattern={on 4.5pt off 4.5pt}]
				(283.96,211.14) -- (163.04,210.06);
				\draw [dash pattern={on 4.5pt off 4.5pt}]
				(276.18,233.05) -- (163.91,231.94);
				\draw [dash pattern={on 4.5pt off 4.5pt}]
				(283.98,247.09) -- (162.43,245.99);
				\draw [dash pattern={on 4.5pt off 4.5pt}]
				(276.55,278.61) -- (162.83,277.5);
				\draw [dash pattern={on 4.5pt off 4.5pt}]
				(283.26,299.64) -- (161.88,298.56);
				\draw [dash pattern={on 4.5pt off 4.5pt}]
				(275.93,330.15) -- (159.38,329.05);
				\draw [dash pattern={on 4.5pt off 4.5pt}]
				(404.61,69.68) -- (162.76,68.92);
				\draw [dash pattern={on 4.5pt off 4.5pt}]
				(368.58,374.38) -- (162.69,373.53);
				
				\draw (476.78,75.7) node
				[anchor=north west,inner sep=0.75pt,rotate=-89.99] {$x_n$};
				
				\draw (159,64.1) node
				[anchor=north west,inner sep=0.75pt,rotate=-89.99] {$n_1$};
				\draw (159,84.97) node
				[anchor=north west,inner sep=0.75pt,rotate=-89.99] {$n_2$};
				\draw (159,103.24) node
				[anchor=north west,inner sep=0.75pt,rotate=-89.99] {$n_3$};
				\draw (159,136.97) node
				[anchor=north west,inner sep=0.75pt,rotate=-89.99] {$n_4$};
				\draw (159,173.98) node
				[anchor=north west,inner sep=0.75pt,rotate=-89.99] {$n_5$};
				\draw (159,203.02) node
				[anchor=north west,inner sep=0.75pt,rotate=-89.99] {$n_6$};
				\draw (159,226.33) node
				[anchor=north west,inner sep=0.75pt,rotate=-89.99] {$n_7$};
				\draw (159,242.38) node
				[anchor=north west,inner sep=0.75pt,rotate=-89.99] {$n_8$};
				\draw (159,270.68) node
				[anchor=north west,inner sep=0.75pt,rotate=-89.99] {$n_9$};
				\draw (159,291.94) node
				[anchor=north west,inner sep=0.75pt,rotate=-89.99] {$n_{10}$};
				\draw (159,324.43) node
				[anchor=north west,inner sep=0.75pt,rotate=-89.99] {$n_{11}$};
				\draw (159,367.09) node
				[anchor=north west,inner sep=0.75pt,rotate=-89.99] {$n_{12}$};
				
				\draw (424.39,25) node
				[anchor=north west,inner sep=0.75pt,rotate=-89.99] {$c+\delta$};
				\draw (285,45) node
				[anchor=north west,inner sep=0.75pt,rotate=-89.99] {$\alpha_i$};
				\draw (240,16) node
				[anchor=north west,inner sep=0.75pt,rotate=-89.99]
				{$\alpha_i-\frac{\delta}{2}$};
				\draw (166,45) node
				[anchor=north west,inner sep=0.75pt,rotate=-89.99] {$0$};
				\draw (165,410) node
				[anchor=north west,inner sep=0.75pt,rotate=-89.99] {$n$};
				
			\end{tikzpicture}%
	}}

	\begin{tikzpicture}
		\draw[->,line width=1pt] (0,0.55) -- (0,-0.55);
	\end{tikzpicture}

	\scalebox{1.05}{
		\rotatebox[origin=c]{90.2}{%
			\begin{tikzpicture}[x=0.75pt,y=0.75pt,yscale=-1,xscale=1]
				
				\draw (162.6,486.44) -- (475.31,486.86);
				\draw [shift={(477.31,486.87)},rotate=180.08]
				[color={rgb,255:red,0;green,0;blue,0}][line width=0.75]
				(10.93,-3.29) .. controls (6.95,-1.4) and (3.31,-0.3)
				.. (0,0) .. controls (3.31,0.3) and (6.95,1.4)
				.. (10.93,3.29);
				
				\draw (231.08,486.53) -- (229.18,809.76);
				\draw (280.62,486.6) -- (278.72,809.82);
				\draw (416.11,486.78) -- (414.22,810.01);
				
				\draw (162.6,486.44) -- (160.72,822.66);
				\draw [shift={(160.7,824.66)},rotate=270.34]
				[color={rgb,255:red,0;green,0;blue,0}][line width=0.75]
				(10.93,-3.29) .. controls (6.95,-1.4) and (3.31,-0.3)
				.. (0,0) .. controls (3.31,0.3) and (6.95,1.4)
				.. (10.93,3.29);
				
				\draw [line width=1.5,dash pattern={on 1.69pt off 2.76pt}]
				(406.42,493.27) -- (259.08,524.52);
				\draw [line width=1.5,dash pattern={on 1.69pt off 2.76pt}]
				(253.79,525.4) -- (332.32,550.8);
				\draw [line width=1.5,dash pattern={on 1.69pt off 2.76pt}]
				(335.59,552.6) -- (218.85,583.36);
				\draw [line width=1.5,dash pattern={on 1.69pt off 2.76pt}]
				(329.48,616.43) -- (218.92,586.77);
				\draw [line width=1.5,dash pattern={on 1.69pt off 2.76pt}]
				(335.21,617.84) -- (250.55,644.42);
				\draw [line width=1.5,dash pattern={on 1.69pt off 2.76pt}]
				(294.7,661.76) -- (246.71,646.24);
				\draw [line width=1.5,dash pattern={on 1.69pt off 2.76pt}]
				(223.62,686.15) -- (303.34,664.57);
				\draw [line width=1.5,dash pattern={on 1.69pt off 2.76pt}]
				(222.86,689.58) -- (313.08,709.37);
				\draw [line width=1.5,dash pattern={on 1.69pt off 2.76pt}]
				(317.9,710.78) -- (247.83,734.96);
				\draw [line width=1.5,dash pattern={on 1.69pt off 2.76pt}]
				(248.55,739.15) -- (369.16,798.82);
				
				\draw [dash pattern={on 4.5pt off 4.5pt}]
				(283.61,566.18) -- (162.14,564.61);
				\draw [dash pattern={on 4.5pt off 4.5pt}]
				(275.65,603.31) -- (162.38,601.74);
				\draw [dash pattern={on 4.5pt off 4.5pt}]
				(283.08,670.19) -- (161.53,668.6);
				\draw [dash pattern={on 4.5pt off 4.5pt}]
				(275.53,701.68) -- (161.81,700.11);
				\draw [dash pattern={on 4.5pt off 4.5pt}]
				(404.42,493.26) -- (162.57,491.53);
				\draw [dash pattern={on 4.5pt off 4.5pt}]
				(367.16,797.81) -- (161.29,796.13);
				
				\draw (476.56,499.57) node
				[anchor=north west,inner sep=0.75pt,rotate=-90.22] {$x_n$};
				
				\draw (159,486.69) node
				[anchor=north west,inner sep=0.75pt,rotate=-90.22] {$n_1$};
				\draw (159,559.57) node
				[anchor=north west,inner sep=0.75pt,rotate=-90.22] {$n_2$};
				\draw (159,596.57) node
				[anchor=north west,inner sep=0.75pt,rotate=-90.22] {$n_3$};
				\draw (159,664.98) node
				[anchor=north west,inner sep=0.75pt,rotate=-90.22] {$n_4$};
				\draw (159,693.28) node
				[anchor=north west,inner sep=0.75pt,rotate=-90.22] {$n_5$};
				\draw (159,789.69) node
				[anchor=north west,inner sep=0.75pt,rotate=-90.22] {$n_6$};
				
				\draw (424.33,450) node
				[anchor=north west,inner sep=0.75pt,rotate=-90.22] {$c+\delta$};
				\draw (285,465) node
				[anchor=north west,inner sep=0.75pt,rotate=-90.22] {$\alpha_i$};
				\draw (240,441) node
				[anchor=north west,inner sep=0.75pt,rotate=-90.22]
				{$\alpha_i-\frac{\delta}{2}$};
				\draw (170.83,465) node
				[anchor=north west,inner sep=0.75pt,rotate=-90.22] {$0$};
				\draw (165,833.36) node
				[anchor=north west,inner sep=0.75pt,rotate=-90.22] {$n$};
				
			\end{tikzpicture}%
		}
	}
	\caption{}
	\label{cbremove}
\end{figure}

After the deletions, by \eqref{cbdmi1}-\eqref{cbdn} we can conclude that (See Figure \ref{cbfigadele})
\begin{enumerate}
	\item[$\bullet$] For any $m=2,3,\cdots,d-1$, 
	\begin{align}
		x_{n_m+1}<\alpha_i\leq x_{n_m},\ \mathrm{if}\ m\ \mathrm{is\ even},\label{cbcrosse}\\
		x_{n_m}<\alpha_i\leq x_{n_m+1},\ \mathrm{if}\ m\ \mathrm{is\ odd}.\label{cbcrosso}
	\end{align}
	\item[$\bullet$] For any  $n\in [n_m,n_{m+1}]$ with odd $m\leq d-1$,
	\begin{align}
		x_n\in \left(\alpha_i-\frac{\delta}{2},c+\delta\right)\subset\cup_{l=1}^i\mathcal{I}_l\cup\frac{1}{2}\mathcal{I}_{i+1}.\label{cboddmxnin}
	\end{align}
	\item[$\bullet$] For any $n\in (n_m,n_{m+1}]$ with even $m\leq d-1$,
	\begin{align}
		x_n< \alpha_i,\label{cbevenmxnin}
	\end{align}
	moreover, there exists $l_m\in (n_m,n_{m+1}]$ such that 
	\begin{align}
		x_{l_m}\leq \alpha_i-\frac{\delta}{2}.\label{cbxlmleq}
	\end{align}
\end{enumerate} 
\begin{figure}[H]
	\centering

	\tikzset{every picture/.style={line width=0.75pt}} 
	
	\begin{tikzpicture}[x=0.75pt,y=0.75pt,yscale=-1,xscale=1]
		
		\draw    (157.66,332.4) -- (156.77,37.7) ;
		\draw [shift={(156.77,35.7)}, rotate = 89.83] [color={rgb, 255:red, 0; green, 0; blue, 0 }  ][line width=0.75]    (10.93,-3.29) .. controls (6.95,-1.4) and (3.31,-0.3) .. (0,0) .. controls (3.31,0.3) and (6.95,1.4) .. (10.93,3.29)   ;
		\draw    (157.47,263.92) -- (480.69,264.47) ;
		\draw    (156.94,78.89) -- (480.17,79.43) ;
		\draw    (157.66,332.4) -- (478.89,332.94) ;
		\draw [shift={(480.89,332.94)}, rotate = 180.1] [color={rgb, 255:red, 0; green, 0; blue, 0 }  ][line width=0.75]    (10.93,-3.29) .. controls (6.95,-1.4) and (3.31,-0.3) .. (0,0) .. controls (3.31,0.3) and (6.95,1.4) .. (10.93,3.29)   ;
		\draw    (157.33,214.39) -- (480.55,214.93) ;
		\draw [line width=1.5]  [dash pattern={on 1.69pt off 2.76pt}]  (163.47,88.56) -- (195.34,235.76) ;
		\draw [line width=1.5]  [dash pattern={on 1.69pt off 2.76pt}]  (196.23,241.05) -- (221.3,162.42) ;
		\draw [line width=1.5]  [dash pattern={on 1.69pt off 2.76pt}]  (223.09,159.13) -- (254.34,275.74) ;
		\draw [line width=1.5]  [dash pattern={on 1.69pt off 2.76pt}]  (286.95,164.98) -- (257.76,275.66) ;
		\draw [line width=1.5]  [dash pattern={on 1.69pt off 2.76pt}]  (288.33,159.24) -- (315.27,243.79) ;
		\draw [line width=1.5]  [dash pattern={on 1.69pt off 2.76pt}]  (332.42,199.57) -- (317.1,247.62) ;
		\draw [line width=1.5]  [dash pattern={on 1.69pt off 2.76pt}]  (357.11,270.54) -- (335.2,190.92) ;
		\draw [line width=1.5]  [dash pattern={on 1.69pt off 2.76pt}]  (360.54,271.29) -- (379.96,180.99) ;
		\draw [line width=1.5]  [dash pattern={on 1.69pt off 2.76pt}]  (381.35,176.16) -- (405.82,246.13) ;
		\draw [line width=1.5]  [dash pattern={on 1.69pt off 2.76pt}]  (410,245.4) -- (469.17,124.54) ;
		\draw  [dash pattern={on 4.5pt off 4.5pt}]  (236.89,211.06) -- (235.83,332.53) ;
		\draw  [dash pattern={on 4.5pt off 4.5pt}]  (274.05,218.87) -- (272.96,332.14) ;
		\draw  [dash pattern={on 4.5pt off 4.5pt}]  (340.9,211.15) -- (339.82,332.71) ;
		\draw  [dash pattern={on 4.5pt off 4.5pt}]  (372.42,218.57) -- (371.33,332.3) ;
		\draw  [dash pattern={on 4.5pt off 4.5pt}]  (163.47,90.56) -- (162.75,332.41) ;
		\draw  [dash pattern={on 4.5pt off 4.5pt}]  (468.17,126.53) -- (467.35,332.42) ;
		\draw  [dash pattern={on 4.5pt off 4.5pt}]  (254.34,275.74) -- (253.24,334.01) ;
		\draw  [dash pattern={on 4.5pt off 4.5pt}]  (357.34,270.74) -- (356.24,332.01) ;
		
		\draw (267.79,335) node [anchor=north west][inner sep=0.75pt]  [rotate=-359.98] [align=left] {$n_3$};
		\draw (371.33,335) node [anchor=north west][inner sep=0.75pt]  [rotate=-359.98] [align=left] {$n_5$};
		\draw (228,335) node [anchor=north west][inner sep=0.75pt]  [rotate=-359.98] [align=left] {$n_2$};
		\draw (332.5,335) node [anchor=north west][inner sep=0.75pt]  [rotate=-359.98] [align=left] {$n_4$};
		\draw (350.24,335) node [anchor=north west][inner sep=0.75pt]  [rotate=-359.98] [align=left] {${l_4}$};
		\draw (110,70.78) node [anchor=north west][inner sep=0.75pt]  [rotate=-359.98] [align=left] {$c+\delta$};
		\draw (135,208) node [anchor=north west][inner sep=0.75pt]  [rotate=-359.98] [align=left] {$\alpha_i$};
		\draw (110,255) node [anchor=north west][inner sep=0.75pt]  [rotate=-359.98] [align=left] {$\alpha_i-\frac{\delta}{2}$};
		\draw (135,324.28) node [anchor=north west][inner sep=0.75pt]  [rotate=-359.98] [align=left] {0};
		\draw (157.93,335) node [anchor=north west][inner sep=0.75pt]  [rotate=-359.98] [align=left] {$n_1$};
		\draw (460.92,335) node [anchor=north west][inner sep=0.75pt]  [rotate=-359.98] [align=left] {$n_6$};
		\draw (489.55,324.89) node [anchor=north west][inner sep=0.75pt]  [rotate=-359.98] [align=left] {$n$};
		\draw (248,335) node [anchor=north west][inner sep=0.75pt]  [rotate=-359.98] [align=left] {${l_2}$};
		\draw (161.08,27.78) node [anchor=north west][inner sep=0.75pt]  [rotate=-359.98] [align=left] {$x_n$};

	\end{tikzpicture}
	\caption{}
	\label{cbfigadele}
\end{figure}

 Note that $d=d(i)$ and each $l_m=l_m(i),n_m=n_m(i)$ depend on $i$, but we leave the dependence on $i$ implicit for simplicity. 
By \eqref{cbxn1xnlle}, \eqref{cbcrosse} and \eqref{cbxlmleq}, for any even $m\leq d-1$,
 \begin{align}
 	\label{cbnmlm}\abs{x_{l_m}-x_{n_m}}\geq \frac{\delta}{4}.
 \end{align}
 Define $L_i$ based on $\mathcal{X}_i=\{n_m\}^d_{m=1}$:
	\begin{align}
		L_i=\begin{cases}
			\cup_{m=1}^{\frac{d}{2}}\{n_{2m-1},n_{2m-1}+1,\cdots,n_{2m}\}, &\mathrm{if}\ d\ \mathrm{is\ even},\\
			\cup_{m=1}^{\frac{d-1}{2}}\{n_{2m-1},n_{2m-1}+1,\cdots,n_{2m}\}, &\mathrm{if}\ d\ \mathrm{is\ odd}.\nonumber
		\end{cases}
	\end{align}
	See Figure \ref{figlie} and Figure \ref{figlio} for example, where the elements of $L_i$ are marked in red.
	By letting $i=1,2,\cdots,K_1-2$, we can obtain $\{L_i\}_{i=1}^{K_1-2}$.
\begin{figure}[H]
	\centering

	\tikzset{every picture/.style={line width=0.75pt}} 
	
	\begin{tikzpicture}[x=0.75pt,y=0.75pt,yscale=-1,xscale=1]
		
		\draw    (152.66,333.41) -- (151.77,38.71) ;
		\draw [shift={(151.77,36.71)}, rotate = 89.83] [color={rgb, 255:red, 0; green, 0; blue, 0 }  ][line width=0.75]    (10.93,-3.29) .. controls (6.95,-1.4) and (3.31,-0.3) .. (0,0) .. controls (3.31,0.3) and (6.95,1.4) .. (10.93,3.29)   ;
		\draw    (152.47,264.94) -- (475.69,265.48) ;
		\draw    (151.94,79.9) -- (475.17,80.45) ;
		\draw    (152.66,333.41) -- (473.89,333.96) ;
		\draw [shift={(475.89,333.96)}, rotate = 180.1] [color={rgb, 255:red, 0; green, 0; blue, 0 }  ][line width=0.75]    (10.93,-3.29) .. controls (6.95,-1.4) and (3.31,-0.3) .. (0,0) .. controls (3.31,0.3) and (6.95,1.4) .. (10.93,3.29)   ;
		\draw    (152.33,215.4) -- (475.55,215.95) ;
		\draw [color={rgb, 255:red, 208; green, 2; blue, 27 }  ,draw opacity=1 ][line width=1.5]  [dash pattern={on 1.69pt off 2.76pt}]  (158.47,89.57) -- (190.34,236.78) ;
		\draw [color={rgb, 255:red, 208; green, 2; blue, 27 }  ,draw opacity=1 ][line width=1.5]  [dash pattern={on 1.69pt off 2.76pt}]  (191.23,242.06) -- (216.3,163.43) ;
		\draw [line width=1.5]  [dash pattern={on 1.69pt off 2.76pt}]  (266.35,223.34) -- (251.76,273.68) ;
		\draw [color={rgb, 255:red, 208; green, 2; blue, 27 }  ,draw opacity=1 ][line width=1.5]  [dash pattern={on 1.69pt off 2.76pt}]  (283.33,160.26) -- (310.27,244.81) ;
		\draw [line width=1.5]  [dash pattern={on 1.69pt off 2.76pt}]  (352.11,271.56) -- (337.2,215.93) ;
		\draw [color={rgb, 255:red, 208; green, 2; blue, 27 }  ,draw opacity=1 ][line width=1.5]  [dash pattern={on 1.69pt off 2.76pt}]  (376.35,177.18) -- (400.82,247.15) ;
		\draw [color={rgb, 255:red, 208; green, 2; blue, 27 }  ,draw opacity=1 ][line width=1.5]  [dash pattern={on 1.69pt off 2.76pt}]  (405,246.41) -- (464.17,125.55) ;
		\draw  [dash pattern={on 4.5pt off 4.5pt}]  (231.89,212.08) -- (230.83,333.55) ;
		\draw  [dash pattern={on 4.5pt off 4.5pt}]  (269.05,219.88) -- (267.96,333.15) ;
		\draw  [dash pattern={on 4.5pt off 4.5pt}]  (335.9,212.17) -- (334.82,333.72) ;
		\draw  [dash pattern={on 4.5pt off 4.5pt}]  (367.42,219.59) -- (366.33,333.32) ;
		\draw  [dash pattern={on 4.5pt off 4.5pt}]  (158.47,91.57) -- (157.75,333.42) ;
		\draw  [dash pattern={on 4.5pt off 4.5pt}]  (463.17,127.55) -- (462.35,333.44) ;
	
		
		\draw [color={rgb, 255:red, 208; green, 2; blue, 27 }  ,draw opacity=1 ][line width=1.5]  [dash pattern={on 1.69pt off 2.76pt}]  (218.09,160.15) -- (233.72,217.46) ;
		\draw [color={rgb, 255:red, 208; green, 2; blue, 27 }  ,draw opacity=1 ][line width=1.5]  [dash pattern={on 1.69pt off 2.76pt}]  (283.33,160.26) -- (267.35,221.34) ;
		\draw [color={rgb, 255:red, 208; green, 2; blue, 27 }  ,draw opacity=1 ][line width=1.5]  [dash pattern={on 1.69pt off 2.76pt}]  (329.42,199.59) -- (311.1,246.64) ;
		\draw [color={rgb, 255:red, 208; green, 2; blue, 27 }  ,draw opacity=1 ][line width=1.5]  [dash pattern={on 1.69pt off 2.76pt}]  (331.2,194.93) -- (336.9,215.17) ;
		\draw [color={rgb, 255:red, 208; green, 2; blue, 27 }  ,draw opacity=1 ][line width=1.5]  [dash pattern={on 1.69pt off 2.76pt}]  (376.35,177.18) -- (367.42,219.59) ;
		\draw [line width=1.5]  [dash pattern={on 1.69pt off 2.76pt}]  (355.54,272.3) -- (366.25,224.15) ;
		\draw [line width=1.5]  [dash pattern={on 1.69pt off 2.76pt}]  (249.34,276.76) -- (233.72,217.46) ;
		\draw (263,335) node [anchor=north west][inner sep=0.75pt]  [rotate=-359.98] [align=left] {$n_3$};
		\draw (363,335) node [anchor=north west][inner sep=0.75pt]  [rotate=-359.98] [align=left] {$n_5$};
		\draw (225,335) node [anchor=north west][inner sep=0.75pt]  [rotate=-359.98] [align=left] {$n_2$};
		\draw (330,335) node [anchor=north west][inner sep=0.75pt]  [rotate=-359.98] [align=left] {$n_4$};
		
		\draw (110,70.78) node [anchor=north west][inner sep=0.75pt]  [rotate=-359.98] [align=left] {$c+\delta$};
		\draw (135,208) node [anchor=north west][inner sep=0.75pt]  [rotate=-359.98] [align=left] {$\alpha_i$};
		\draw (105,255) node [anchor=north west][inner sep=0.75pt]  [rotate=-359.98] [align=left] {$\alpha_i-\frac{\delta}{2}$};
		\draw (135,324.28) node [anchor=north west][inner sep=0.75pt]  [rotate=-359.98] [align=left] {0};
		\draw (153,335) node [anchor=north west][inner sep=0.75pt]  [rotate=-359.98] [align=left] {$n_1$};
		\draw (455,335) node [anchor=north west][inner sep=0.75pt]  [rotate=-359.98] [align=left] {$n_6$};
		\draw (484,327) node [anchor=north west][inner sep=0.75pt]  [rotate=-359.98] [align=left] {$n$};
	
		\draw (161.08,29) node [anchor=north west][inner sep=0.75pt]  [rotate=-359.98] [align=left] {$x_n$};

	\end{tikzpicture}
	\caption{$d=6$ is even}
	\label{figlie}
\end{figure}
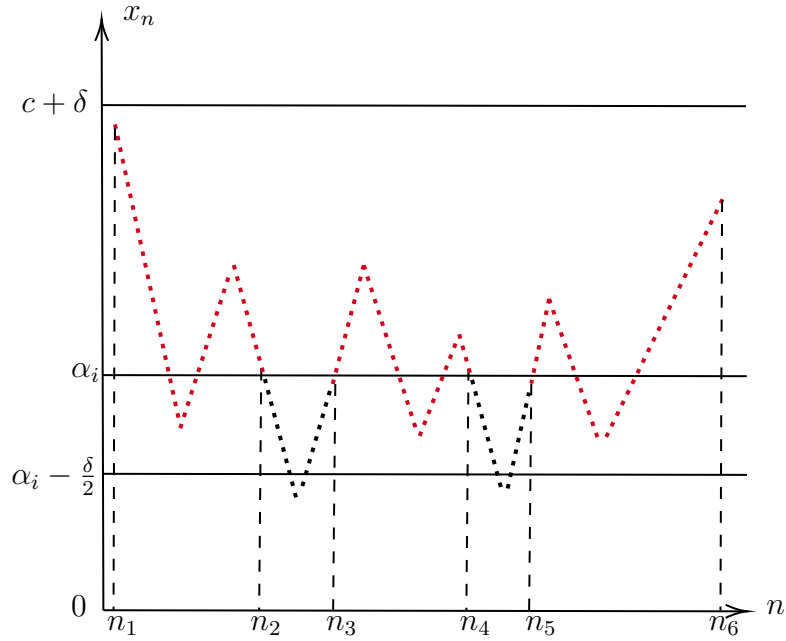

\begin{figure}[H]
	\centering

	\tikzset{every picture/.style={line width=0.75pt}} 
	
	\begin{tikzpicture}[x=0.75pt,y=0.75pt,yscale=-1,xscale=1]
		
		\draw    (152.66,333.41) -- (151.77,38.71) ;
		\draw [shift={(151.77,36.71)}, rotate = 89.83] [color={rgb, 255:red, 0; green, 0; blue, 0 }  ][line width=0.75]    (10.93,-3.29) .. controls (6.95,-1.4) and (3.31,-0.3) .. (0,0) .. controls (3.31,0.3) and (6.95,1.4) .. (10.93,3.29)   ;
		\draw    (152.47,264.94) -- (475.69,265.48) ;
		\draw    (151.94,79.9) -- (475.17,80.45) ;
		\draw    (152.66,333.41) -- (473.89,333.96) ;
		\draw [shift={(475.89,333.96)}, rotate = 180.1] [color={rgb, 255:red, 0; green, 0; blue, 0 }  ][line width=0.75]    (10.93,-3.29) .. controls (6.95,-1.4) and (3.31,-0.3) .. (0,0) .. controls (3.31,0.3) and (6.95,1.4) .. (10.93,3.29)   ;
		\draw    (152.33,215.4) -- (475.55,215.95) ;
		\draw [color={rgb, 255:red, 208; green, 2; blue, 27 }  ,draw opacity=1 ][line width=1.5]  [dash pattern={on 1.69pt off 2.76pt}]  (158.47,89.57) -- (190.34,236.78) ;
		\draw [color={rgb, 255:red, 208; green, 2; blue, 27 }  ,draw opacity=1 ][line width=1.5]  [dash pattern={on 1.69pt off 2.76pt}]  (191.23,242.06) -- (216.3,163.43) ;
		\draw [line width=1.5]  [dash pattern={on 1.69pt off 2.76pt}]  (266.35,223.34) -- (251.76,273.68) ;
		\draw [color={rgb, 255:red, 208; green, 2; blue, 27 }  ,draw opacity=1 ][line width=1.5]  [dash pattern={on 1.69pt off 2.76pt}]  (283.33,160.26) -- (310.27,244.81) ;
		\draw [line width=1.5]  [dash pattern={on 1.69pt off 2.76pt}]  (352.11,271.56) -- (337.2,215.93) ;
		\draw [color={rgb, 255:red, 208; green, 2; blue, 27 }  ,draw opacity=1 ][line width=1.5]  [dash pattern={on 1.69pt off 2.76pt}]  (380.35,178.18) -- (409.67,213.67) ;
		\draw  [dash pattern={on 4.5pt off 4.5pt}]  (231.89,212.08) -- (230.83,333.55) ;
		\draw  [dash pattern={on 4.5pt off 4.5pt}]  (269.05,219.88) -- (267.96,333.15) ;
		\draw  [dash pattern={on 4.5pt off 4.5pt}]  (335.9,212.17) -- (334.82,333.72) ;
		\draw  [dash pattern={on 4.5pt off 4.5pt}]  (367.42,219.59) -- (366.33,333.32) ;
		\draw  [dash pattern={on 4.5pt off 4.5pt}]  (158.47,91.57) -- (157.75,333.42) ;
		\draw  [dash pattern={on 4.5pt off 4.5pt}]  (462.17,284.55) -- (461.35,333.44) ;

		\draw [color={rgb, 255:red, 208; green, 2; blue, 27 }  ,draw opacity=1 ][line width=1.5]  [dash pattern={on 1.69pt off 2.76pt}]  (218.09,160.15) -- (233.72,217.46) ;
		\draw [color={rgb, 255:red, 208; green, 2; blue, 27 }  ,draw opacity=1 ][line width=1.5]  [dash pattern={on 1.69pt off 2.76pt}]  (283.33,160.26) -- (267.35,221.34) ;
		\draw [color={rgb, 255:red, 208; green, 2; blue, 27 }  ,draw opacity=1 ][line width=1.5]  [dash pattern={on 1.69pt off 2.76pt}]  (329.42,199.59) -- (311.1,246.64) ;
		\draw [color={rgb, 255:red, 208; green, 2; blue, 27 }  ,draw opacity=1 ][line width=1.5]  [dash pattern={on 1.69pt off 2.76pt}]  (331.2,194.93) -- (336.9,215.17) ;
		\draw [color={rgb, 255:red, 208; green, 2; blue, 27 }  ,draw opacity=1 ][line width=1.5]  [dash pattern={on 1.69pt off 2.76pt}]  (376.35,177.18) -- (367.42,219.59) ;
		\draw [line width=1.5]  [dash pattern={on 1.69pt off 2.76pt}]  (355.54,272.3) -- (366.25,224.15) ;
		\draw [line width=1.5]  [dash pattern={on 1.69pt off 2.76pt}]  (249.34,276.76) -- (233.72,217.46) ;
		\draw  [dash pattern={on 4.5pt off 4.5pt}]  (408.17,212.55) -- (407.35,334.44) ;
	
		\draw [line width=1.5]  [dash pattern={on 1.69pt off 2.76pt}]  (462.17,286.55) -- (411.72,217.46) ;

		\draw (455,335) node [anchor=north west][inner sep=0.75pt]   [align=left] {$n_7$};
			\draw (263,335) node [anchor=north west][inner sep=0.75pt]  [rotate=-359.98] [align=left] {$n_3$};
		\draw (363,335) node [anchor=north west][inner sep=0.75pt]  [rotate=-359.98] [align=left] {$n_5$};
		\draw (225,335) node [anchor=north west][inner sep=0.75pt]  [rotate=-359.98] [align=left] {$n_2$};
		\draw (330,335) node [anchor=north west][inner sep=0.75pt]  [rotate=-359.98] [align=left] {$n_4$};
		\draw (110,70.78) node [anchor=north west][inner sep=0.75pt]  [rotate=-359.98] [align=left] {$c+\delta$};
		\draw (135,208) node [anchor=north west][inner sep=0.75pt]  [rotate=-359.98] [align=left] {$\alpha_i$};
		\draw (105,255) node [anchor=north west][inner sep=0.75pt]  [rotate=-359.98] [align=left] {$\alpha_i-\frac{\delta}{2}$};
		\draw (135,324.28) node [anchor=north west][inner sep=0.75pt]  [rotate=-359.98] [align=left] {0};
		\draw (153,335) node [anchor=north west][inner sep=0.75pt]  [rotate=-359.98] [align=left] {$n_1$};
		\draw (400,335) node [anchor=north west][inner sep=0.75pt]  [rotate=-359.98] [align=left] {$n_6$};
		\draw (484,327) node [anchor=north west][inner sep=0.75pt]  [rotate=-359.98] [align=left] {$n$};
		
		\draw (161.08,29) node [anchor=north west][inner sep=0.75pt]  [rotate=-359.98] [align=left] {$x_n$};
		
	\end{tikzpicture}
\caption{$d=7$ is odd}
\label{figlio}
	
\end{figure}
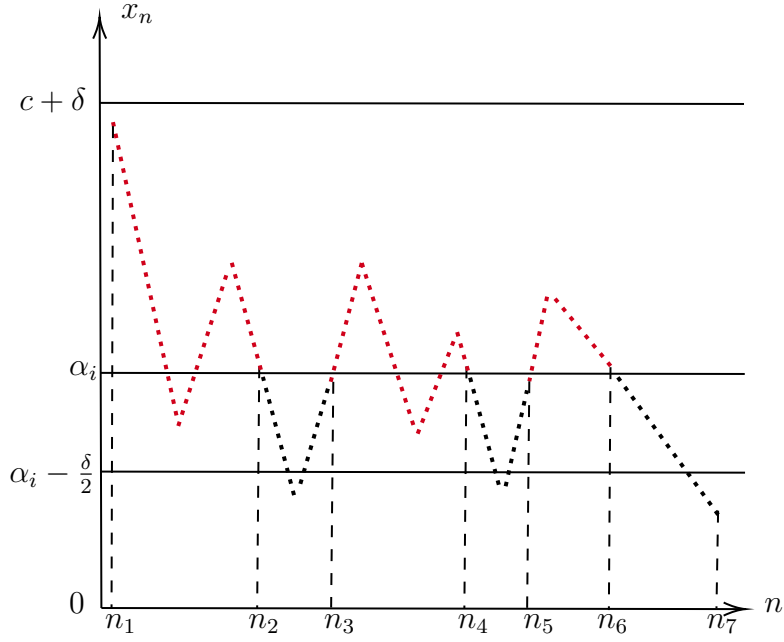

	 By \eqref{cbdmi1}-\eqref{cbdhalfmii}, \eqref{cboddmxnin}, \eqref{cbevenmxnin} and the definition of $L_i$, we can conclude that 
	\begin{align}
		\cup_{l=1}^{i}\mathcal{I}_l\subset L_i\subset \cup_{l=1}^{i}\mathcal{I}_l\cup\frac{1}{2}\mathcal{I}_{i+1},\ i=1,2,\cdots,K_1-2.\label{cbcupmilmih}
	\end{align}
Then by \eqref{cbdmii}, \eqref{cbdhalfmii} and \eqref{cbcupmilmih}, it holds that
\begin{align}
L_i\subset\cup_{l=1}^{i}\mathcal{I}_l\cup\frac{1}{2}\mathcal{I}_{i+1}\subset \cup_{l=1}^{i+1}\mathcal{I}_l\subset L_{i+1},\ i=1,2,\cdots,K_1-3.\nonumber
\end{align}
Hence,
	\begin{align}
		L_1\subset L_2\subset\cdots\subset L_{K_1-2}.\label{cbl1l2lk}
	\end{align}
	
We claim that 
	\begin{align}
		\abs{\sum_{n\in L_i}(x_{n+1}-x_n)}\leq 1,\ i=1,2,\cdots,K_1-2.\label{fixnmxnli}
	\end{align}
	Then one can prove Theorem \ref{caseb} by letting
	\begin{align}
		I_1&=L_1,\label{cbi1l1}\\
		I_{i}=L_i\setminus L_{i-1}&,i=2,\cdots,K_1-2,\label{cbiili}\\
		I_{K_1-1}&=I\setminus L_{K_1-2}.\label{cbikc-1}
	\end{align}
Indeed, by \eqref{cbl1l2lk} and the definitions of $I_i,i=1,\cdots,K_1-1$, we know that $I_i,i=1,\cdots,K_1-1,$ are disjoint subsets of $I$.
By \eqref{cbdmi1},  \eqref{cbdhalfmii}, \eqref{cbcupmilmih} and \eqref{cbi1l1}, one has for any $n\in I_1$,
\begin{align}
	x_n\in \Big[\alpha_2+\frac{\delta}{2},c+\delta\Big),\nonumber
\end{align} combined with \eqref{dealphai} and \eqref{deofVi}
we can obtain \eqref{aiinvi} for $i=1$. For any $i=2,\cdots,K_1-2$, by \eqref{cbcupmilmih}, \eqref{cbl1l2lk} and \eqref{cbiili}, we have 
\begin{align}
	I_i\subset \mathcal{I}_i\cup\frac{1}{2}\mathcal{I}_{i+1}.\nonumber
\end{align}
Then by \eqref{dealphai}, \eqref{deofVi}, \eqref{cbdmii} and \eqref{cbdhalfmii} we can obtain \eqref{aiinvi} for $i=2,\cdots,K_1-2$. By \eqref{xninbcmpe}-\eqref{cbdmii}, \eqref{cbcupmilmih} and \eqref{cbikc-1}, for any $n\in I_{K_1-1}$,
it holds that
\begin{align}
	x_n\in (b-\delta,\alpha_{K_1-2}).\nonumber
\end{align}
Then combining with \eqref{dealphai} and  \eqref{deofVi} one can obtain \eqref{aiinvi} for $i=K_1-1$.

By \eqref{fixnmxnli} and \eqref{cbi1l1}, we directly obtain \eqref{Iileq10} for $i=1$. By \eqref{cbl1l2lk}, \eqref{fixnmxnli} and \eqref{cbiili}, one has
\begin{align}
	\abs{\sum_{n\in I_i}(x_{n+1}-x_n)}\leq &\abs{\sum_{n\in L_i}(x_{n+1}-x_n)}+\abs{\sum_{n\in L_{i-1}}(x_{n+1}-x_n)}\nonumber\\
	\leq &2,\nonumber
\end{align}
from which we can obtain \eqref{Iileq10} for $i=2,\cdots,K_1-2$. By \eqref{fixnmxnli} and \eqref{cbikc-1}, we have
	\begin{align}
		\abs{\sum_{n\in I_{K_1-1}}(x_{n+1}-x_n)}=&\abs{\sum_{n\in I}(x_{n+1}-x_n)-\sum_{n\in L_{K_1-2}}(x_{n+1}-x_n)}\nonumber\\
		\leq&\abs{\sum_{n=n_{\delta}}^{n_\delta+N}(x_{n+1}-x_n)}+\abs{\sum_{n\in L_{K_1-2}}(x_{n+1}-x_n)}\nonumber\\
		\leq &\abs{x_{n_\delta+N+1}-x_{n_\delta}}+1.\nonumber
	\end{align}
	 Then we can obtain \eqref{Iileq10} for $i=K_1-1$ from \eqref{xninbcmpe}.
	 
	 By \eqref{cbdmi1}-\eqref{cbdhalfmii}, \eqref{cbcupmilmih}, \eqref{cbl1l2lk} and \eqref{cbi1l1}-\eqref{cbikc-1}, one can conclude that
	 \begin{align}
	 	\Big\{n\in I:x_n\in \Big[\alpha_i,\alpha_i+\frac{\delta}{2}\Big)\Big\}\subset \mathcal{I}_i\setminus \frac{1}{2}\mathcal{I}_i\subset I_i, \ i=1,2,\cdots,K_1-1.\label{cbmiminushalf}
	 \end{align}
	Let $i\in \{1,2,\cdots,K_1-1\}$ be fixed. Since 
	$$\Big[\alpha_i,\alpha_i+\frac{\delta}{2}\Big)\subset A,$$
	by \eqref{cbxn1xnlle} we know that there exist infinitely many different pairs of $(s,t)\in\mathbb{N}^2$ with $n_\delta<s<t$ such that 
	 \begin{align}
	 	x_s\in \Big(\alpha_i,\alpha_i+\frac{\delta}{10}\Big),\ x_t&\in \Big(\alpha_i+\frac{3\delta}{10},\alpha_i+\frac{2\delta}{5}\Big),\nonumber\\
	 	x_n\in \Big(\alpha_i,\alpha_i+\frac{\delta}{2}\Big),\ &\mathrm{for\ any }\ n\in [s,t].\label{cbxninaist}
	 \end{align}
Then by \eqref{xnbigo1}, it holds that
\begin{align}
	\frac{\delta}{5}\leq x_t-x_s\leq \abs{\sum_{n=s}^{t-1}(x_{n+1}-x_n)}
	\leq O(1)\sum_{n=s}^{t-1}\frac{1}{1+n}.\label{eover5le}
\end{align}
By \eqref{cbmiminushalf} and \eqref{cbxninaist}, we know  $\{s,s+1,\cdots,t\}\subset I_i$. Since such pairs of $(s,t)$ are infinitely many, by letting $N$ be large enough, $I_i$ would contain enough such pairs, then \eqref{Iifractionlargecasec} follows from \eqref{eover5le}.

Now we are in a position to prove \eqref{fixnmxnli}. By the definition of $L_i$, 
\begin{align}
	\abs{\sum_{n\in L_i}(x_{n+1}-x_n)}=\abs{\sum_{m=1}^{\frac{d}{2}}(x_{n_{2m}+1}-x_{n_{2m-1}})}&,\ \mathrm{if}\ d\ \mathrm{is\ even},\label{cbevsum}\\
	\abs{\sum_{n\in L_i}(x_{n+1}-x_n)}=\abs{\sum_{m=1}^{\frac{d-1}{2}}(x_{n_{2m}+1}-x_{n_{2m-1}})}&,\ \mathrm{if}\ d\ \mathrm{is\ odd}.\label{cbodsum}
\end{align}

For any $m=2,\cdots,d-1$, by \eqref{xnbigo1}, \eqref{cbcrosse} and \eqref{cbcrosso}, we have
\begin{align}
	x_{n_m}-\alpha_i=\frac{O(1)}{1+n_m},\ x_{n_m+1}-\alpha_i=\frac{O(1)}{1+n_m}.\nonumber
\end{align}
Then if $d$ is even, 
\begin{align}
	\abs{\sum_{m=1}^{\frac{d}{2}}(x_{n_{2m}+1}-x_{n_{2m-1}})}\nonumber=&\abs{\sum_{m=1}^{\frac{d}{2}}(x_{n_{2m}+1}-\alpha_i-x_{n_{2m-1}}+\alpha_i)}\\
	\leq &\abs{x_{n_d+1}-x_{n_1}}+\sum_{m=2}^{d-1}\frac{O(1)}{1+n_m},
\end{align}
if $d$ is odd,
\begin{align}
	\abs{\sum_{m=1}^{\frac{d-1}{2}}(x_{n_{2m}+1}-x_{n_{2m-1}})}\leq \abs{x_{n_{d-1}+1}-x_{n_1}}+\sum_{m=2}^{d-1}\frac{O(1)}{1+n_m}.
\end{align}
For any even $m\leq d-1$, by \eqref{xnbigo1} and \eqref{cbnmlm}, one can obtain
\begin{align}
	\frac{\delta}{4}\leq \abs{x_{l_m}-x_{n_m}}\leq \sum_{n=n_m}^{n_{m+1}}\abs{x_{n+1}-x_n}=O(1)\ln \frac{n_{m+1}}{n_m}.\nonumber
\end{align}
Hence, there exists $C_0>1$ such that for any even $m$,
\begin{align}
	n_m>C_0n_{m-2}>\cdots>C_0^{\frac{m}{2}-1}n_2>C_0^{\frac{m}{2}-1}n_1.\nonumber
\end{align}
Therefore,
\begin{align}
	\sum_{m=2}^{d-1}\frac{1}{1+n_m}\leq \frac{2C_0}{n_1(C_0-1)}.\label{cbsumfin}
\end{align}
By \eqref{xninbcmpe} and \eqref{cbevsum}-\eqref{cbsumfin}, 
\begin{align}
		\abs{\sum_{n\in L_i}(x_{n+1}-x_n)}\leq c-b+2\delta+\frac{O(1)}{1+n_1}.\nonumber
		\end{align}
Then by letting $n_\delta$ be large enough, we can obtain \eqref{fixnmxnli}.
\end{proof}

\section{Partitions in the Full-Circle Case}\label{pcd}
\begin{proof}[\textbf{Proof of Theorem \ref{cased}}]
	Let $n_\delta\in\mathbb{N}$ be large, to be specified below.
	For $i=1,2,\cdots, K_3-1,$ define 
	\begin{align}
		\mathcal{A}_i=\{n\in I: \left[x_n\right]\in [\varphi_i,\varphi_{i-1}]\},\ \tilde{\mathcal{A}}_i=\{n\in I: \left[x_n\right]\in [\tilde{\varphi}_{i-1},\tilde{\varphi}_i]\}, \label{definitionofai}
	\end{align}
	where $\varphi_i,\tilde{\varphi}_i$ are defined by \eqref{definitionofthetai} and 
	$$\varphi_0=\tilde{\varphi}_0=\frac{1}{2}.$$ 
	For any $i=1,2,\cdots,K_3-1$, define
\begin{align}
\frac{1}{2}\mathcal{A}_i&=\left\{n\in I: \left[x_n\right]\in \left[\varphi_{i}+\frac{\delta}{2},\varphi_{i-1}\right]\right\},\label{deofhalfai}\\  
\frac{1}{2}\tilde{\mathcal{A}}_i&=\left\{n\in I: \left[x_n\right]\in \left[\tilde{\varphi}_{i-1},\tilde{\varphi}_{i}-\frac{\delta}{2}\right]\right\}.\label{deofhalfai2}
\end{align}
Let  $n_\delta\gg \frac{1}{\delta}$ and by \eqref{xnbigo1} we have
\begin{align}
	\abs{x_{n+1}-x_n}\ll \delta, \ \mathrm{for\ any\ } n\geq n_\delta.\label{deltaxnllvarepsilon}
\end{align} 
Since $A=[0,1]$, we can choose $n_\delta$ so that
\begin{align}
	\left[x_{n_{\delta}}\right]\in (\varphi_1,\tilde{\varphi}_1).\label{xne}
\end{align}
	Let $i\in\{2,\cdots,K_3-2\}$ be fixed. Define
	\begin{align}
		X_1={\Bigg\{}n: & n_\delta\leq n< n_\delta+N,\ [x_n]\in\left[\varphi_1-\frac{\delta}{2},\varphi_1+\frac{\delta}{2}\right],\label{definitionofX1}\\ 
		&[x_n]\leq \varphi_1< [x_{n+1}]\ \mathrm{or}\ [x_{n+1}]\leq \varphi_1< [x_n]{\Bigg\}},\nonumber
	\end{align}
	\begin{align}
		\label{definitionofX1tilde}\tilde{X}_1={\Bigg\{}n: & n_\delta\leq n< n_\delta+N,\ [x_n]\in\left[\tilde{\varphi}_1-\frac{\delta}{2},\tilde{\varphi}_1+\frac{\delta}{2}\right],\\
		 &[x_n]< \tilde{\varphi}_1\leq [x_{n+1}]\ \mathrm{or}\ [x_{n+1}]< \tilde{\varphi}_1\leq [x_n]{\Bigg\}},\nonumber
		 	\end{align}
		 	\begin{align}
		\label{definitionofXi}X_i={\Bigg\{}n: & n_\delta\leq n< n_\delta+N,\ [x_n]\in\left[\varphi_i-\frac{\delta}{2},\varphi_i+\frac{\delta}{2}\right],\\
		&[x_n]< \varphi_i\leq [x_{n+1}]\ \mathrm{or}\ [x_{n+1}]< \varphi_i\leq [x_n]{\Bigg\}},\nonumber
	\end{align}
	\begin{align}
		\label{definitionofXitilde}\tilde{X}_i={\Bigg\{}n: & n_\delta\leq n< n_\delta+N,\ [x_n]\in\left[\tilde{\varphi}_i-\frac{\delta}{2},\tilde{\varphi}_i+\frac{\delta}{2}\right],\\
		&[x_n]\leq \tilde{\varphi}_i< [x_{n+1}]\ \mathrm{or}\ [x_{n+1}]\leq \tilde{\varphi}_i< [x_n]{\Bigg\}}\nonumber.
	\end{align}

Let 
\begin{align}
	\{n_m\}_{m=1}^d:=\{n_\delta+N\}\cup X_1\cup \tilde{X}_1\cup X_i\cup \tilde{X}_i\label{cdnm1d}
\end{align}
with 
\begin{align}
	n_1<n_2<\cdots<n_d:=n_\delta+N.\label{cddn}
\end{align}
Then by \eqref{xne}-\eqref{cddn}, one has for any $n\in(n_m,n_{m+1}]$ with odd $m\leq d-1$, 
\begin{align}
	\left[x_n\right]\in [\varphi_i,\varphi_1]\cup [\tilde{\varphi}_1,\tilde{\varphi}_i],
\end{align}
and for any $n\in(n_m,n_{m+1}]$ with even $m\leq d-1$,
\begin{align}
	\left[x_n\right]\in [0,\varphi_i)\cup (\varphi_1,\tilde{\varphi}_1)\cup (\tilde{\varphi}_i,1).\label{evenbefore}
\end{align}
Roughly speaking, $\left[x_{n_m}\right]$ enters  $[\varphi_i,\varphi_1]\cup[\tilde{\varphi}_1,\tilde{\varphi}_i]$ when $m$ is odd, and exits  $[\varphi_i,\varphi_1]\cup[\tilde{\varphi}_1,\tilde{\varphi}_i]$ when $m$ is even. 	A natural choice is to set $I_i=\mathcal{A}_i\cup\tilde{\mathcal{A}}_i$. However, 
\begin{align}
\abs{\sum_{n\in X_i\cap \mathcal{A}_i}(x_{n+1}-x_n)-\sum_{n\in \tilde{X}_i\cap\tilde{\mathcal{A}}_i}(x_{n+1}-x_n)}\nonumber
\end{align}
may be very large so that \eqref{aiminusbileq10} fails. Therefore, we take the following constructions.

We remove the redundant elements from $\{n_m\}_{m=1}^d$. For any $m_1,m_2\in \{1,2,\cdots,d\}$ with $m_1<m_2$, if for any $n\in[n_{m_1},n_{m_2}]$, one has 
\begin{align}
	\left[x_{n}\right]\in \left(\varphi_i-\frac{\delta}{2},\varphi_1+\frac{\delta}{2}\right)\cup\left(\tilde{\varphi}_1-\frac{\delta}{2},\tilde{\varphi}_i+\frac{\delta}{2}\right),\nonumber
\end{align}then we remove all $n_m\in (n_{m_1},n_{m_2})$ from $\{n_m\}_{m=1}^d$. 

After this deletion procedure,  the remaining crossing points inherit
the entering/exiting alternation.
We denote the set of the remaining points by ${\mathcal{X}}_i.$ To avoid introducing additional notation, we continue to denote the remaining elements by ${\mathcal{X}}_i=\{n_m\}_{m=1}^d$ ($n_1$ and $n_d=n_\delta+N$ will not be removed). See Figure \ref{figcdre} for example.

\begin{figure}[H]
	\centering
	
	\tikzset{every picture/.style={line width=0.75pt}}
	\scalebox{1.15}{
		\rotatebox[origin=c]{90}{%
			\begin{tikzpicture}[x=0.75pt,y=0.75pt,yscale=-1,xscale=1]
				
				\draw (336.82,95.9) -- (336.98,404.17);
				\draw (393,95.9) -- (393,404.17);
				\draw (159.83,96.16) -- (458.85,95.13);
				\draw [shift={(460.85,95.12)},rotate=179.8]
				[color={rgb,255:red,0;green,0;blue,0}][line width=0.75]
				(10.93,-3.29) .. controls (6.95,-1.4) and (3.31,-0.3)
				.. (0,0) .. controls (3.31,0.3) and (6.95,1.4)
				.. (10.93,3.29);
				
				\draw (215.58,95.97) -- (215.73,405.59);
				\draw (364.69,95.45) -- (364.85,405.08);
				\draw (159.83,96.16) -- (159.99,403.78);
				\draw [shift={(159.99,405.78)},rotate=269.97]
				[color={rgb,255:red,0;green,0;blue,0}][line width=0.75]
				(10.93,-3.29) .. controls (6.95,-1.4) and (3.31,-0.3)
				.. (0,0) .. controls (3.31,0.3) and (6.95,1.4)
				.. (10.93,3.29);
				
				\draw (243.45,95.87) -- (243.6,405.49);
				
				\draw [line width=1.5,dash pattern={on 1.69pt off 2.76pt}]
				(371.53,99.08) -- (323.53,127.07);
				\draw [line width=1.5,dash pattern={on 1.69pt off 2.76pt}]
				(351.52,142.07) -- (319.53,131.07);
				\draw [line width=1.5,dash pattern={on 1.69pt off 2.76pt}]
				(357.52,146.08) -- (320.52,159.07);
				\draw [line width=1.5,dash pattern={on 1.69pt off 2.76pt}]
				(374.52,185.08) -- (316.52,163.07);
				\draw [line width=1.5,dash pattern={on 1.69pt off 2.76pt}]
				(368.52,188.08) -- (207.51,228.05);
				\draw [line width=1.5,dash pattern={on 1.69pt off 2.76pt}]
				(253.51,247.06) -- (206.51,229.05);
				\draw [line width=1.5,dash pattern={on 1.69pt off 2.76pt}]
				(253.51,247.06) -- (221.51,255.05);
				\draw [line width=1.5,dash pattern={on 1.69pt off 2.76pt}]
				(269.5,279.06) -- (221.5,258.05);
				\draw [line width=1.5,dash pattern={on 1.69pt off 2.76pt}]
				(200.5,304.05) -- (271.5,283.06);
				\draw [line width=1.5,dash pattern={on 1.69pt off 2.76pt}]
				(346.49,338.07) -- (201.5,307.05);
				\draw [line width=1.5,dash pattern={on 1.69pt off 2.76pt}]
				(315.49,350.07) -- (353.49,340.07);
				\draw [line width=1.5,dash pattern={on 1.69pt off 2.76pt}]
				(309.49,353.07) -- (375.49,369.08);
				
				\draw [dash pattern={on 4.5pt off 4.5pt}]
				(340.86,118.07) -- (159.86,119.04);
				\draw [dash pattern={on 4.5pt off 4.5pt}]
				(332.86,136.07) -- (159.86,137.04);
				\draw [dash pattern={on 4.5pt off 4.5pt}]
				(339.02,152.57) -- (159.02,153.54);
				\draw [dash pattern={on 4.5pt off 4.5pt}]
				(333.85,170.07) -- (158.85,171.04);
				\draw [dash pattern={on 4.5pt off 4.5pt}]
				(338.85,195.07) -- (160.85,196.04);
				\draw [dash pattern={on 4.5pt off 4.5pt}]
				(246.84,218.06) -- (159.84,219.04);
				\draw [dash pattern={on 4.5pt off 4.5pt}]
				(240.84,243.05) -- (159.84,244.04);
				\draw [dash pattern={on 4.5pt off 4.5pt}]
				(246.91,248.99) -- (159.91,249.97);
				\draw [dash pattern={on 4.5pt off 4.5pt}]
				(239.84,267.05) -- (159.84,268.04);
				\draw [dash pattern={on 4.5pt off 4.5pt}]
				(246.83,290.06) -- (159.83,291.04);
				\draw [dash pattern={on 4.5pt off 4.5pt}]
				(239.83,316.05) -- (158.83,317.04);
				\draw [dash pattern={on 4.5pt off 4.5pt}]
				(332.82,336.07) -- (159.82,337.04);
				\draw [dash pattern={on 4.5pt off 4.5pt}]
				(338.82,344.07) -- (159.82,345.04);
				\draw [dash pattern={on 4.5pt off 4.5pt}]
				(332.82,359.07) -- (159.82,360.04);
				
				\draw (460.7,106.99) node
				[anchor=north west,inner sep=0.75pt,rotate=-89.89]
				{$[x_n]$};
				
				\draw (158,112) node
				[anchor=north west,inner sep=0.75pt,rotate=-89.89]
				{\scriptsize$n_1$};
				\draw (158,128) node
				[anchor=north west,inner sep=0.75pt,rotate=-89.89]
				{\scriptsize$n_2$};
				\draw (158,144) node
				[anchor=north west,inner sep=0.75pt,rotate=-89.89]
				{\scriptsize$n_3$};
				\draw (158,164.17) node
				[anchor=north west,inner sep=0.75pt,rotate=-89.89]
				{\scriptsize$n_4$};
				\draw (158,190.98) node
				[anchor=north west,inner sep=0.75pt,rotate=-89.89]
				{\scriptsize$n_5$};
				\draw (158,211.44) node
				[anchor=north west,inner sep=0.75pt,rotate=-89.89]
				{\scriptsize$n_6$};
				\draw (158,235) node
				[anchor=north west,inner sep=0.75pt,rotate=-89.89]
				{\scriptsize$n_7$};
				\draw (158,248) node
				[anchor=north west,inner sep=0.75pt,rotate=-89.89]
				{\scriptsize$n_8$};
				\draw (158,263.04) node
				[anchor=north west,inner sep=0.75pt,rotate=-89.89]
				{\scriptsize$n_9$};
				\draw (158,285.04) node
				[anchor=north west,inner sep=0.75pt,rotate=-89.89]
				{\scriptsize$n_{10}$};
				\draw (158,308) node
				[anchor=north west,inner sep=0.75pt,rotate=-89.89]
				{\scriptsize$n_{11}$};
				\draw (158,325) node
				[anchor=north west,inner sep=0.75pt,rotate=-89.89]
				{\scriptsize$n_{12}$};
				\draw (158,342) node
				[anchor=north west,inner sep=0.75pt,rotate=-89.89]
				{\scriptsize$n_{13}$};
				\draw (158,359) node
				[anchor=north west,inner sep=0.75pt,rotate=-89.89]
				{\scriptsize$n_{14}$};
				
				\draw (375.67,45) node
				[anchor=north west,inner sep=0.75pt,rotate=-89.89]
				{\scriptsize$\varphi_1+\frac{\delta}{2}$};
				\draw (342,70) node
				[anchor=north west,inner sep=0.75pt,rotate=-89.89]
				{\scriptsize$\varphi_1$};
				\draw (250,70) node
				[anchor=north west,inner sep=0.75pt,rotate=-89.89]
				{\scriptsize$\varphi_i$};
				\draw (400,73) node
				[anchor=north west,inner sep=0.75pt,rotate=-89.89]
				{\scriptsize$\frac{1}{2}$};
				\draw (222.77,45) node
				[anchor=north west,inner sep=0.75pt,rotate=-89.89]
				{\scriptsize$\varphi_i-\frac{\delta}{2}$};
				\draw (166,80) node
				[anchor=north west,inner sep=0.75pt,rotate=-89.89]
				{\scriptsize$0$};
				
			\end{tikzpicture}%
		}
	}
	
	\begin{tikzpicture}
		\draw[->,line width=1pt] (0,0.55) -- (0,-0.55);
	\end{tikzpicture}
	
	\scalebox{1.15}{
		\rotatebox[origin=c]{90}{%
			\begin{tikzpicture}[x=0.75pt,y=0.75pt,yscale=-1,xscale=1]
				
				\draw (335.82,505.98) -- (335.98,814.25);
				\draw (158.83,506.24) -- (457.85,505.21);
				\draw [shift={(459.85,505.2)},rotate=179.8]
				[color={rgb,255:red,0;green,0;blue,0}][line width=0.75]
				(10.93,-3.29) .. controls (6.95,-1.4) and (3.31,-0.3)
				.. (0,0) .. controls (3.31,0.3) and (6.95,1.4)
				.. (10.93,3.29);
				
				\draw (214.58,506.05) -- (214.73,815.67);
				\draw (363.69,505.53) -- (363.85,815.16);
				\draw (393,505.53) -- (393,815.16);
				\draw (158.83,506.24) -- (158.99,813.86);
				\draw [shift={(158.99,815.86)},rotate=269.97]
				[color={rgb,255:red,0;green,0;blue,0}][line width=0.75]
				(10.93,-3.29) .. controls (6.95,-1.4) and (3.31,-0.3)
				.. (0,0) .. controls (3.31,0.3) and (6.95,1.4)
				.. (10.93,3.29);
				
				\draw (242.45,505.95) -- (242.6,815.57);
				
				\draw [line width=1.5,dash pattern={on 1.69pt off 2.76pt}]
				(370.53,509.16) -- (322.53,537.15);
				\draw [line width=1.5,dash pattern={on 1.69pt off 2.76pt}]
				(350.52,552.15) -- (318.53,541.15);
				\draw [line width=1.5,dash pattern={on 1.69pt off 2.76pt}]
				(356.52,556.16) -- (319.52,569.15);
				\draw [line width=1.5,dash pattern={on 1.69pt off 2.76pt}]
				(373.52,595.16) -- (315.52,573.15);
				\draw [line width=1.5,dash pattern={on 1.69pt off 2.76pt}]
				(367.52,598.16) -- (206.51,638.13);
				\draw [line width=1.5,dash pattern={on 1.69pt off 2.76pt}]
				(252.51,657.14) -- (205.51,639.13);
				\draw [line width=1.5,dash pattern={on 1.69pt off 2.76pt}]
				(252.51,657.14) -- (220.51,665.13);
				\draw [line width=1.5,dash pattern={on 1.69pt off 2.76pt}]
				(268.5,689.14) -- (220.5,668.13);
				\draw [line width=1.5,dash pattern={on 1.69pt off 2.76pt}]
				(199.5,714.13) -- (270.5,693.14);
				\draw [line width=1.5,dash pattern={on 1.69pt off 2.76pt}]
				(345.49,748.15) -- (200.5,717.13);
				\draw [line width=1.5,dash pattern={on 1.69pt off 2.76pt}]
				(314.49,760.15) -- (352.49,750.15);
				\draw [line width=1.5,dash pattern={on 1.69pt off 2.76pt}]
				(308.49,763.15) -- (374.49,779.16);
				
				\draw [dash pattern={on 4.5pt off 4.5pt}]
				(339.86,528.15) -- (158.86,529.12);
				\draw [dash pattern={on 4.5pt off 4.5pt}]
				(332.85,580.15) -- (157.85,581.12);
				\draw [dash pattern={on 4.5pt off 4.5pt}]
				(337.85,605.15) -- (159.85,606.12);
				\draw [dash pattern={on 4.5pt off 4.5pt}]
				(245.84,628.14) -- (158.84,629.12);
				\draw [dash pattern={on 4.5pt off 4.5pt}]
				(239.84,653.13) -- (158.84,654.12);
				\draw [dash pattern={on 4.5pt off 4.5pt}]
				(245.83,700.14) -- (158.83,701.12);
				\draw [dash pattern={on 4.5pt off 4.5pt}]
				(238.83,726.13) -- (157.83,727.12);
				\draw [dash pattern={on 4.5pt off 4.5pt}]
				(331.82,769.15) -- (158.82,770.12);
				
				\draw (459.7,517.07) node
				[anchor=north west,inner sep=0.75pt,rotate=-89.89]
				{$[x_n]$};
				
				\draw (158,523.09) node
				[anchor=north west,inner sep=0.75pt,rotate=-89.89]
				{\scriptsize$n_1$};
				\draw (158,574.71) node
				[anchor=north west,inner sep=0.75pt,rotate=-89.89]
				{\scriptsize$n_2$};
				\draw (158,600.46) node
				[anchor=north west,inner sep=0.75pt,rotate=-89.89]
				{\scriptsize$n_3$};
				\draw (158,623.91) node
				[anchor=north west,inner sep=0.75pt,rotate=-89.89]
				{\scriptsize$n_4$};
				\draw (158,649.71) node
				[anchor=north west,inner sep=0.75pt,rotate=-89.89]
				{\scriptsize$n_5$};
				\draw (158,695.71) node
				[anchor=north west,inner sep=0.75pt,rotate=-89.89]
				{\scriptsize$n_6$};
				\draw (158,721.52) node
				[anchor=north west,inner sep=0.75pt,rotate=-89.89]
				{\scriptsize$n_7$};
				\draw (158,765.12) node
				[anchor=north west,inner sep=0.75pt,rotate=-89.89]
				{\scriptsize$n_8$};
				
				\draw (374.67,452) node
				[anchor=north west,inner sep=0.75pt,rotate=-89.89]
				{\scriptsize$\varphi_1+\frac{\delta}{2}$};
				\draw (342,479.77) node
				[anchor=north west,inner sep=0.75pt,rotate=-89.89]
				{\scriptsize$\varphi_1$};
				\draw (250,480.08) node
				[anchor=north west,inner sep=0.75pt,rotate=-89.89]
				{\scriptsize$\varphi_i$};
				\draw (400,480) node
				[anchor=north west,inner sep=0.75pt,rotate=-89.89]
				{\scriptsize$\frac{1}{2}$};
				\draw (222.77,452) node
				[anchor=north west,inner sep=0.75pt,rotate=-89.89]
				{\scriptsize$\varphi_i-\frac{\delta}{2}$};
				\draw (166,490) node
				[anchor=north west,inner sep=0.75pt,rotate=-89.89]
				{\scriptsize$0$};
				
			\end{tikzpicture}%
		}
	}
	\caption{}
	\label{figcdre}
\end{figure}

After the deletions,  by \eqref{definitionofai}-\eqref{cddn} we can conclude that (See Figure \ref{figcdxnin})
\begin{enumerate}
	\item[$\bullet$] For any $n\in [n_m,n_{m+1}]$ with odd $m\leq d-1$,
	\begin{align}
		\left[x_{n}\right]\in \left(\varphi_i-\frac{\delta}{2},\varphi_1+\frac{\delta}{2}\right)\cup\left(\tilde{\varphi}_1-\frac{\delta}{2},\tilde{\varphi}_i+\frac{\delta}{2}\right)\subset \cup_{l=1}^i(\mathcal{A}_l\cup \tilde{\mathcal{A}}_l)\cup \frac{1}{2}\mathcal{A}_{i+1}\cup\frac{1}{2}\tilde{\mathcal{A}}_{i+1}.\label{oddafter}
	\end{align}
	\item[$\bullet$] For any $n\in (n_m,n_{m+1}]$ with even $m\leq d-1$,
	\begin{align}
		\left[x_n\right]\in [0,\varphi_i)\cup (\varphi_1,\tilde{\varphi}_1)\cup (\tilde{\varphi}_i,1),\label{evenafter}
	\end{align}
	moreover, there exists $l_m\in (n_m,n_{m+1}]$  such that 
	\begin{align}
		\left[x_{l_m}\right]\in \Big[0,\varphi_i-\frac{\delta}{2}\Big]\cup \Big[\varphi_1+\frac{\delta}{2},\tilde{\varphi}_1-\frac{\delta}{2}\Big]\cup \Big[\tilde{\varphi}_i+\frac{\delta}{2},1\Big).\label{lm}
	\end{align} 
\end{enumerate}

\begin{figure}[H]
	\centering

\tikzset{every picture/.style={line width=0.75pt}} 

\begin{tikzpicture}[x=0.75pt,y=0.75pt,yscale=-1,xscale=1]
	
	\draw  (77,399) -- (593,399)(77,23) -- (77,399) -- cycle (586,394) -- (593,399) -- (586,404) (72,30) -- (77,23) -- (82,30)  ;
	\draw    (77,370) -- (589,370) ;
	\draw    (77,350) -- (589,350) ;
	\draw    (77,270) -- (589,270) ;
	\draw    (77,251) -- (589,251) ;
	\draw    (77,200) -- (589,200) ;
	\draw    (77,180) -- (589,180) ;
	\draw    (77,101) -- (589,101) ;
	\draw    (77,81) -- (589,81) ;
	\draw    (77,50) -- (589,50) ;
	\draw [line width=1.5]  [dash pattern={on 1.69pt off 2.76pt}]  (91,260) -- (103,292) ;
	\draw [line width=1.5]  [dash pattern={on 1.69pt off 2.76pt}]  (118,261) -- (105,292) ;
	\draw [line width=1.5]  [dash pattern={on 1.69pt off 2.76pt}]  (122,260) -- (134,292) ;
	\draw [line width=1.5]  [dash pattern={on 1.69pt off 2.76pt}]  (161,240) -- (137,291) ;
	\draw [line width=1.5]  [dash pattern={on 1.69pt off 2.76pt}]  (164,243) -- (206,357) ;
	\draw [line width=1.5]  [dash pattern={on 1.69pt off 2.76pt}]  (220,335) -- (210,358) ;
	\draw [line width=1.5]  [dash pattern={on 1.69pt off 2.76pt}]  (224,335) -- (236,367) ;
	\draw [line width=1.5]  [dash pattern={on 1.69pt off 2.76pt}]  (248,334) -- (239,365) ;
	\draw [line width=1.5]  [dash pattern={on 1.69pt off 2.76pt}]  (253,336) -- (275,397) ;
	\draw [line width=1.5]  [dash pattern={on 1.69pt off 2.76pt}]  (278,52) -- (299,110) ;
	\draw [line width=1.5]  [dash pattern={on 1.69pt off 2.76pt}]  (317,82) -- (303,112) ;
	\draw [line width=1.5]  [dash pattern={on 1.69pt off 2.76pt}]  (320,85) -- (332,117) ;
	\draw [line width=1.5]  [dash pattern={on 1.69pt off 2.76pt}]  (349,75) -- (335,117) ;
	\draw [line width=1.5]  [dash pattern={on 1.69pt off 2.76pt}]  (351,68) -- (386,187) ;
	\draw [line width=1.5]  [dash pattern={on 1.69pt off 2.76pt}]  (403,161) -- (389,187) ;
	\draw [line width=1.5]  [dash pattern={on 1.69pt off 2.76pt}]  (406,164) -- (424,212) ;
	\draw [line width=1.5]  [dash pattern={on 1.69pt off 2.76pt}]  (442,171) -- (439.23,177.03) -- (425,208) ;
	\draw [line width=1.5]  [dash pattern={on 1.69pt off 2.76pt}]  (444,175) -- (522,366) ;
	\draw [line width=1.5]  [dash pattern={on 1.69pt off 2.76pt}]  (550,263) -- (526,364) ;
	\draw [line width=1.5]  [dash pattern={on 1.69pt off 2.76pt}]  (553,258) -- (569,324) ;
	\draw  [dash pattern={on 4.5pt off 4.5pt}]  (93,265) -- (93,399) ;
	\draw  [dash pattern={on 4.5pt off 4.5pt}]  (145,275) -- (145,399) ;
	\draw  [dash pattern={on 4.5pt off 4.5pt}]  (172,265) -- (172,399) ;
	\draw  [dash pattern={on 4.5pt off 4.5pt}]  (258,347) -- (258,399) ;
	\draw  [dash pattern={on 4.5pt off 4.5pt}]  (294,98) -- (294,399) ;
	\draw  [dash pattern={on 4.5pt off 4.5pt}]  (340,105) -- (340,401) ;
	\draw  [dash pattern={on 4.5pt off 4.5pt}]  (360,96) -- (360,399) ;
	\draw  [dash pattern={on 4.5pt off 4.5pt}]  (411,177) -- (413,402) ;
	\draw  [dash pattern={on 4.5pt off 4.5pt}]  (437,181) -- (439,399) ;
	\draw  [dash pattern={on 4.5pt off 4.5pt}]  (444,177) -- (446,399) ;
	\draw  [dash pattern={on 4.5pt off 4.5pt}]  (482,270) -- (484,400) ;

	\draw  [dash pattern={on 4.5pt off 4.5pt}]  (569,324) -- (570,399) ;
	\draw  [dash pattern={on 4.5pt off 4.5pt}]  (161,242) -- (161,399) ;
	\draw  [dash pattern={on 4.5pt off 4.5pt}]  (269,381) -- (269,402) ;
	\draw  [dash pattern={on 4.5pt off 4.5pt}]  (351,69) -- (351,402) ;
	\draw  [dash pattern={on 4.5pt off 4.5pt}]  (423,208) -- (425,399) ;
	\draw  [dash pattern={on 4.5pt off 4.5pt}]  (465,227) -- (467,399) ;
	
	\draw (60,45) node [anchor=north west][inner sep=0.75pt]   [align=left] {\scriptsize$1$};
	\draw (30,72) node [anchor=north west][inner sep=0.75pt]   [align=left] {\scriptsize$\tilde{\varphi}_i+\frac{\delta}{2}$};
	\draw (55,92) node [anchor=north west][inner sep=0.75pt]   [align=left] {\scriptsize$\tilde{\varphi}_i$};
	\draw (55,172) node [anchor=north west][inner sep=0.75pt]   [align=left] {\scriptsize$\tilde{\varphi}_1$};
	\draw (30,190) node [anchor=north west][inner sep=0.75pt]   [align=left] {\scriptsize$\tilde{\varphi}_1-\frac{\delta}{2}$};
	\draw (30,242) node [anchor=north west][inner sep=0.75pt]   [align=left] {\scriptsize$\varphi_1+\frac{\delta}{2}$};
	\draw (55,265) node [anchor=north west][inner sep=0.75pt]   [align=left] {\scriptsize$\varphi_1$};
	\draw (55,344) node [anchor=north west][inner sep=0.75pt]   [align=left] {\scriptsize$\varphi_i$};
	\draw (30,362) node [anchor=north west][inner sep=0.75pt]   [align=left] {\scriptsize$\varphi_i-\frac{\delta}{2}$};
	\draw (60,389) node [anchor=north west][inner sep=0.75pt]   [align=left] {\scriptsize0};
	\draw (85,401) node [anchor=north west][inner sep=0.75pt]   [align=left] {\scriptsize$n_1$};
	\draw (138,401) node [anchor=north west][inner sep=0.75pt]   [align=left] {\scriptsize$n_2$};
	\draw (165,401) node [anchor=north west][inner sep=0.75pt]   [align=left] {\scriptsize$n_3$};
	\draw (248,401) node [anchor=north west][inner sep=0.75pt]   [align=left] {\scriptsize$n_4$};
	\draw (286,401) node [anchor=north west][inner sep=0.75pt]   [align=left] {\scriptsize$n_5$};
	\draw (329,401) node [anchor=north west][inner sep=0.75pt]   [align=left] {\scriptsize$n_6$};
	\draw (355,401) node [anchor=north west][inner sep=0.75pt]   [align=left] {\scriptsize$n_7$};
	\draw (403,401) node [anchor=north west][inner sep=0.75pt]   [align=left] {\scriptsize$n_8$};
	\draw (430,401) node [anchor=north west][inner sep=0.75pt]   [align=left] {\scriptsize$n_9$};
	\draw (444,401) node [anchor=north west][inner sep=0.75pt]   [align=left] {\scriptsize$n_{10}$};
	\draw (477,401) node [anchor=north west][inner sep=0.75pt]   [align=left] {\scriptsize$n_{11}$};
	\draw (563,401) node [anchor=north west][inner sep=0.75pt]   [align=left] {\scriptsize$n_{12}$};
	\draw (82,16) node [anchor=north west][inner sep=0.75pt]   [align=left] {$[x_n]$};
	\draw (600,395) node [anchor=north west][inner sep=0.75pt]   [align=left] {$n$};
	\draw (154,401) node [anchor=north west][inner sep=0.75pt]   [align=left] {\scriptsize$l_2$};
	\draw (265,401) node [anchor=north west][inner sep=0.75pt]   [align=left] {\scriptsize$l_4$};
	\draw (345,401) node [anchor=north west][inner sep=0.75pt]   [align=left] {\scriptsize$l_6$};
	\draw (420,401) node [anchor=north west][inner sep=0.75pt]   [align=left] {\scriptsize$l_8$};
	\draw (464,401) node [anchor=north west][inner sep=0.75pt]   [align=left] {\scriptsize$l_{10}$};

\end{tikzpicture}
	\caption{}
\label{figcdxnin}
\end{figure}

We emphasize that $d=d(i)$ and each $l_m=l_m(i),n_m=n_m(i)$ depend on $i$, but we leave the dependence on $i$ implicit for simplicity when there is no ambiguity. By \eqref{deltaxnllvarepsilon} and \eqref{lm}, for any even $m$,
\begin{align}
\abs{x_{l_m}-x_{n_m}}\geq \frac{\delta}{4}.\label{lmlargeepsilon}
\end{align}

Define $\mathcal{L}_i$ based on $\mathcal{X}_i=\{n_m\}^d_{m=1}$:
\begin{align}
\label{demcLi}\mathcal{L}_i=\begin{cases}
\cup_{m=1}^{\frac{d}{2}}\{n_{2m-1},n_{2m-1}+1,\cdots,n_{2m}\},\ &\mathrm{\ if \ } d \ \mathrm{\ is \ even},\\
\\
\cup_{m=1}^{\frac{d-1}{2}}\{n_{2m-1},n_{2m-1}+1,\cdots,n_{2m}\},\ & \mathrm{\ if \ } d \ \mathrm{\ is \ odd}.
\end{cases}
\end{align}
See Figure \ref{figcdevend} and Figure \ref{figcdoddd} for example, where the elements of $\mathcal{L}_i$ are marked in red. By letting $i=2,3,\cdots,K_3-2$, we can obtain $\{\mathcal{L}_i\}_{i=2}^{K_3-2}$.  

\begin{figure}[H]
	\centering

\tikzset{every picture/.style={line width=0.75pt}} 

\begin{tikzpicture}[x=0.75pt,y=0.75pt,yscale=-1,xscale=1]
	
	\draw  (81,389) -- (597,389)(81,13) -- (81,389) -- cycle (590,384) -- (597,389) -- (590,394) (76,20) -- (81,13) -- (86,20)  ;
	\draw    (81,360) -- (593,360) ;
	\draw    (81,340) -- (593,340) ;
	\draw    (81,260) -- (593,260) ;
	\draw    (81,241) -- (593,241) ;
	\draw    (81,190) -- (593,190) ;
	\draw    (81,170) -- (593,170) ;
	\draw    (81,91) -- (593,91) ;
	\draw    (81,71) -- (593,71) ;
	\draw    (81,40) -- (593,40) ;
	\draw [line width=1.5]  [dash pattern={on 1.69pt off 2.76pt}]  (95,250) -- (97,255) ;
	\draw [line width=1.5]  [dash pattern={on 1.69pt off 2.76pt}]  (165,230) -- (150.67,260) ;
	\draw [line width=1.5]  [dash pattern={on 1.69pt off 2.76pt}]  (168,233) -- (176,255) ;
	\draw [color={rgb, 255:red, 208; green, 2; blue, 27 }  ,draw opacity=1 ][line width=1.5]  [dash pattern={on 1.69pt off 2.76pt}]  (224,325) -- (214,348) ;
	\draw [color={rgb, 255:red, 208; green, 2; blue, 27 }  ,draw opacity=1 ][line width=1.5]  [dash pattern={on 1.69pt off 2.76pt}]  (228,325) -- (240,357) ;
	\draw [color={rgb, 255:red, 208; green, 2; blue, 27 }  ,draw opacity=1 ][line width=1.5]  [dash pattern={on 1.69pt off 2.76pt}]  (252,324) -- (243,355) ;
	\draw [line width=1.5]  [dash pattern={on 1.69pt off 2.76pt}]  (264,343) -- (279,387) ;
	\draw [line width=1.5]  [dash pattern={on 1.69pt off 2.76pt}]  (282,42) -- (297,87) ;
	\draw [color={rgb, 255:red, 208; green, 2; blue, 27 }  ,draw opacity=1 ][line width=1.5]  [dash pattern={on 1.69pt off 2.76pt}]  (321,72) -- (307,102) ;
	\draw [color={rgb, 255:red, 208; green, 2; blue, 27 }  ,draw opacity=1 ][line width=1.5]  [dash pattern={on 1.69pt off 2.76pt}]  (324,75) -- (336,107) ;
	\draw [line width=1.5]  [dash pattern={on 1.69pt off 2.76pt}]  (353,65) -- (344.67,89.33) ;
	\draw [line width=1.5]  [dash pattern={on 1.69pt off 2.76pt}]  (357,58) -- (366,86) ;
	\draw [color={rgb, 255:red, 208; green, 2; blue, 27 }  ,draw opacity=1 ][line width=1.5]  [dash pattern={on 1.69pt off 2.76pt}]  (407,151) -- (393,177) ;
	\draw [line width=1.5]  [dash pattern={on 1.69pt off 2.76pt}]  (417.67,171.33) -- (428,202) ;
	\draw [line width=1.5]  [dash pattern={on 1.69pt off 2.76pt}]  (438.67,176.33) -- (429,198) ;
	\draw [line width=1.5]  [dash pattern={on 1.69pt off 2.76pt}]  (451,172) -- (484,254) ;
	\draw [color={rgb, 255:red, 208; green, 2; blue, 27 }  ,draw opacity=1 ][line width=1.5]  [dash pattern={on 1.69pt off 2.76pt}]  (554,253) -- (530,354) ;
	\draw [color={rgb, 255:red, 208; green, 2; blue, 27 }  ,draw opacity=1 ][line width=1.5]  [dash pattern={on 1.69pt off 2.76pt}]  (557,248) -- (573,314) ;
	\draw  [dash pattern={on 4.5pt off 4.5pt}]  (97,255) -- (97,389) ;
	\draw  [dash pattern={on 4.5pt off 4.5pt}]  (149,265) -- (149,389) ;
	\draw  [dash pattern={on 4.5pt off 4.5pt}]  (176,255) -- (176,389) ;
	\draw  [dash pattern={on 4.5pt off 4.5pt}]  (262,337) -- (262,389) ;
	\draw  [dash pattern={on 4.5pt off 4.5pt}]  (298,88) -- (298,389) ;
	\draw  [dash pattern={on 4.5pt off 4.5pt}]  (344,95) -- (344,391) ;
	\draw  [dash pattern={on 4.5pt off 4.5pt}]  (366,86) -- (366,389) ;
	\draw  [dash pattern={on 4.5pt off 4.5pt}]  (415,167) -- (417,392) ;
	\draw  [dash pattern={on 4.5pt off 4.5pt}]  (441,171) -- (443,389) ;
	\draw  [dash pattern={on 4.5pt off 4.5pt}]  (448,167) -- (450,389) ;
	\draw  [dash pattern={on 4.5pt off 4.5pt}]  (484,255) -- (486,389) ;
	\draw  [dash pattern={on 4.5pt off 4.5pt}]  (573,314) -- (574,389) ;
	\draw [color={rgb, 255:red, 208; green, 2; blue, 27 }  ,draw opacity=1 ][line width=1.5]  [dash pattern={on 1.69pt off 2.76pt}]  (97,255) -- (107,282) ;
	\draw [color={rgb, 255:red, 208; green, 2; blue, 27 }  ,draw opacity=1 ][line width=1.5]  [dash pattern={on 1.69pt off 2.76pt}]  (122,251) -- (109,282) ;
	\draw [color={rgb, 255:red, 208; green, 2; blue, 27 }  ,draw opacity=1 ][line width=1.5]  [dash pattern={on 1.69pt off 2.76pt}]  (126,250) -- (137,279.5) ;
	\draw [color={rgb, 255:red, 208; green, 2; blue, 27 }  ,draw opacity=1 ][line width=1.5]  [dash pattern={on 1.69pt off 2.76pt}]  (140,281) -- (150,263) ;
	\draw [color={rgb, 255:red, 208; green, 2; blue, 27 }  ,draw opacity=1 ][line width=1.5]  [dash pattern={on 1.69pt off 2.76pt}]  (176,255) -- (210,347) ;
	\draw [color={rgb, 255:red, 208; green, 2; blue, 27 }  ,draw opacity=1 ][line width=1.5]  [dash pattern={on 1.69pt off 2.76pt}]  (257,326) -- (263.67,342.33) ;
	\draw [color={rgb, 255:red, 208; green, 2; blue, 27 }  ,draw opacity=1 ][line width=1.5]  [dash pattern={on 1.69pt off 2.76pt}]  (298,88) -- (303.67,104.33) ;
	\draw [color={rgb, 255:red, 208; green, 2; blue, 27 }  ,draw opacity=1 ][line width=1.5]  [dash pattern={on 1.69pt off 2.76pt}]  (339,107) -- (344,93) ;
	\draw [color={rgb, 255:red, 208; green, 2; blue, 27 }  ,draw opacity=1 ][line width=1.5]  [dash pattern={on 1.69pt off 2.76pt}]  (366,86) -- (392,177) ;
	\draw [color={rgb, 255:red, 208; green, 2; blue, 27 }  ,draw opacity=1 ][line width=1.5]  [dash pattern={on 1.69pt off 2.76pt}]  (410,154) -- (415,167) ;
	\draw [color={rgb, 255:red, 208; green, 2; blue, 27 }  ,draw opacity=1 ][line width=1.5]  [dash pattern={on 1.69pt off 2.76pt}]  (447,161) -- (439,175) ;
	\draw [color={rgb, 255:red, 208; green, 2; blue, 27 }  ,draw opacity=1 ][line width=1.5]  [dash pattern={on 1.69pt off 2.76pt}]  (447,161) -- (450,172) ;
	\draw [color={rgb, 255:red, 208; green, 2; blue, 27 }  ,draw opacity=1 ][line width=1.5]  [dash pattern={on 1.69pt off 2.76pt}]  (526,356) -- (484,255) ;
	
	\draw (62,33) node [anchor=north west][inner sep=0.75pt]   [align=left] {1};
	\draw (42,62) node [anchor=north west][inner sep=0.75pt]   [align=left] {\scriptsize$\tilde{\varphi}_i+\frac{\delta}{2}$};
	\draw (58,84) node [anchor=north west][inner sep=0.75pt]   [align=left] {\scriptsize$\tilde{\varphi}_i$};
	\draw (58,164) node [anchor=north west][inner sep=0.75pt]   [align=left] {\scriptsize$\tilde{\varphi}_1$};
	\draw (42,186) node [anchor=north west][inner sep=0.75pt]   [align=left] {\scriptsize$\tilde{\varphi}_1-\frac{\delta}{2}$};
	\draw (42,232) node [anchor=north west][inner sep=0.75pt]   [align=left] {\scriptsize${\varphi}_1+\frac{\delta}{2}$};
	\draw (58,255) node [anchor=north west][inner sep=0.75pt]   [align=left] {\scriptsize${\varphi}_1$};
	\draw (58,333) node [anchor=north west][inner sep=0.75pt]   [align=left] {\scriptsize$\varphi_i$};
	\draw (42,353) node [anchor=north west][inner sep=0.75pt]   [align=left] {\scriptsize$\varphi_i-\frac{\delta}{2}$};
	\draw (64,381) node [anchor=north west][inner sep=0.75pt]   [align=left] {\scriptsize{0}};
	\draw (89,395) node [anchor=north west][inner sep=0.75pt]   [align=left] {\scriptsize$n_1$};
	\draw (142,395) node [anchor=north west][inner sep=0.75pt]   [align=left] {\scriptsize$n_2$};
	\draw (168,395) node [anchor=north west][inner sep=0.75pt]   [align=left] {\scriptsize$n_3$};
	\draw (253,395) node [anchor=north west][inner sep=0.75pt]   [align=left] {\scriptsize$n_4$};
	\draw (290,395) node [anchor=north west][inner sep=0.75pt]   [align=left] {\scriptsize$n_5$};
	\draw (338,395) node [anchor=north west][inner sep=0.75pt]   [align=left] {\scriptsize$n_6$};
	\draw (358,395) node [anchor=north west][inner sep=0.75pt]   [align=left] {\scriptsize$n_7$};
	\draw (408,395) node [anchor=north west][inner sep=0.75pt]   [align=left] {\scriptsize$n_8$};
	\draw (432,395) node [anchor=north west][inner sep=0.75pt]   [align=left] {\scriptsize$n_9$};
	\draw (446,395) node [anchor=north west][inner sep=0.75pt]   [align=left] {\scriptsize$n_{10}$};
	\draw (481,395) node [anchor=north west][inner sep=0.75pt]   [align=left] {\scriptsize$n_{11}$};
	\draw (568,395) node [anchor=north west][inner sep=0.75pt]   [align=left] {\scriptsize$n_{12}$};
	\draw (86,6) node [anchor=north west][inner sep=0.75pt]   [align=left] {$[x_n]$};
	\draw (602,385) node [anchor=north west][inner sep=0.75pt]   [align=left] {$n$};

\end{tikzpicture}
	\caption{$d=12$ is even}
	\label{figcdevend}
\end{figure}
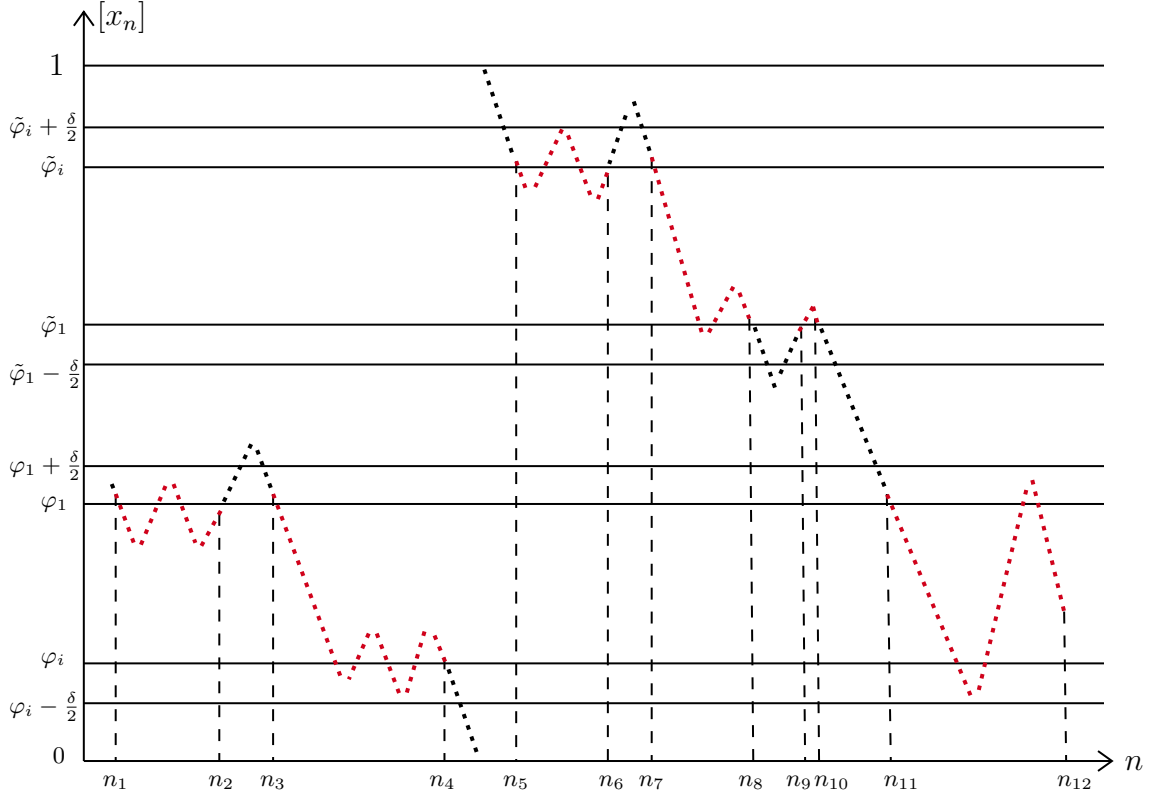

\begin{figure}[H]
	\centering

	\tikzset{every picture/.style={line width=0.75pt}} 
	
	\begin{tikzpicture}[x=0.75pt,y=0.75pt,yscale=-1,xscale=1]
		
		\draw  (81,389) -- (597,389)(81,13) -- (81,389) -- cycle (590,384) -- (597,389) -- (590,394) (76,20) -- (81,13) -- (86,20)  ;
	\draw    (81,360) -- (593,360) ;
	\draw    (81,340) -- (593,340) ;
	\draw    (81,260) -- (593,260) ;
	\draw    (81,241) -- (593,241) ;
	\draw    (81,190) -- (593,190) ;
	\draw    (81,170) -- (593,170) ;
	\draw    (81,91) -- (593,91) ;
	\draw    (81,71) -- (593,71) ;
	\draw    (81,40) -- (593,40) ;
	\draw [line width=1.5]  [dash pattern={on 1.69pt off 2.76pt}]  (95,250) -- (97,255) ;
	\draw [line width=1.5]  [dash pattern={on 1.69pt off 2.76pt}]  (165,230) -- (150.67,260) ;
	\draw [line width=1.5]  [dash pattern={on 1.69pt off 2.76pt}]  (168,233) -- (176,255) ;
	\draw [color={rgb, 255:red, 208; green, 2; blue, 27 }  ,draw opacity=1 ][line width=1.5]  [dash pattern={on 1.69pt off 2.76pt}]  (224,325) -- (214,348) ;
	\draw [color={rgb, 255:red, 208; green, 2; blue, 27 }  ,draw opacity=1 ][line width=1.5]  [dash pattern={on 1.69pt off 2.76pt}]  (228,325) -- (240,357) ;
	\draw [color={rgb, 255:red, 208; green, 2; blue, 27 }  ,draw opacity=1 ][line width=1.5]  [dash pattern={on 1.69pt off 2.76pt}]  (252,324) -- (243,355) ;
	\draw [line width=1.5]  [dash pattern={on 1.69pt off 2.76pt}]  (264,343) -- (279,387) ;
	\draw [line width=1.5]  [dash pattern={on 1.69pt off 2.76pt}]  (282,42) -- (297,87) ;
	\draw [color={rgb, 255:red, 208; green, 2; blue, 27 }  ,draw opacity=1 ][line width=1.5]  [dash pattern={on 1.69pt off 2.76pt}]  (321,72) -- (307,102) ;
	\draw [color={rgb, 255:red, 208; green, 2; blue, 27 }  ,draw opacity=1 ][line width=1.5]  [dash pattern={on 1.69pt off 2.76pt}]  (324,75) -- (336,107) ;
	\draw [line width=1.5]  [dash pattern={on 1.69pt off 2.76pt}]  (353,65) -- (344.67,89.33) ;
	\draw [line width=1.5]  [dash pattern={on 1.69pt off 2.76pt}]  (357,58) -- (366,86) ;
	\draw [color={rgb, 255:red, 208; green, 2; blue, 27 }  ,draw opacity=1 ][line width=1.5]  [dash pattern={on 1.69pt off 2.76pt}]  (407,151) -- (393,177) ;
	\draw [line width=1.5]  [dash pattern={on 1.69pt off 2.76pt}]  (417.67,171.33) -- (428,202) ;
	\draw [line width=1.5]  [dash pattern={on 1.69pt off 2.76pt}]  (438.67,176.33) -- (429,198) ;
	\draw [line width=1.5]  [dash pattern={on 1.69pt off 2.76pt}]  (451,172) -- (484,254) ;
	\draw [color={rgb, 255:red, 208; green, 2; blue, 27 }  ,draw opacity=1 ][line width=1.5]  [dash pattern={on 1.69pt off 2.76pt}]  (550,262) -- (530,354) ;
\draw [line width=1.5]  [dash pattern={on 1.69pt off 2.76pt}]  (551,258) -- (566,195) ;

	\draw  [dash pattern={on 4.5pt off 4.5pt}]  (97,255) -- (97,389) ;
	\draw  [dash pattern={on 4.5pt off 4.5pt}]  (149,265) -- (149,389) ;
	\draw  [dash pattern={on 4.5pt off 4.5pt}]  (176,255) -- (176,389) ;
	\draw  [dash pattern={on 4.5pt off 4.5pt}]  (262,337) -- (262,389) ;
	\draw  [dash pattern={on 4.5pt off 4.5pt}]  (298,88) -- (298,389) ;
	\draw  [dash pattern={on 4.5pt off 4.5pt}]  (344,95) -- (344,391) ;
	\draw  [dash pattern={on 4.5pt off 4.5pt}]  (366,86) -- (366,389) ;
	\draw  [dash pattern={on 4.5pt off 4.5pt}]  (415,167) -- (417,392) ;
	\draw  [dash pattern={on 4.5pt off 4.5pt}]  (441,171) -- (443,389) ;
	\draw  [dash pattern={on 4.5pt off 4.5pt}]  (448,167) -- (450,389) ;
	\draw  [dash pattern={on 4.5pt off 4.5pt}]  (484,255) -- (486,389) ;
	\draw  [dash pattern={on 4.5pt off 4.5pt}]  (550,262) -- (552,389) ;
	\draw  [dash pattern={on 4.5pt off 4.5pt}]  (565,198) -- (567,389) ;
	\draw [color={rgb, 255:red, 208; green, 2; blue, 27 }  ,draw opacity=1 ][line width=1.5]  [dash pattern={on 1.69pt off 2.76pt}]  (97,255) -- (107,282) ;
	\draw [color={rgb, 255:red, 208; green, 2; blue, 27 }  ,draw opacity=1 ][line width=1.5]  [dash pattern={on 1.69pt off 2.76pt}]  (122,251) -- (109,282) ;
	\draw [color={rgb, 255:red, 208; green, 2; blue, 27 }  ,draw opacity=1 ][line width=1.5]  [dash pattern={on 1.69pt off 2.76pt}]  (126,250) -- (137,279.5) ;
	\draw [color={rgb, 255:red, 208; green, 2; blue, 27 }  ,draw opacity=1 ][line width=1.5]  [dash pattern={on 1.69pt off 2.76pt}]  (140,281) -- (150,263) ;
	\draw [color={rgb, 255:red, 208; green, 2; blue, 27 }  ,draw opacity=1 ][line width=1.5]  [dash pattern={on 1.69pt off 2.76pt}]  (176,255) -- (210,347) ;
	\draw [color={rgb, 255:red, 208; green, 2; blue, 27 }  ,draw opacity=1 ][line width=1.5]  [dash pattern={on 1.69pt off 2.76pt}]  (257,326) -- (263.67,342.33) ;
	\draw [color={rgb, 255:red, 208; green, 2; blue, 27 }  ,draw opacity=1 ][line width=1.5]  [dash pattern={on 1.69pt off 2.76pt}]  (298,88) -- (303.67,104.33) ;
	\draw [color={rgb, 255:red, 208; green, 2; blue, 27 }  ,draw opacity=1 ][line width=1.5]  [dash pattern={on 1.69pt off 2.76pt}]  (339,107) -- (344,93) ;
	\draw [color={rgb, 255:red, 208; green, 2; blue, 27 }  ,draw opacity=1 ][line width=1.5]  [dash pattern={on 1.69pt off 2.76pt}]  (366,86) -- (392,177) ;
	\draw [color={rgb, 255:red, 208; green, 2; blue, 27 }  ,draw opacity=1 ][line width=1.5]  [dash pattern={on 1.69pt off 2.76pt}]  (410,154) -- (415,167) ;
	\draw [color={rgb, 255:red, 208; green, 2; blue, 27 }  ,draw opacity=1 ][line width=1.5]  [dash pattern={on 1.69pt off 2.76pt}]  (447,161) -- (439,175) ;
	\draw [color={rgb, 255:red, 208; green, 2; blue, 27 }  ,draw opacity=1 ][line width=1.5]  [dash pattern={on 1.69pt off 2.76pt}]  (447,161) -- (450,172) ;
	\draw [color={rgb, 255:red, 208; green, 2; blue, 27 }  ,draw opacity=1 ][line width=1.5]  [dash pattern={on 1.69pt off 2.76pt}]  (526,356) -- (484,255) ;
	
	\draw (62,33) node [anchor=north west][inner sep=0.75pt]   [align=left] {1};
	\draw (42,62) node [anchor=north west][inner sep=0.75pt]   [align=left] {\scriptsize$\tilde{\varphi}_i+\frac{\delta}{2}$};
	\draw (58,84) node [anchor=north west][inner sep=0.75pt]   [align=left] {\scriptsize$\tilde{\varphi}_i$};
	\draw (58,164) node [anchor=north west][inner sep=0.75pt]   [align=left] {\scriptsize$\tilde{\varphi}_1$};
	\draw (42,186) node [anchor=north west][inner sep=0.75pt]   [align=left] {\scriptsize$\tilde{\varphi}_1-\frac{\delta}{2}$};
	\draw (42,232) node [anchor=north west][inner sep=0.75pt]   [align=left] {\scriptsize${\varphi}_1+\frac{\delta}{2}$};
	\draw (58,255) node [anchor=north west][inner sep=0.75pt]   [align=left] {\scriptsize${\varphi}_1$};
	\draw (58,333) node [anchor=north west][inner sep=0.75pt]   [align=left] {\scriptsize$\varphi_i$};
	\draw (42,353) node [anchor=north west][inner sep=0.75pt]   [align=left] {\scriptsize$\varphi_i-\frac{\delta}{2}$};
	\draw (64,381) node [anchor=north west][inner sep=0.75pt]   [align=left] {\scriptsize{0}};
	\draw (89,395) node [anchor=north west][inner sep=0.75pt]   [align=left] {\scriptsize$n_1$};
	\draw (142,395) node [anchor=north west][inner sep=0.75pt]   [align=left] {\scriptsize$n_2$};
	\draw (168,395) node [anchor=north west][inner sep=0.75pt]   [align=left] {\scriptsize$n_3$};
	\draw (253,395) node [anchor=north west][inner sep=0.75pt]   [align=left] {\scriptsize$n_4$};
	\draw (290,395) node [anchor=north west][inner sep=0.75pt]   [align=left] {\scriptsize$n_5$};
	\draw (338,395) node [anchor=north west][inner sep=0.75pt]   [align=left] {\scriptsize$n_6$};
	\draw (358,395) node [anchor=north west][inner sep=0.75pt]   [align=left] {\scriptsize$n_7$};
	\draw (408,395) node [anchor=north west][inner sep=0.75pt]   [align=left] {\scriptsize$n_8$};
	\draw (432,395) node [anchor=north west][inner sep=0.75pt]   [align=left] {\scriptsize$n_9$};
	\draw (446,395) node [anchor=north west][inner sep=0.75pt]   [align=left] {\scriptsize$n_{10}$};
	\draw (481,395) node [anchor=north west][inner sep=0.75pt]   [align=left] {\scriptsize$n_{11}$};
		\draw (540,395) node [anchor=north west][inner sep=0.75pt]   [align=left] {\scriptsize$n_{12}$};
	\draw (568,395) node [anchor=north west][inner sep=0.75pt]   [align=left] {\scriptsize$n_{13}$};
	\draw (86,6) node [anchor=north west][inner sep=0.75pt]   [align=left] {$[x_n]$};
	\draw (602,385) node [anchor=north west][inner sep=0.75pt]   [align=left] {$n$};

	\end{tikzpicture}
	\caption{$d=13$ is odd}
	\label{figcdoddd}
\end{figure}
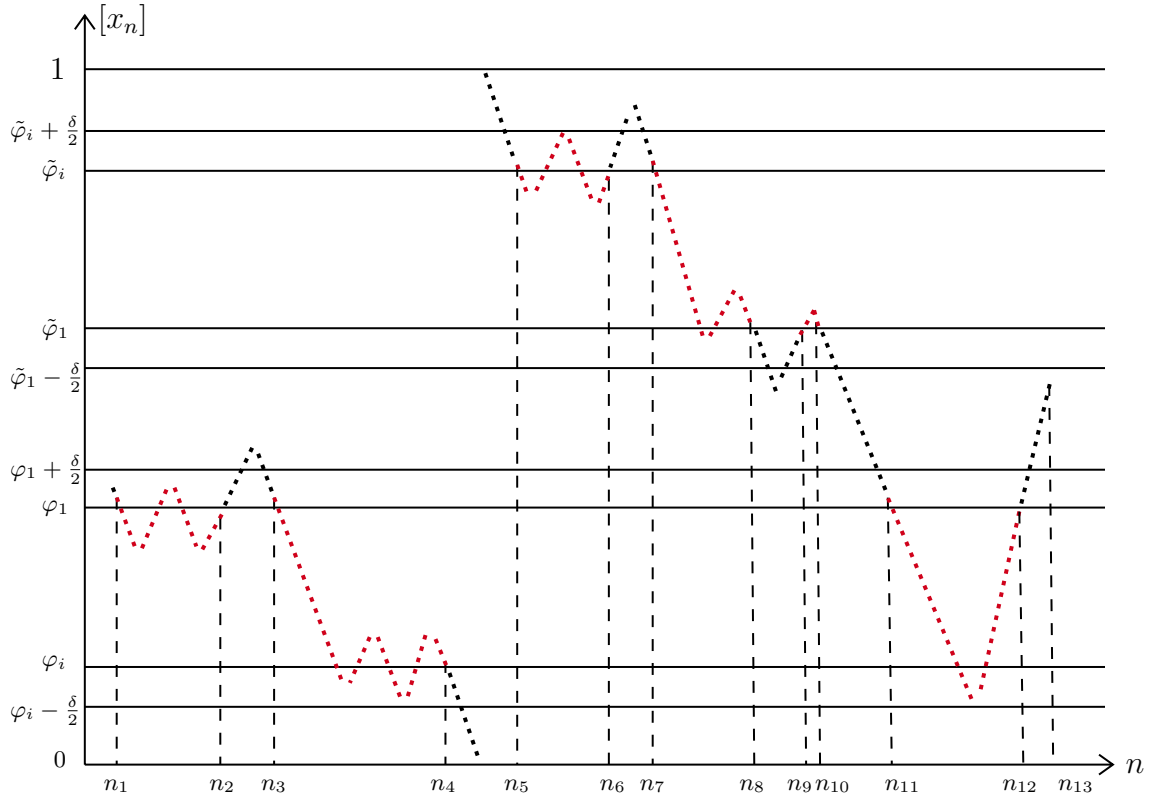

It follows from the construction that
\begin{align}
	\mathcal{L}_2\subset \mathcal{L}_3\subset\cdots\subset\mathcal{L}_{K_3-2}.\nonumber
\end{align}
Let
\begin{align}
L_i=\left\{n\in \mathcal{L}_i:\left[x_n\right]<\frac{1}{2}\right\},\ i=2,3,\cdots,K_3-2,\label{deLi}
\end{align}
and 
\begin{align}
	\tilde{L}_i=\left\{n\in\mathcal{L}_i:\left[x_n\right]>\frac{1}{2}\right\},\ i=2,3,\cdots,K_3-2,\label{deLitilde}
\end{align}
then 
\begin{align}
	&L_2\subset L_3\subset \cdots\subset L_{K_3-2},\label{subl2k2}\\
	&\tilde{L}_2\subset \tilde{L}_3\subset \cdots\subset \tilde{L}_{K_3-2}.\label{subl2k23}
\end{align}
By \eqref{deltaxnllvarepsilon}, \eqref{oddafter} and \eqref{evenafter}, we have that for any $m$,
\begin{align}
\{n_{2m-1},n_{2m-1}+1,\cdots,n_{2m}\}\subset L_i,\label{wholeLi}
\end{align}
or
\begin{align}
\{n_{2m-1},n_{2m-1}+1,\cdots,n_{2m}\}\subset \tilde{L}_i.\label{wholeLitilde}
\end{align}
By \eqref{definitionofai}-\eqref{deofhalfai2}, \eqref{oddafter}, \eqref{evenafter}, \eqref{demcLi}-\eqref{deLitilde}, we have
for any $i=2,3,\cdots,K_3-2,$
\begin{align}
\label{subsetALA}\cup_{l=2}^i\mathcal{A}_l\subset L_i\subset \cup_{l=1}^{i}\mathcal{A}_l\cup\frac{1}{2}\mathcal{A}_{i+1}
\end{align}
and
\begin{align} 
	\cup_{l=2}^i\tilde{\mathcal{A}}_l\subset \tilde{L}_i\subset\cup_{l=1}^{i}\tilde{\mathcal{A}}_l\cup\frac{1}{2}\tilde{\mathcal{A}}_{i+1}.\label{subsetALAtilde}
\end{align}

We claim that
\begin{align}
L_{i}\setminus L_{i-1}\subset \mathcal{A}_{i}\cup \mathcal{A}_{i+1},\ \tilde{L}_{i}\setminus \tilde{L}_{i-1}\subset \tilde{\mathcal{A}}_{i}\cup \tilde{\mathcal{A}}_{i+1},\ i=3,\cdots,K_3-2,\label{claim1}
\end{align}
and 
\begin{align}
\abs{\sum_{n\in L_i}(x_{n+1}-x_n)-\sum_{n\in \tilde{L}_i}(x_{n+1}-x_n)}\leq 1, \ i=2,3,\cdots,K_3-2.\label{claim2}
\end{align}
Then one can prove the theorem by letting 
\begin{align} &A_1=\mathcal{A}_1\setminus L_{K_3-2},\label{deA1}\\ \tilde{A}_1&=\tilde{\mathcal{A}}_1\setminus\left(\tilde{L}_{K_3-2}\cup A_1\right),\label{deA1tilde}\\
A_i=L_i\setminus L_{i-1}&,\ \tilde{A}_i=\tilde{L}_i\setminus \tilde{L}_{i-1},\ i=2,\cdots,K_3-2,\label{deAi}
\end{align}
and
\begin{align}
&A_{K_3-1}=\left\{n\in I:\left[x_n\right]<\varphi_{K_3-2},\ n\notin L_{K_3-2}\right\},\label{deAK-1}\\
&\tilde{A}_{K_3-1}=\left\{n\in I:\left[x_n\right]>\tilde{\varphi}_{K_3-2},\ n\notin\tilde{L}_{K_3-2}\right\},\label{deAK-1tilde}
\end{align}
where we let 
\begin{align}
	L_1=\tilde{L}_1=\emptyset\label{L1L1tilde}
\end{align}
in \eqref{deAi}.

Indeed, it follows from \eqref{definitionofai}-\eqref{deofhalfai2}, \eqref{subl2k2}, \eqref{subl2k23} and \eqref{deA1}-\eqref{deAK-1tilde} that $$A_1,A_2,\cdots,A_{K_3-1},\tilde{A}_1,\tilde{A}_2,\cdots,\tilde{A}_{K_3-1}$$ 
are disjoint. 

By \eqref{definitionofai}-\eqref{deofhalfai2}, \eqref{subsetALA}, \eqref{subsetALAtilde} and \eqref{deA1}-\eqref{deAK-1tilde},
one can obtain
\begin{align}
	I=\cup_{i=1}^{K_3-1}I_i,\nonumber
\end{align}
where $I_i=A_i\cup \tilde{A}_i$.

By \eqref{definitionofai}, \eqref{deA1} and \eqref{deA1tilde} one can obtain \eqref{aiinui} and \eqref{aiinuitilde} for $i=1$. By \eqref{definitionofai}-\eqref{deofhalfai2}, \eqref{deLi}, \eqref{deLitilde}, \eqref{subsetALA}, \eqref{subsetALAtilde} and \eqref{deAi}, one can obtain \eqref{aiinui} and \eqref{aiinuitilde} for $i=2$.
Applying \eqref{definitionofai}, \eqref{claim1}, \eqref{deAi}, one can obtain \eqref{aiinui} and \eqref{aiinuitilde} for $i=3,\cdots,K_3-2$. 
From  \eqref{deAK-1} and \eqref{deAK-1tilde} one can obtain \eqref{aiinui} and \eqref{aiinuitilde} for $i=K_3-1$.

By \eqref{claim2}, \eqref{deAi} and \eqref{L1L1tilde}, one can obtain \eqref{aiminusbileq10} for $i=2$.
By \eqref{subl2k2}, \eqref{subl2k23}, \eqref{claim2} and \eqref{deAi}, one has that 
\begin{align}
\abs{\sum_{n\in A_i}(x_{n+1}-x_n)-\sum_{n\in\tilde{A}_i}(x_{n+1}-x_n)}\leq& \abs{\sum_{n\in L_i}(x_{n+1}-x_n)-\sum_{n\in \tilde{L}_i}(x_{n+1}-x_n)}\nonumber\\
&+\abs{\sum_{n\in L_{i-1}}(x_{n+1}-x_n)-\sum_{n\in \tilde{L}_{i-1}}(x_{n+1}-x_n)}\nonumber\\
\leq &2.\nonumber
\end{align}
Then one can obtain \eqref{aiminusbileq10} for $i=3,\cdots,K_3-2$. 

In order to prove \eqref{Iifractionlarge} based on \eqref{deA1}-\eqref{deAK-1tilde}, let $i\in\{2,\cdots,K_3-2\}$ be fixed. 
Recall that the accumulation set of $([x_n])_{n\in\mathbb{N}}$ is $A=[0,1]$, and 
$$\Big[\varphi_i,\varphi_{i}+\frac{\delta}{2}\Big)\subset [0,1],$$ 
by \eqref{deltaxnllvarepsilon},
 there are infinitely many pairs of $s,t\in\mathbb{N}$ with
\begin{align}
\abs{\left[x_t\right]-\left[x_s\right]}\geq \frac{\delta}{8},\label{xsmxtg}
\end{align}
and for any $n\in [s,t]$, 
\begin{align}
	[x_n]&\in \Big[\varphi_i,\varphi_{i}+\frac{\delta}{2}\Big),\label{xsxnm2}\\
	x_n-&[x_n]=x_s-[x_s].\label{xsxnm}
\end{align}
Applying \eqref{subsetALA} and \eqref{deAi}, it holds that
\begin{align}
\mathcal{A}_i\setminus \frac{1}{2}\mathcal{A}_i\subset A_i.\nonumber
\end{align} Hence, by \eqref{definitionofai}, \eqref{deofhalfai} and \eqref{xsxnm2} one has
$\{s,s+1,\cdots,t\}\subset A_i\subset I_i$ for arbitrarily large $s>n_\delta$ and $N>t-n_\delta$.
By \eqref{xnbigo1}, \eqref{xsmxtg} and \eqref{xsxnm}, one has 
\begin{align}
\frac{\delta}{8}\leq \abs{\sum_{n=s}^{t-1}(\left[x_{n+1}\right]-\left[x_n\right])}= \abs{\sum_{n=s}^{t-1}(x_{n+1}-x_n)}\leq O(1)\sum_{n=s}^{t-1}\frac{1}{1+n}\leq O(1)\sum_{n\in I_i}\frac{1}{1+n}.\nonumber
\end{align}
Since such pairs of $s,t$ are infinitely many,  one can obtain \eqref{Iifractionlarge} by letting $N$ be large enough.

Next, we are going to prove \eqref{claim1} and \eqref{claim2}. Consider \eqref{claim1} first. Let $i\in \{3,\cdots,K_3-2\}$ be fixed.
We only show that $L_{i}\setminus L_{i-1}\subset \mathcal{A}_{i}\cup \mathcal{A}_{i+1}, $ the proof of $ \tilde{L}_{i}\setminus \tilde{L}_{i-1}\subset \tilde{\mathcal{A}}_{i}\cup \tilde{\mathcal{A}}_{i+1}$ is similar and omitted.  From \eqref{subsetALA}, we have that
\begin{align}
\cup_{l=2}^{i-1}\mathcal{A}_l\subset L_{i-1},\nonumber
\end{align}
and 
\begin{align}
L_i\subset \cup_{l=1}^i\mathcal{A}_l\cup\frac{1}{2}\mathcal{A}_{i+1}.\nonumber
\end{align}
Thus by \eqref{definitionofai} and \eqref{deofhalfai},
\begin{align}
L_i\setminus L_{i-1}\subset \mathcal{A}_1\cup \mathcal{A}_i\cup\frac{1}{2}\mathcal{A}_{i+1}\subset \mathcal{A}_1\cup \mathcal{A}_i\cup\mathcal{A}_{i+1}.\nonumber
\end{align}
Therefore, we only need to prove 
\begin{align}
	(L_i\setminus L_{i-1})\cap \mathcal{A}_1=\emptyset.\nonumber
\end{align}
Suppose that there exists $\tilde{n}$ with
\begin{align}
	\tilde{n}\in L_i\cap\mathcal{A}_1,\ \tilde{n}\notin L_{i-1}.\nonumber
\end{align}
Recall that $d$ and each $n_m,l_m$ depend on $i$. By \eqref{demcLi}, \eqref{deLi} and $\tilde{n}\notin L_{i-1}$ we know that there exist $$n_{m}=n_m(i-1)\in\mathcal{X}_{i-1},n_{m+1}=n_{m+1}(i-1)\in {\mathcal{X}}_{i-1}$$ 
with even $m$ such that $\tilde{n}\in(n_{m},n_{m+1}]$. By \eqref{cddn} and \eqref{demcLi} we know that $\tilde{n}=n_{m+1}\notin L_{i-1}$ can occur only when $\tilde{n}=n_{d}=n_{\delta}+N$ with odd $d=d(i-1)$.


If $\tilde{n}\in(n_{m},n_{m+1})$ for even $m\leq d-1$,
 by \eqref{deltaxnllvarepsilon}, \eqref{evenafter} and $\tilde{n}\in\mathcal{A}_1$, one has that for any ${n}\in(n_{m},n_{m+1}]$, 
\begin{align}
\left[x_n\right]\in (\varphi_1,\tilde{\varphi}_1).\label{fixnin}
\end{align}
By \eqref{lm} and \eqref{fixnin}, there exists $l_m\in(n_{m},n_{m+1}]$ such that 
\begin{align}
\left[x_{{l_{m}}}\right]\in \left[\varphi_1+\frac{\delta}{2},\tilde{\varphi}_1-\frac{\delta}{2}\right].\label{fixlmin}
\end{align}
Hence, by \eqref{oddafter}, \eqref{fixnin}, \eqref{fixlmin} and definitions of $\mathcal{L}_i$ and $L_i$, one obtains that 
\begin{align}
\{n_{m}+1,\cdots,n_{m+1}-1\}\cap L_i=\emptyset,\nonumber
\end{align}
which contradicts $\tilde{n}\in L_i$.
 Thus, we can obtain \eqref{claim1} if $\tilde{n}\in(n_{m},n_{m+1})$ for even $m\leq d-1$.
The proof of \eqref{claim1} for $\tilde{n}=n_d=n_{\delta}+N$ with odd $d=d(i-1)$ is similar and omitted.

Next, we prove \eqref{claim2}.
By \eqref{definitionofai}-\eqref{deofhalfai2}, \eqref{subsetALA} and \eqref{subsetALAtilde}, we have for any $n\in \cup_{i=2}^{K_3-2}\left(L_i\cup \tilde{L}_i\right)$,
\begin{align}
	[x_n]\in \left[\delta,1-\delta\right].\nonumber
\end{align}
Then by \eqref{deltaxnllvarepsilon}, for any $n\in \cup_{i=2}^{K_3-2}\left(L_i\cup \tilde{L}_i\right)$,
\begin{align}
	x_{n+1}-x_n=[x_{n+1}]-[x_n].\label{xn1xne}
\end{align}
Therefore, to prove \eqref{claim2}, we only need to prove 
\begin{align}
\label{fitopr}	\abs{\sum_{n\in L_i}(\left[x_{n+1}\right]-\left[x_n\right])-\sum_{n\in \tilde{L}_i}(\left[x_{n+1}\right]-\left[x_n\right])}\leq 1, \ i=2,3,\cdots,K_3-2.
\end{align}

Define $\sigma:\{n_m\}_{m=1}^{d-1}\to \{1,-1\}$ by
\begin{align}
\sigma(n_m)=\begin{cases}
1, &n_m\in X_1\cup X_i,\\
-1, &n_m\in \tilde{X}_1\cup \tilde{X}_i.
\end{cases}\nonumber
\end{align}
Since 
\begin{align}
\sum_{n=n_{2m-1}}^{n_{2m}}(\left[x_{n+1}\right]-\left[x_n\right])=\left[x_{n_{2m}+1}\right]-\left[x_{n_{2m-1}}\right],\nonumber
\end{align}
by \eqref{demcLi}, \eqref{wholeLi} and \eqref{wholeLitilde}, if $d$ is odd,
\begin{align}
	\abs{\sum_{n\in L_i}(\left[x_{n+1}\right]-\left[x_n\right])-\sum_{n\in \tilde{L}_i}(\left[x_{n+1}\right]-\left[x_n\right])}
	=\abs{\sum_{m=1}^{\frac{d-1}{2}}\sigma(n_{2m-1})(\left[x_{n_{2m}+1}\right]-\left[x_{n_{2m-1}}\right])},\label{transition}
\end{align}
and if $d$ is even,
\begin{align}
	\abs{\sum_{n\in L_i}(\left[x_{n+1}\right]-\left[x_n\right])-\sum_{n\in \tilde{L}_i}(\left[x_{n+1}\right]-\left[x_n\right])}
	=\abs{\sum_{m=1}^{\frac{d}{2}}\sigma(n_{2m-1})(\left[x_{n_{2m}+1}\right]-\left[x_{n_{2m-1}}\right])}.\nonumber
\end{align}
We only prove \eqref{fitopr} for the case that $d$ is odd, the proof of even case is similar and omitted.

Define the map: $\gamma:\{n_m\}_{m=1}^{d-1}\to \{\varphi_1,\varphi_i,\tilde{\varphi}_1,\tilde{\varphi}_i\}$ by 
\begin{align}
\gamma(n_{m})=\begin{cases}
\varphi_1,&\mathrm{\ if \ }n_m\in X_1,\\
\varphi_i,&\mathrm{\ if \ }n_m\in X_i,\\
\tilde{\varphi}_1,&\mathrm{\ if \ }n_m\in \tilde{X}_1,\\
\tilde{\varphi}_i,&\mathrm{\ if \ }n_m\in \tilde{X}_i.\\
\end{cases}\nonumber
\end{align}
By \eqref{definitionofX1}-\eqref{definitionofXitilde}, one has
\begin{align}
\gamma(n_m)\in [\left[x_{n_m}\right],\left[x_{n_{m}+1}\right]] \ \mathrm{or}\ \gamma(n_m)\in [\left[x_{n_m+1}\right],\left[x_{n_{m}}\right]].\nonumber
\end{align}
Then by \eqref{xnbigo1} and \eqref{xn1xne} one obtains 
\begin{align}
{\left[x_{n_m}\right]-\gamma(n_m)}= \frac{O(1)}{1+n_m},\
 {\left[x_{n_m+1}\right]-\gamma(n_m)}= \frac{O(1)}{1+n_m}.\nonumber
\end{align}
Therefore,
\begin{align}
\sigma(n_{2m-1})(\left[x_{n_{2m}+1}\right]-\left[x_{n_{2m-1}}\right])=&\sigma(n_{2m-1})(\left[x_{n_{2m}+1}\right]-\gamma(n_{2m})+\gamma(n_{2m})-\gamma(n_{2m-1})\nonumber\\
&+\gamma(n_{2m-1})-\left[x_{n_{2m-1}}\right])\nonumber\\
=&\sigma(n_{2m-1})(\gamma(n_{2m})-\gamma(n_{2m-1}))+\frac{O(1)}{1+{n_{2m-1}}}.\label{sigmaeachpart}
\end{align}

We prove the following by induction:
For any $1\leq s\leq t$, one has
\begin{align}
{\sum_{m=s}^t\sigma(n_{2m-1})(\left[x_{n_{2m}+1}\right]-\left[x_{n_{2m-1}}\right])}=g(n_{2s-1},n_{2t})+ \sum_{m=s}^{t}\frac{O(1)}{1+n_{2m-1}},\label{sigmasum}
\end{align}
where $O(1)$ is independent of $s,t$. Here
\begin{align}
g(n_{m},n_{l})=\begin{cases}
0, &\mathrm{if}\ (n_{m},n_{l})\in \left(X_1\times X_1\right)\cup \left(X_i\times X_i\right)\cup \left(\tilde{X}_1\times \tilde{X}_1\right)\cup \left(\tilde{X}_i\times \tilde{X}_i\right)\\
&\cup\left(X_1\times \tilde{X}_1\right)\cup \left(\tilde{X}_1\times X_1\right)\cup\left(X_i\times \tilde{X}_i\right)\cup\left(\tilde{X}_i\times X_i\right),\\
\varphi_i-\varphi_1, &\mathrm{if}\ (n_{m},n_{l})\in \left(X_1\times X_i\right)\cup \left({X}_1\times \tilde{X}_i\right)\cup \left(\tilde{X}_1\times \tilde{X}_i\right)\cup \left(\tilde{X}_1\times X_i\right),\\
\varphi_1-\varphi_i, &\mathrm{if}\ (n_{m},n_{l})\in \left(X_i\times X_1\right)\cup\left({X}_i\times \tilde{X}_1\right) \cup \left(\tilde{X}_i\times \tilde{X}_1\right)\cup \left(\tilde{X}_i\times X_1\right).
\end{cases}\nonumber
\end{align}
If $t-s=0$, namely, $s=t$, by \eqref{wholeLi} and \eqref{wholeLitilde} one has that 
\begin{align}
(n_{2s-1},n_{2s})\in &\left(X_1\times X_1\right)\cup \left(X_i\times X_i\right)\cup \left(\tilde{X}_1\times \tilde{X}_1\right)\cup \left(\tilde{X}_i\times \tilde{X}_i\right)\nonumber\\
&\cup\left(X_1\times X_i\right)\cup \left(X_i\times X_1\right)\cup \left(\tilde{X}_1\times \tilde{X}_i\right)\cup \left(\tilde{X}_i\times \tilde{X}_1\right).\nonumber
\end{align}

If $(n_{2s-1},n_{2s})\in \left(X_1\times X_1\right)\cup \left(X_i\times X_i\right)\cup \left(\tilde{X}_1\times \tilde{X}_1\right)\cup \left(\tilde{X}_i\times \tilde{X}_i\right)$, then $\gamma(n_{2s-1})=\gamma(n_{2s})$ and $g(n_{2s-1},n_{2s})=0$, by \eqref{sigmaeachpart} one can obtain \eqref{sigmasum}.

If $(n_{2s-1},n_{2s})\in\left(X_1\times X_i\right)\cup \left(X_i\times X_1\right)\cup \left(\tilde{X}_1\times \tilde{X}_i\right)\cup \left(\tilde{X}_i\times \tilde{X}_1\right)$, one can verify that 
\begin{align}
	\sigma(n_{2s-1})(\gamma(n_{{2s}})-\gamma(n_{2s-1}))=\begin{cases}
		\varphi_i-\varphi_1, &\mathrm{if}\ (n_{2s-1},n_{2s})\in (X_1\times X_i)\cup \left(\tilde{X}_1\times\tilde{X}_i\right),\nonumber\\
		\varphi_1-\varphi_i,&\mathrm{if}\ (n_{2s-1},n_{2s})\in (X_i\times X_1)\cup \left(\tilde{X}_i\times\tilde{X}_1\right),\nonumber
	\end{cases}
\end{align}
from which one obtains
\begin{align}
\sigma(n_{2s-1})(\gamma(n_{{2s}})-\gamma(n_{2s-1}))=g(n_{2s-1},n_{2s}),\nonumber
\end{align}
then by \eqref{sigmaeachpart} we can obtain \eqref{sigmasum}.

If $t-s=1$, namely, $t=s+1$. 
We consider the case that $n_{2s-1}\in X_1$, the remaining cases are analogous. 
Since $x_{n_m}$ enters $[\varphi_i,\varphi_1]\cup[\tilde{\varphi}_1,\tilde{\varphi}_i]$ when $m$ is odd, and exits $[\varphi_i,\varphi_1]\cup[\tilde{\varphi}_1,\tilde{\varphi}_i]$ when $m$ is even, by \eqref{definitionofX1}-\eqref{definitionofXitilde}, \eqref{oddafter} and \eqref{evenafter}, the paths of $n_{2s-1}\to n_{2s}\to n_{2s+1}\to n_{2s+2}$ can only be
\begin{figure}[H]
\begin{tikzpicture}[semithick,->,level/.style={sibling distance=120pt/#1, level distance=40pt},
                    every node/.style={align=center, inner sep=0, minimum size=25pt}]
\node {$n_{2s-1}$}
child {node {$n_{2s}$}
child {node {$n_{2s+1}$}
child {node {$n_{2s+2}$}}
}
}
;
\end{tikzpicture}
\begin{tikzpicture}[semithick,->,level/.style={sibling distance=120pt/#1, level distance=40pt},
                    every node/.style={align=center, inner sep=0, minimum size=25pt}]
\node {$X_1$}
    child {node {$X_1$}
child { node{$X_1$}
child{node{$X_1$}}
child{node{$X_i$}}
}
child { node{$\tilde{X}_1$}
child{node{$\tilde{X}_1$}}
child{node{$\tilde{X}_i$}}
}
} 
    child {
      node {$X_i$}
        child {node {$X_i$}
child{node{$X_1$}}
child{node{$X_i$}}}
        child {node {$\tilde{X}_i$}
child{node{$\tilde{X}_1$}}
child{node{$\tilde{X}_i$}}}
    }
;
\end{tikzpicture}
\caption{}
\end{figure}

To avoid repetition, we only prove the first two cases.

If $(n_{2s-1},n_{2s},n_{2s+1},n_{2s+2})\in \left(X_1\times X_1\times X_1\times X_1\right)$, then $$\gamma(n_{2s+2})=\gamma(n_{2s+1})=\gamma(n_{2s})=\gamma(n_{2s-1})=\varphi_1$$ 
and $g(n_{2s-1},n_{2s+2})=0$, by \eqref{sigmaeachpart} one obtains \eqref{sigmasum}.

If $(n_{2s-1},n_{2s},n_{2s+1},n_{2s+2})\in \left(X_1\times X_1\times X_1\times X_i\right)$, then
\begin{align}
	\gamma(n_{2s+2})=\varphi_i&,\ \gamma(n_{2s+1})=\gamma(n_{2s})=\gamma(n_{2s-1})=\varphi_1,\nonumber\\
	&\sigma(n_{2s-1})=\sigma(n_{2s+1})=1,\nonumber
\end{align}
and $g(n_{2s-1},n_{2s+2})=\varphi_i-\varphi_1$, by \eqref{sigmaeachpart} we can obtain \eqref{sigmasum}.

Suppose that we have \eqref{sigmasum} for any $t-s\leq l$, $l\geq 1$. Let $t-s=l+1$, by the same reason, we assume $n_{2s-1}\in X_1$ and omit the proofs of cases $n_{2s-1}\in X_i\cup\tilde{X}_1\cup\tilde{X}_i$. Since $x_{n_m}$ enters  $[\varphi_i,\varphi_1]\cup[\tilde{\varphi}_1,\tilde{\varphi}_i]$ when $m$ is odd, and exits  $[\varphi_i,\varphi_1]\cup[\tilde{\varphi}_1,\tilde{\varphi}_i]$ when $m$ is even, by \eqref{definitionofX1}-\eqref{definitionofXitilde}, \eqref{oddafter} and \eqref{evenafter}, the paths of $n_{2s-1}\to n_{2t-2}\to n_{2t-1}\to n_{2t}$ can be

\begin{figure}[H]
\begin{tikzpicture}[semithick,->,level/.style={sibling distance=105pt/#1, level distance=40pt},
                    every node/.style={align=center, inner sep=0, minimum size=10pt}]
\node {$n_{2s-1}$}
child {node {$n_{2t-2}$}
child {node {$n_{2t-1}$}
child {node {$n_{2t}$}}
}
}
;
\end{tikzpicture}
\begin{tikzpicture}[semithick,->,level/.style={sibling distance=105pt/#1, level distance=40pt},
                    every node/.style={align=center, inner sep=0, minimum size=15pt}]
\node {$X_1$}
    child {node {$X_1$}
child { node{$X_1$}
child{node{$X_1$}}
child{node{$X_i$}}
}
child { node{$\tilde{X}_1$}
child{node{$\tilde{X}_1$}}
child{node{$\tilde{X}_i$}}
}
} 
    child {
      node {$X_i$}
        child {node {$X_i$}
child{node{$X_1$}}
child{node{$X_i$}}}
        child {node {$\tilde{X}_i$}
child{node{$\tilde{X}_1$}}
child{node{$\tilde{X}_i$}}}
    }
child {node {$\tilde{X}_1$}
child { node{$X_1$}
child{node{$X_1$}}
child{node{$X_i$}}
}
child { node{$\tilde{X}_1$}
child{node{$\tilde{X}_1$}}
child{node{$\tilde{X}_i$}}
}
}
  child {
      node {$\tilde{X}_i$}
        child {node {$X_i$}
child{node{$X_1$}}
child{node{$X_i$}}}
        child {node {$\tilde{X}_i$}
child{node{$\tilde{X}_1$}}
child{node{$\tilde{X}_i$}}}
    }
;
\end{tikzpicture}
\caption{}
\end{figure}
Since the proofs are similar, we consider the path that $(n_{2s-1},n_{2t-2},n_{2t-1},n_{2t})\in X_1\times {X}_i\times \tilde{X}_i\times \tilde{X}_1$ as an example. See Figure \ref{ficcdfin}.

\begin{figure}[H]
\centering          

\begin{tikzpicture}[x=0.75pt,y=0.75pt,yscale=-1,xscale=1]
	
	\draw    (117,160) -- (467,160) ;
	\draw    (117,100) -- (467,100) ;
	\draw    (117,290) -- (467,290) ;
	\draw    (117,230) -- (467,230) ;
	\draw    (117,341) -- (117,13) ;
	\draw [shift={(117,11)}, rotate = 90] [color={rgb, 255:red, 0; green, 0; blue, 0 }  ][line width=0.75]    (10.93,-3.29) .. controls (6.95,-1.4) and (3.31,-0.3) .. (0,0) .. controls (3.31,0.3) and (6.95,1.4) .. (10.93,3.29)   ;
	\draw    (117,210) -- (467,210) ;
	\draw    (117,180) -- (467,180) ;
	\draw    (117,310) -- (467,310) ;
	\draw    (117,80) -- (467,80) ;
	\draw [line width=1.5]  [dash pattern={on 1.69pt off 2.76pt}]  (137.67,238.5) -- (127.67,223.5) ;
	\draw [line width=1.5]  [dash pattern={on 1.69pt off 2.76pt}]  (122.7,216.5) -- (128.7,225.5) ;
	\draw [line width=1.5]  [dash pattern={on 1.69pt off 2.76pt}]  (286,280.5) -- (319.67,339.5) ;
	\draw    (117,46) -- (467,46) ;
	\draw [line width=1.5]  [dash pattern={on 1.69pt off 2.76pt}]  (327.67,51.83) -- (381.7,172.5) ;
	\draw [line width=1.5]  [dash pattern={on 1.69pt off 2.76pt}]  (385.7,172.5) -- (398.7,148.5) ;
	\draw [line width=1.5]  [dash pattern={on 1.69pt off 2.76pt}]  (399.67,154.5) -- (405.67,185.5) ;
	\draw  [dash pattern={on 4.5pt off 4.5pt}]  (130.7,227.5) -- (130.67,342.83) ;
	\draw  [dash pattern={on 4.5pt off 4.5pt}]  (289.67,287.33) -- (289.67,342.67) ;
	\draw  [dash pattern={on 4.5pt off 4.5pt}]  (347.6,98.5) -- (347.6,340.83) ;
	\draw  [dash pattern={on 4.5pt off 4.5pt}]  (399.67,154.5) -- (399.67,341) ;
	\draw    (117,341) -- (490.67,342.33) ;
	\draw [shift={(492.67,342.33)}, rotate = 180.2] [color={rgb, 255:red, 0; green, 0; blue, 0 }  ][line width=0.75]    (10.93,-3.29) .. controls (6.95,-1.4) and (3.31,-0.3) .. (0,0) .. controls (3.31,0.3) and (6.95,1.4) .. (10.93,3.29)   ;
	
	\draw (126,5) node [anchor=north west][inner sep=0.75pt]   [align=left] {$[x_n]$};
	\draw (98,40) node [anchor=north west][inner sep=0.75pt]   [align=left] {\scriptsize 1};
	\draw (78,69) node [anchor=north west][inner sep=0.75pt]   [align=left] {\scriptsize$\tilde{\varphi}_i+\frac{\delta}{2}$};
	\draw (95,93) node [anchor=north west][inner sep=0.75pt]   [align=left] {\scriptsize$\tilde{\varphi}_i$};
	\draw (95,155) node [anchor=north west][inner sep=0.75pt]   [align=left] {\scriptsize$\tilde{\varphi}_1$};
	\draw (78,170) node [anchor=north west][inner sep=0.75pt]   [align=left] {\scriptsize$\tilde{\varphi}_1-\frac{\delta}{2}$};
	\draw (78,199) node [anchor=north west][inner sep=0.75pt]   [align=left] {\scriptsize$\varphi_1+\frac{\delta}{2}$};
	\draw (95,227) node [anchor=north west][inner sep=0.75pt]   [align=left] {\scriptsize$\varphi_1$};
	\draw (95,285) node [anchor=north west][inner sep=0.75pt]   [align=left] {\scriptsize$\varphi_i$};
	\draw (78,300) node [anchor=north west][inner sep=0.75pt]   [align=left] {\scriptsize$\varphi_i-\frac{\delta}{2}$};
	\draw (95,335) node [anchor=north west][inner sep=0.75pt]   [align=left] {\scriptsize 0};
	\draw (123,345) node [anchor=north west][inner sep=0.75pt]   [align=left] {\scriptsize$n_{2s-1}$};
	\draw (280,345) node [anchor=north west][inner sep=0.75pt]   [align=left] {\scriptsize$n_{2t-2}$};
	\draw (336,345) node [anchor=north west][inner sep=0.75pt]   [align=left] {\scriptsize$n_{2t-1}$};
	\draw (392,345) node [anchor=north west][inner sep=0.75pt]   [align=left] {\scriptsize$n_{2t}$};
	\draw (499,338) node [anchor=north west][inner sep=0.75pt]   [align=left] {\scriptsize$n$};
	\draw (193,250) node [anchor=north west][inner sep=0.75pt]   [align=left] {$\cdots$};
	
	\draw (193,345) node [anchor=north west][inner sep=0.75pt]   [align=left] {$\cdots$};

\end{tikzpicture}
\caption{}
\label{ficcdfin}
\end{figure}

Since $(n_{2s-1},n_{2t-2})\in X_1\times {X}_i$ and $(t-1)-s=l$, by the assumption that \eqref{sigmasum} holds for $t-s\leq l$, one has
\begin{align}
{\sum_{m=s}^{t-1}\sigma(n_{2m-1})(\left[x_{n_{2m}+1}\right]-\left[x_{n_{2m-1}}\right])}=&g(n_{2s-1},n_{2t-2})+ \sum_{m=s}^{t-1}\frac{O(1)}{1+n_{2m-1}}\nonumber\\
=&\varphi_i-\varphi_1+\sum_{m=s}^{t-1}\frac{O(1)}{1+n_{2m-1}}.\nonumber
\end{align}
Since $(n_{2t-1},n_{2t})\in \tilde{X}_i\times \tilde{X}_1$, by \eqref{sigmaeachpart} one has
\begin{align}
{\sigma(n_{2t-1})(\left[x_{n_{2t}+1}\right]-\left[x_{n_{2t-1}}\right])}=\varphi_1-\varphi_i+\frac{O(1)}{1+n_{2t-1}}.\nonumber
\end{align}
By $(n_{2s-1},n_{2t})\in (X_1\times \tilde{X}_1)$, one has $g(n_{2s-1},n_{2t})=0$, from which we can prove \eqref{sigmasum} for $t-s=l+1$. By induction, we have \eqref{sigmasum} for any $1\leq s\leq t$.

	For any even $m$, by \eqref{xnbigo1} and \eqref{lmlargeepsilon} one can obtain
\begin{align}
\frac{\delta}{4}\leq \abs{x_{l_m}-x_{n_m}}\leq \sum_{n=n_{m}}^{n_{m+1}}\abs{x_{n+1}-x_n}= O(1)\ln \frac{n_{m+1}}{n_m}.\nonumber
\end{align}
Hence, there exists $C_0>1$ such that for any even $m$,
\begin{align}
n_m>C_0n_{m-2}>\cdots>C_0^{\frac{m}{2}-1}n_2>C_0^{\frac{m}{2}-1}n_1.\nonumber
\end{align}
Thus, one has that 
\begin{align}
\sum_{m=1}^{\frac{d-1}{2}}\frac{1}{1+n_{2m-1}}\leq \frac{C_0}{n_1(C_0-1)}.\label{sumsmall}
\end{align}
Let $n_\delta$ be large enough, by \eqref{transition}, \eqref{sigmasum} and \eqref{sumsmall} one can obtain \eqref{fitopr}.
\end{proof}

\section{Partitions of Pr\"ufer Angles}\label{secparpru}
In this section, we assume that $k=\frac{p}{q}$ with odd $q\geq 3$. Recall that $\theta(n)$ is the modified Pr\"ufer angle determined by \eqref{ctt} under the condition \eqref{vno1}.
Denote the set of accumulation points of $(\left[\theta(nq)\right])_{n\in\mathbb{N}}$ by $B$. Set $$x_n:=\theta(nq)-nkq.$$ 
Clearly, by \eqref{demodx},
\begin{align}
	[x_n]=\left[\theta(nq)\right],\nonumber
\end{align}
and hence
$B$ is also the set of accumulation points of $(\left[x_n\right])_{n\in\mathbb{N}}$. Moreover, it holds that
\begin{align}
	x_{n+1}-x_n=\theta(nq+q)-\theta(nq)-kq.\nonumber
\end{align}
Then by \eqref{thetan+jn} one can obtain \eqref{xnbigo1}. Therefore, all the results in the previous section for $(\left[x_n\right])_{n\in\mathbb{N}}$  can be carried over, in a parallel manner, to $(\left[\theta(nq)\right])_{n\in\mathbb{N}}$. Thus, we have the following results.
\begin{corollary}\label{corangle}
	 There are four cases of $B$:
	
	\textbf{Case A.} There exists some $\varphi_0\in [0,1]$ such that $	B=\{\varphi_0\}$ or $B=\{0,1\}$.
	
	\textbf{Case B.}  $B=[b,c],\ 0\leq b<c\leq 1, b^2+(c-1)^2\neq 0$.
	
	\textbf{Case C.} $B=[0,d]\cup [e,1], \ 0\leq d<e\leq 1,d^2+(e-1)^2\neq0$.

		\textbf{Case D.} $B=[0,1].$
\end{corollary}
\begin{lemma}[\textbf{Case A}]\label{caseA}
Let $\theta(n)$ be determined by \eqref{ctt} with $B=\{\varphi_0\},\varphi_0\neq 0,1$. Let $\delta>0$ be small.  Then there exists large $n_0\in\mathbb{N}$ with $n_0\gg \frac{1}{\delta}$, such that 
for any $n\geq n_0$, one has
\begin{align}
	[\theta(nq)]\in (\varphi_0-3\delta,\varphi_0+3\delta),\label{caseAaiinuitilde}
\end{align}
and for any $N\in\mathbb{N}$ large enough,
\begin{align}
	\abs{\sum_{n=n_0}^{n_0+N}(\theta(nq+q)-\theta(nq)-kq)}\leq 10,\label{caseAaiminusbileq10}
\end{align}
\begin{align}
\sum_{n=n_0}^{n_0+N}\frac{1}{1+n}\gg \frac{1}{\delta}.\label{caseAIifractionlarge}
\end{align}
\end{lemma}

\begin{lemma}[\textbf{Case B}]\label{caseB}
  Let $\theta(n)$ be determined by \eqref{ctt} with $B=[b,c],0\leq b<c\leq 1, b^2+(c-1)^2\neq 0$. Let $\delta>0$ be small and $K_1=\left \lfloor \frac{c-b}{\delta} \right \rfloor$. Then there exists large $n_0\in\mathbb{N}$ with $n_0\gg \frac{1}{\delta}$, such that for any $N\in\mathbb{N}$ large enough, we can divide $I:=\{n\}_{n=n_0}^{n_0+N}$ into $(K_1-1)$ disjoint subsets $I_i, i=1,2,\cdots,K_1-1,$ and for any $i=1,2,\cdots,K_1-1$, one has
\begin{align}
	[\theta(nq)]\in Q_i(\delta),\ \mathrm{for\ any}\ n\in I_i,\label{aiinvitheta}
\end{align}
\begin{align}
	\abs{\sum_{n\in I_i}(\theta(nq+q)-\theta(nq)-kq)}\leq 10,\label{Iileq10theta}
\end{align}
and 
\begin{align}
\sum_{n\in I_i}\frac{1}{1+n}\gg \frac{1}{\delta},\label{Iifractionlargecasectheta}
\end{align}
where $Q_i(\delta)$ is defined by \eqref{deofVi}.
\end{lemma}
\begin{lemma}[\textbf{Case C}]\label{caseC}
 Let $\theta(n)$ be determined by \eqref{ctt} with $B=[0,d]\cup [e,1], \ 0\leq d<e\leq 1,d^2+(e-1)^2\neq0.$ Let $\delta>0$ be small and $K_2=\left \lfloor \frac{1+d-e}{\delta} \right \rfloor$. Then there exists large $n_0\in\mathbb{N}$ with $n_0\gg \frac{1}{\delta}$, such that for any $N\in\mathbb{N}$ large enough, we can divide $I:=\{n\}_{n=n_0}^{n_0+N}$ into $(K_2-1)$ disjoint subsets $I_i, i=1,2,\cdots,K_2-1,$ and for any $i=1,2,\cdots,K_2-1$, one has
\begin{align}
[\theta(nq)]\in W_i(\delta),\ \mathrm{for\ any}\ n\in I_i,\label{aiinwitheta}
\end{align}
\begin{align}
	\abs{\sum_{n\in I_i}(\theta(nq+q)-\theta(nq)-kq)}\leq 10,\label{Iileq10casedtheta}
\end{align}
and 
\begin{align}
\sum_{n\in I_i}\frac{1}{1+n}\gg \frac{1}{\delta},\label{Iifractionlargecasedtheta}
\end{align}
where $W_i(\delta)$ is defined by \eqref{deofWi}.
\end{lemma}

\begin{lemma}[\textbf{Case D}]\label{caseD}
	Let $\theta(n)$ be determined by \eqref{ctt} with $B=[0,1]$. Let $\delta>0$ be small and $K_3=\left \lfloor \frac{1}{2\delta} \right \rfloor$. Then there exists large $n_0\in\mathbb{N}$ with $n_0\gg \frac{1}{\delta}$, such that for any  $N\in\mathbb{N}$ large enough, we can divide $I:=\{n\}_{n=n_0}^{n_0+N}$ into $2(K_3-1)$ disjoint subsets: $A_i, \tilde{{A}}_i,\ i=1,2,\cdots,K_3-1,$ and
	\begin{enumerate}
		\item For any $i=1,2,\cdots,K_3-1$, one has
		\begin{align}
			[\theta(nq)]\in U_i(\delta),\ \mathrm{for\ any}\ n\in A_i,\label{aiinuitheta}
		\end{align}
		and 
		\begin{align}
			[\theta(nq)]\in \tilde{U}_i(\delta),\ \mathrm{for\ any}\ n\in \tilde{A}_i,\label{aiinuitildetheta}
		\end{align}
		where $U_i(\delta),\tilde{U}_i(\delta)$ are defined by \eqref{definitionofui}.
		\item For any $i=2,\cdots,K_3-2$, one has
		\begin{align}
			\abs{\sum_{n\in A_i}(\theta(nq+q)-\theta(nq)-kq)-\sum_{n\in \tilde{A}_i}(\theta(nq+q)-\theta(nq)-kq)}\leq 10,\label{aiminusbileq10theta}
		\end{align}
		and 
		\begin{align}
			\sum_{n\in I_i}\frac{1}{1+n}\gg \frac{1}{\delta},\label{Iifractionlargetheta}
		\end{align}
		where $I_i=A_i\cup \tilde{{A}}_i$.
	\end{enumerate} 
\end{lemma}

\section{The Sharp Bound in the Odd-Denominator Case}\label{sec8}
\begin{proof}[\textbf{Proof of Theorem \ref{main}}]Let $a=a(V)$ be defined by
	\begin{align}
		a(V)=\limsup_{n\to \infty}\abs{nV(n)}.\nonumber
	\end{align}
	To prove the result, we show that for any small $\varepsilon>0$, there exists $n_0 \in\mathbb{N}$ such that as $N\to\infty$,
	\begin{align}\label{mmmdelta01}
	\frac{\sum_{m=0}^{N-1}\sum_{j=0}^{q-1}V(n_0q+mq+j)\sin2\pi\theta(n_0q+mq+j)}{\sum_{m=0}^{N-1}\frac{a+\varepsilon}{1+n_0q+mq}}	\leq
		qB_q+O(\varepsilon).
	\end{align}
In the following, we always assume that $\varepsilon>0$ is sufficiently small and may change even in the same formula. 
	Indeed, if one has \eqref{mmmdelta01}, then by \eqref{lr} we obtain that for any large enough $N$,
	\begin{align}
		\ln R(n_0q+Nq)^2=&\ln R(n_0q)^2
		-\sum_{m=0}^{N-1}\sum_{j=0}^{q-1}\frac{V(n_0q+mq+j)}{\sin\pi k}\sin2\pi\theta(n_0q+mq+j)+O(1)\nonumber\\
		\geq &O(1)-\left(\frac{qB_q}{\sin\pi k}+\varepsilon\right)\sum_{m=0}^{N-1}\frac{a+\varepsilon}{1+n_0q+mq}\nonumber\\
		=&O(1)-\frac{(a+\varepsilon)(B_q+\varepsilon)}{\sin\pi k}\ln (n_0q+Nq).\label{lnRn0m0qlargethan}
	\end{align}
	If $a<\frac{\sin\pi k}{B_q}$, then 
	\begin{align}
		\frac{(a+\varepsilon)(B_q+\varepsilon)}{\sin\pi k}< 1\nonumber
	\end{align}
		for any small enough $\varepsilon>0$. 
Hence, by \eqref{lnRn0m0qlargethan} one has that for any large enough $N$,
\begin{align}
R(n_0q+Nq)^2\geq \frac{1}{n_0q+Nq}.\nonumber
\end{align}
	This implies that $R(\cdot)\notin \ell^2(\mathbb{N})$. By \eqref{yne} and \eqref{rne} we have that $u(\cdot,E)\notin \ell^2(\mathbb{N})$, which implies that $E$ is not an eigenvalue of $H$.
	
Let $B$ be the set of accumulation points of $(\left[\theta(nq)\right])_{n\in\mathbb{N}}$. We prove \eqref{mmmdelta01} in each case. We emphasize that for any large $n$,
\begin{align}
	\abs{V(n)}\leq \frac{a+\varepsilon}{1+n}\label{fvc}.
\end{align}

 \textbf{Case A.}
		If $B=\{\varphi_0\}$ with $\varphi_0\in \{\frac{j}{2q}\}_{j=0}^{2q}$ or $B=\{0,1\}$, then \eqref{mmmdelta01} follows from \eqref{fvc} and Lemma \ref{lemmafixedpoints} with $\delta$ being replaced by $\varepsilon$.
		
	If $B=\{\varphi_0\}$ with $\varphi_0\notin \{\frac{j}{2q}\}_{j=0}^{2q}$, 
	let 
		\begin{align}
		0<\varepsilon\ll \min\left\{\abs{\varphi_0-\frac{j}{2q}}:j=0,1,\cdots,2q\right\}.\nonumber
	\end{align} 
	Then we can obtain
	\eqref{mmmdelta01} from \eqref{fvc}, Corollary \ref{zhongyaoyinli} and Lemma \ref{caseA} by letting $\delta=\varepsilon$.
	
 \textbf{Case B \& Case C.} We only give the proof of \textbf{Case B}, the proof of \textbf{Case C} is similar and omitted.
In this case, $B=[b,c],$ where $0\leq b<c\leq 1, b^2+(c-1)^2\neq 0$. Let $0<\tilde{\varepsilon}\ll\varepsilon\ll1$ and $K_1=\left \lfloor \frac{c-b}{\tilde{\varepsilon}} \right \rfloor$. Then by Lemma \ref{caseB}, for any large enough $N$, we can divide $I=\{n\}_{n=n_0}^{n_0+N}$ into $(K_1-1)$ disjoint subsets $I_i,i=1,\cdots,K_1-1$, such that \eqref{aiinvitheta}-\eqref{Iifractionlargecasectheta} hold with $\delta=\tilde{\varepsilon}$. 

To satisfy the requirement 
	\begin{align}
	\tilde{\varepsilon}=\delta\ll \min\left\{\abs{\varphi_0-\frac{j}{2q}}:j=0,1,\cdots,2q\right\}\label{vedllm}
\end{align} 
in Corollary \ref{zhongyaoyinli}, we need to divide sets $I_1,I_2,\cdots,I_{K_1-2}$ and $I_{K_1-1}$ into two classes. 
The first class consists of those whose corresponding sequences $[\theta(nq)]$ are close to 
$\frac{j}{2q}$ for some $j=0,\cdots,2q$, while the second class consists of those that are relatively far from these points.

Let
\begin{align}
	X=\left\{1\leq i\leq K_1-1: Q_i(\tilde{\varepsilon})\subset\cup_{j=0}^{2q} \left(\frac{j}{2q}-\varepsilon,\frac{j}{2q}+\varepsilon\right)_{\mathrm{mod}}\right\},\nonumber
\end{align}
where $Q_i(\tilde{\varepsilon})$ is defined by \eqref{deofVi}.

Then by \eqref{aiinvitheta}, for any $n\in \cup_{i\in X} I_i$, one has that 
\begin{align}
	[\theta(nq)]\in \cup_{j=0}^{2q}{\left(\frac{j}{2q}-\varepsilon,\frac{j}{2q}+\varepsilon\right)_{\mathrm{mod}}},\label{fiqix}
\end{align}
and for any $i\in Y:=\{1,2,\cdots,K_1-1\}\setminus X$,  
\begin{align}
	Q_i(\tilde{\varepsilon}) \not\subset\cup_{j=0}^{2q} \left(\frac{j}{2q}-\varepsilon,\frac{j}{2q}+\varepsilon\right)_{\mathrm{mod}}.\label{fiqite}
\end{align}
See Figure \ref{figfin} for example, where $\alpha_i=\alpha_i(\tilde{\varepsilon})$ is defined by \eqref{dealphai}.

\begin{figure}[H]

\tikzset{every picture/.style={line width=0.75pt}} 

\begin{tikzpicture}[x=0.75pt,y=0.75pt,yscale=-1,xscale=1]

	\draw    (130.67,341) -- (528.67,341) ;
	\draw [shift={(530.67,341)}, rotate = 180] [color={rgb, 255:red, 0; green, 0; blue, 0 }  ][line width=0.75]    (10.93,-3.29) .. controls (6.95,-1.4) and (3.31,-0.3) .. (0,0) .. controls (3.31,0.3) and (6.95,1.4) .. (10.93,3.29)   ;
	\draw    (130.73,201.2) -- (130.67,23) ;
	\draw [shift={(130.67,21)}, rotate = 89.98] [color={rgb, 255:red, 0; green, 0; blue, 0 }  ][line width=0.75]    (10.93,-3.29) .. controls (6.95,-1.4) and (3.31,-0.3) .. (0,0) .. controls (3.31,0.3) and (6.95,1.4) .. (10.93,3.29)   ;
	\draw [color={rgb, 255:red, 208; green, 2; blue, 27 }  ,draw opacity=1 ]   (130.67,100) -- (522.67,99) ;
	\draw [color={rgb, 255:red, 208; green, 2; blue, 27 }  ,draw opacity=1 ]   (130.67,313) -- (522.67,312) ;
	\draw    (130.67,183) -- (522.67,182) ;
	\draw    (130.67,231) -- (522.67,230) ;
	\draw    (130.67,113) -- (522.67,112) ;
	\draw    (130.67,83) -- (522.67,82) ;
	\draw    (130.67,302) -- (522.67,301) ;
	\draw   (130.78,195.06) -- (130.73,201.2) -- (137.83,202.63) -- (123.58,205.23) -- (137.78,208.08) -- (123.53,210.69) -- (137.74,213.54) -- (123.49,216.15) -- (137.69,219) -- (123.44,221.61) -- (130.54,223.03) -- (130.49,229.17) ;
	\draw    (130.49,229.17) -- (130.67,341) ;
	\draw    (130.67,327) -- (522.67,326) ;
	\draw (90,91) node [anchor=north west][inner sep=0.75pt]   [align=left] {\scriptsize$\frac{j_0}{2q}+\varepsilon$};
	\draw (80,107) node [anchor=north west][inner sep=0.75pt]   [align=left] {\scriptsize $\alpha_{i_0}+3\tilde{\varepsilon}$};
	\draw (80,175) node [anchor=north west][inner sep=0.75pt]   [align=left] {\scriptsize $\alpha_{i_0}-3\tilde{\varepsilon}$};
	\draw (80,224) node [anchor=north west][inner sep=0.75pt]   [align=left] {\scriptsize $\alpha_{i_1}+3\tilde{\varepsilon}$};
	\draw (80,293) node [anchor=north west][inner sep=0.75pt]   [align=left] {\scriptsize $\alpha_{i_1}-3\tilde{\varepsilon}$};
	\draw (90,305) node [anchor=north west][inner sep=0.75pt]   [align=left] {\scriptsize$\frac{j_0}{2q}-\varepsilon$};
	\draw (70,320) node [anchor=north west][inner sep=0.75pt]   [align=left] {\scriptsize $\alpha_{{i_1+1}}-3\tilde{\varepsilon}$};
	\draw (70,76) node [anchor=north west][inner sep=0.75pt]   [align=left] {\scriptsize $\alpha_{{i_0-1}}+3\tilde{\varepsilon}$};
	\draw (139,11) node [anchor=north west][inner sep=0.75pt]   [align=left] {$[\theta(nq)]$};
	\draw (539,337) node [anchor=north west][inner sep=0.75pt]   [align=left] {$n$};

\end{tikzpicture}
\caption{$\{i_0,i_0+1,\cdots,i_1\}\subset X$ and $\{i_0-1,i_1+1\}\subset Y$.}
\label{figfin}
\end{figure}
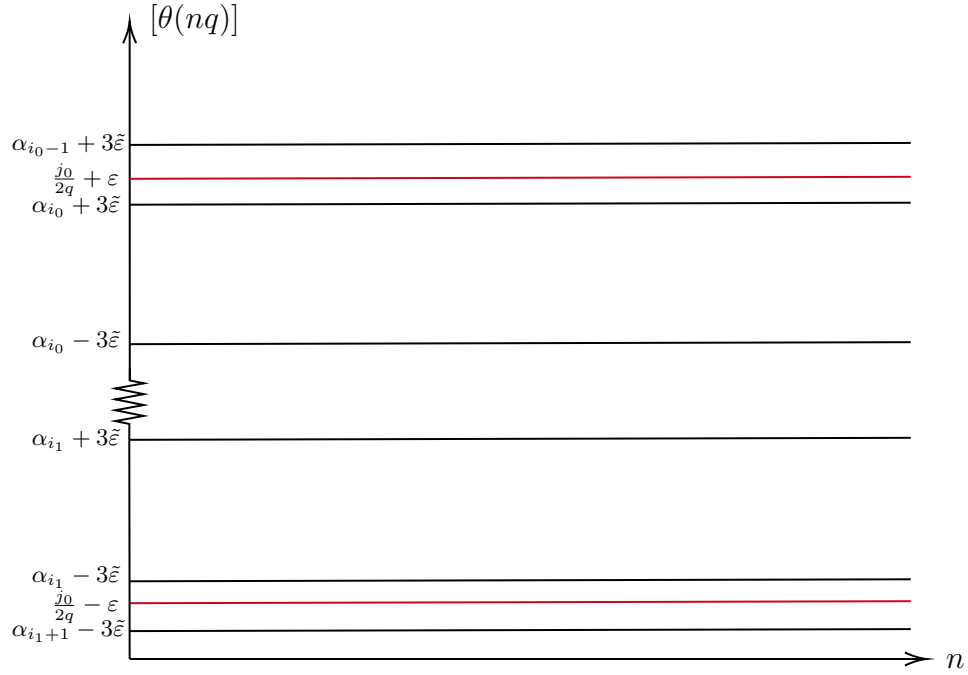
Then by  \eqref{fvc}, \eqref{fiqix} and Lemma \ref{lemmafixedpoints}, one can obtain
\begin{align}
	\frac{\sum_{n\in \cup_{i\in X} I_i}\sum_{j=0}^{q-1}V(nq+j)\sin2\pi\theta(nq+j)}{\sum_{n\in \cup_{i\in X} I_i}\frac{a+\varepsilon}{1+nq}}\leq qB_q+O(\varepsilon).\label{caseBXi1}
\end{align}

Let $i\in Y$, 
by \eqref{deofVi}, \eqref{fiqite} and $\tilde{\varepsilon}\ll \varepsilon$,
we can conclude that
\begin{align}
	\min\left\{\abs{\alpha_i-\frac{j}{2q}}:j=0,1,\cdots,2q\right\}\geq \frac{\varepsilon}{2}\gg\tilde{\varepsilon},\nonumber
\end{align}
where $\alpha_i=\alpha_i(\tilde{\varepsilon})$ is defined by \eqref{dealphai}. 
Therefore, by \eqref{fvc}, Corollary \ref{zhongyaoyinli} and Lemma \ref{caseB}, for any $i\in Y$, we have
\begin{align}
	\frac{\sum_{n\in I_i}\sum_{j=0}^{q-1}V(nq+j)\sin2\pi\theta(nq+j)}{\sum_{n\in I_i}\frac{a+\varepsilon}{1+nq}}\leq qB_q+O(\tilde{\varepsilon}),\nonumber
\end{align}
and hence,
\begin{align}
	\frac{\sum_{n\in \cup_{i\in Y} I_i}\sum_{j=0}^{q-1}V(nq+j)\sin2\pi\theta(nq+j)}{\sum_{n\in \cup_{i\in Y} I_i}\frac{a+\varepsilon}{1+nq}}\leq qB_q+O(\tilde{\varepsilon}).\label{caseBYi1}
\end{align}
By \eqref{caseBXi1} and \eqref{caseBYi1} one can obtain \eqref{mmmdelta01}.

\textbf{Case D.}
 In this case, $B=[0,1]$. Let $0<\tilde{\varepsilon}\ll\varepsilon\ll1$ and $K_3=\lfloor \frac{1}{2\tilde{\varepsilon}}\rfloor$. Then by Lemma \ref{caseD}, for any large enough $N$, we can divide $I=\{n\}_{n=n_0}^{n_0+N}$ into $2(K_3-1)$ disjoint subsets $A_i,\tilde{A}_i,i=1,\cdots,K_3-1$, such that \eqref{aiinuitheta}-\eqref{Iifractionlargetheta} hold with $\delta=\tilde{\varepsilon}$. 
		
		Denote by 
\begin{align}
X=\left\{1\leq i\leq K_3-1: U_i(\tilde{\varepsilon})\subset \cup_{j=0}^q\left(\frac{j}{2q}-\varepsilon,\frac{j}{2q}+\varepsilon\right)_{\mathrm{mod}}\right\},\nonumber
\end{align}
where 
 $U_i(\tilde{\varepsilon})$ is defined by \eqref{definitionofui},
and let $Y:=\{1,\cdots,K_3-1\}\setminus X$. We emphasize that because of $\tilde{\varepsilon}\ll\varepsilon$, by \eqref{definitionofui}, one has
\begin{align}
	U_1(\tilde{\varepsilon})\subset \left(\frac{1}{2}-\varepsilon,\frac{1}{2}+\varepsilon\right)_{\mathrm{mod}}\nonumber
\end{align}
and
\begin{align}
	U_{K_3-1}(\tilde{\varepsilon})\subset \left(-\varepsilon,\varepsilon\right)_{\mathrm{mod}}\nonumber,
\end{align}
then
	$1\in X,\ K_3-1\in X.$
Therefore,
\begin{align}
	1\notin Y,\ K_3-1\notin Y.\label{fiyno}
\end{align}
By \eqref{definitionofui} one obtains that $$U_i(\tilde{\varepsilon})\subset \left(\frac{j}{2q}-\varepsilon,\frac{j}{2q}+\varepsilon\right)_{\mathrm{mod}}$$ 
if and only if 
\begin{align}
\tilde{U}_i(\tilde{\varepsilon})\subset \left(\frac{2q-j}{2q}-\varepsilon,\frac{2q-j}{2q}+\varepsilon\right)_{\mathrm{mod}}.\nonumber
\end{align}
By \eqref{aiinuitheta} and \eqref{aiinuitildetheta} (where $\delta$ is replaced by $\tilde{\varepsilon}$), for any $n\in \cup_{i\in X} I_i=\cup_{i\in X}\left(A_i\cup \tilde{A}_i\right)$, one has that 
\begin{align}
[\theta(nq)]\in \cup_{j=0}^{2q}{\left(\frac{j}{2q}-\varepsilon,\frac{j}{2q}+\varepsilon\right)_{\mathrm{mod}}}.\nonumber
\end{align}
Then applying Lemma \ref{lemmafixedpoints} and by \eqref{fvc}, one can obtain \eqref{caseBXi1}, where $I_i=A_i\cup\tilde{A}_i$.

Let $i\in Y$, then $U_i(\tilde{\varepsilon})$ and $\tilde{U}_i(\tilde{\varepsilon})$ are not subsets of  $\cup_{j=0}^{2q} \left(\frac{j}{2q}-\varepsilon,\frac{j}{2q}+\varepsilon\right)_{\mathrm{mod}}$, by \eqref{definitionofui} and $\tilde{\varepsilon}\ll \varepsilon$,
we can conclude that
\begin{align}
\min\left\{\abs{\varphi_i-\frac{j}{2q}}:j=0,1,\cdots,2q\right\}\geq \frac{\varepsilon}{2}\gg\tilde{\varepsilon},\nonumber
\end{align}
and
\begin{align}
\min\left\{\abs{\tilde{\varphi}_i-\frac{j}{2q}}:j=0,1,\cdots,2q\right\}\geq \frac{\varepsilon}{2}\gg\tilde{\varepsilon},\nonumber
\end{align}
where $\varphi_i=\varphi_i(\tilde{\varepsilon})$ and $\tilde{\varphi}_i=\tilde{\varphi}_i(\tilde{\varepsilon})$ are defined by \eqref{definitionofthetai}. Then by \eqref{fiyno}, $Y\subset\{2,3,\cdots,K_3-2\}$.
Hence, applying \eqref{fvc}, Lemma  \ref{lemmaduichencancel} and  Lemma \ref{caseD}, one can obtain for any $i\in Y$ that
\begin{align}
\frac{\sum_{n\in A_i}\sum_{j=0}^{q-1}V(nq+j)\sin2\pi\theta(nq+j)+\sum_{n\in \tilde{A}_i}\sum_{j=0}^{q-1}V(nq+j)\sin2\pi\theta(nq+j)}{\sum_{n\in A_i}\frac{a+\varepsilon}{1+nq}+\sum_{n\in \tilde{A}_i}\frac{a+\varepsilon}{1+nq}}\nonumber\\
\leq qB_q+O(\tilde{\varepsilon}).\nonumber
\end{align}
Namely,
\begin{align}
\frac{\sum_{n\in I_i}\sum_{j=0}^{q-1}V(nq+j)\sin2\pi\theta(nq+j)}{\sum_{n\in I_i}\frac{a+\varepsilon}{1+nq}}\leq qB_q+O(\tilde{\varepsilon}).\nonumber
\end{align}
Hence one has \eqref{caseBYi1}.
By \eqref{caseBXi1} and \eqref{caseBYi1} one can obtain \eqref{mmmdelta01}.
\end{proof}

		\section{Critical Cases}\label{sec9}
	In this section,  we consider in turn the cases where $k=\frac{p}{q}$ with $q\geq 3$ odd, where $q\geq 2$ even, and where $k$ is irrational.

	Suppose that $\{N_n\}_{n=1}^\infty\subset\mathbb{N}$ with $N_n>0\ (n\geq 1)$ are given. Define $\{T_n\}_{n=1}^\infty$ by 
\begin{align}
T_{n+1}=T_n+m_nN_n,\label{deTn}
\end{align}
where the positive integer $T_1$ is given, and $\{m_n\}_{n=1}^\infty\subset \mathbb{N}$ depends on $\{N_n\}_{n=1}^{\infty}$ so that
\begin{align}
\sum_{n=1}^{\infty}n\left(\frac{N_n}{T_n}\right)^{\frac{1}{n}}<\infty.\label{sumnNnTnleqinfty}
\end{align}
We have the following lemma.
\begin{lemma}\label{lemmabasiccritical} Let $\{N_n\}_{n=1}^\infty\subset\mathbb{N}$ with $N_n>0\ (n\geq 1)$ be given. Assume that $\{T_n\}_{n=1}^\infty$ defined by \eqref{deTn} satisfy \eqref{sumnNnTnleqinfty}.
		Suppose that $F=F(n)>0\ (n\in\mathbb{N})$
satisfies, for every $n=1,2,\cdots,$ and  any $m=1,2,\cdots, {m}_{n}$, 
\begin{align}
\ln F(T_n+mN_n)^2\leq \ln F(T_n)^2-\left(1+\frac{1}{n}\right)\ln\frac{T_n+mN_n}{T_n},\label{crilnF}
\end{align}
and as $n\to\infty$,
\begin{align}
\frac{F(T_n+mN_n+j)}{F(T_n+mN_n)}=O(1)\label{criFTn}
\end{align}
uniformly in $j\in [1,N_n]$ and $m\in [0,m_n-1]$. Then $F(\cdot)\in\ell^2(\mathbb{N})$.
\end{lemma}
\begin{proof}
By \eqref{criFTn}, we only need to prove 
\begin{align}
\sum_{n=1}^{\infty}N_n\sum_{m=0}^{m_n-1}F(T_n+mN_n)^2<\infty.\label{Finl2}
\end{align}
	Let $m={m}_{n}$, \eqref{crilnF} implies that
\begin{align}
F(T_{n+1})^2\leq F(T_n)^2\left(\frac{T_{n+1}}{T_n}\right)^{-\left(1+\frac{1}{n}\right)}.\nonumber
\end{align}
Then
\begin{align}
F(T_{n+1})^2=O(1)\prod_{l=1}^n\left(\frac{T_{l+1}}{T_l}\right)^{-(1+\frac{1}{l})}.\nonumber
\end{align}
By \eqref{crilnF}, for any $0\leq m \leq m_n-1$, 
\begin{align}
F(T_n+mN_n)^2=O(1)\prod_{l=1}^{n-1}\left(\frac{T_{l+1}}{T_l}\right)^{-\left(1+\frac{1}{l}\right)}\left(\frac{T_n+mN_n}{T_n}\right)^{-\left(1+\frac{1}{n}\right)}.\nonumber
\end{align}
Therefore, 
\begin{align}
			N_n\sum_{m=0}^{{m}_{n}-1}{F}({T}_{n}+mN_n)^2= &O(1)N_n\prod_{l=1}^{n-1}\left(\frac{{T}_{l+1}}{{T}_l}\right)^{-(1+\frac{1}{l})}\sum_{m=0}^{{m}_{n}-1}\left(\frac{{T}_{n}+{m}{N_n}}{{T}_{n}}\right)^{-\left(1+\frac{1}{n}\right)}\nonumber\\
			=&O(1)N_n\prod_{l=1}^{n}{T}_l^{-\frac{1}{l^2}}\sum_{m=0}^{{m}_{n}-1}\left({{T}_{n}+{m}{N_n}}\right)^{-\left(1+\frac{1}{n}\right)}\nonumber\\
			=&O(1)\int_{0}^{{m}_{n}}\left(\frac{{T}_{n}}{{N_n}}+x\right)^{-(1+\frac{1}{n})}dx\nonumber\\
			=&O\left(n\left(\frac{N_n}{T_n}\right)^{\frac{1}{n}}\right).\nonumber
		\end{align}
Then by \eqref{sumnNnTnleqinfty} we can obtain \eqref{Finl2}. Hence, one obtains $F(\cdot)\in \ell^2(\mathbb{N})$.
\end{proof}

	Let us consider the case $k=\frac{p}{q}$  with odd $q\geq3$ first. 
	\begin{proof}[\textbf{Proof of Theorem \ref{main2} for odd $q$.}] Let
		\begin{align}
		a=\frac{\sin\pi k}{B_q}\label{oddaBq}	
		\end{align}	
and $\{T_n\}_{n=1}^{\infty}$ be defined by \eqref{deTn} with $T_1=1$,
$N_n\equiv q$ and 
\begin{align}
m_n=\left\lfloor\frac{e^{(n+1)^3}-e^{n^3}}{q}\right\rfloor,\ n\geq 1.\nonumber
\end{align}
Then one directly obtains \eqref{sumnNnTnleqinfty}.
Define the potential $V$ by $V(0)=0$ and on each $[{T}_{n},{T}_{n+1})$,
		\begin{align}
			V(T_n+j)=\frac{a+\delta_{n}}{1+
T_n+j}{\rm{sgn}}(\sin2\pi\theta(T_n+j)), \ 0\leq j<T_{n+1}-T_n,\label{devoddqcir}
		\end{align}
		where $\{\delta_{n}\}_{n=1}^{\infty}\subset \mathbb{R}$ with 
		$\lim_{n\to\infty}\delta_{n}=0$
		will be defined later. Therefore, one has
		$$\limsup_{n\to \infty}\abs{nV(n)}=\frac{\sin\pi k}{B_q}.$$ Then by \eqref{lr}, one has for any $1\leq j\leq q$,
\begin{align}
\frac{R(n+j)}{R(n)}=O(1)\label{oddO1}
\end{align}
as $n\to\infty$.
	By \eqref{lr} and \eqref{devoddqcir},  for any $n=1,2,\cdots,$ and $0\leq m \leq {m}_{n}-1$, one has
\begin{align}
&\ln R(T_n+(m+1)q)^2\nonumber\\
&=\ln R(T_n+mq)^2-\sum_{j=0}^{q-1}\frac{V(T_n+mq+j)}{\sin\pi k}\sin2\pi\theta(T_n+mq+j)+\frac{O(1)}{1+T_n^2+m^2}\nonumber\\
&=\ln R(T_n+mq)^2-\sum_{j=0}^{q-1}\frac{(a+\delta_n)\abs{\sin2\pi\theta(T_n+mq+j)}}{(1+T_n+mq)\sin\pi k}+\frac{O(1)}{1+T_n^2+m^2}.\label{criticaloddlnr}
\end{align}
Applying \eqref{bhgx}, \eqref{permupqlj} and \eqref{definitionofP_qtheta} one has
\begin{align}
\sum_{j=0}^{q-1}\abs{\sin2\pi\theta(T_n+mq+j)}=&\sum_{j=0}^{q-1}\abs{\sin2\pi\left(\theta(T_n+mq)+\frac{p}{q}j\right)}+\frac{O(1)}{1+T_n+m}\nonumber\\
=&\sum_{j=0}^{q-1}\abs{\sin2\pi\left(\theta(T_n+mq)+\frac{p}{q}{l_j}\right)}+\frac{O(1)}{1+T_n+m}\nonumber\\
=&\sum_{j=0}^{q-1}\abs{\sin2\pi\left(\theta(T_n+mq)+\frac{j}{q}\right)}+\frac{O(1)}{1+T_n+m}\nonumber\\
=&P(\theta(T_n+mq))+\frac{O(1)}{1+T_n+m}.\nonumber
\end{align}
By \eqref{perioPvarphi} and \eqref{Bq}, we can obtain 
\begin{align}
\sum_{j=0}^{q-1}\abs{\sin2\pi\theta(T_n+mq+j)}\geq qB_q+\frac{O(1)}{1+T_n+m}.\nonumber
\end{align}
Then by \eqref{criticaloddlnr} we conclude that
\begin{align}
\ln R(T_n+(m+1)q)^2\leq\ln R(T_n+mq)^2-\frac{(a+\delta_n)qB_q}{(1+T_n+mq)\sin\pi k}+\frac{O(1)}{1+T_n^2+m^2}\nonumber.
\end{align}
Therefore, there exist $\{\varepsilon_{n,m}\}_{n\geq 1, 0\leq m\leq m_n-1}\subset\R$ with 
\begin{align}
	\varepsilon_{n,m}=\frac{O(1)}{1+T_n+m},\nonumber
\end{align}
such that
\begin{align}
\ln R(T_n+(m+1)q)^2\leq \ln R(T_n+mq)^2-\frac{(a+\delta_n)(qB_q+\varepsilon_{n,m})}{(1+T_n+mq)\sin\pi k}.\nonumber
\end{align}
This implies
\begin{align}
\ln R(T_n+(m+1)q)^2\leq& \ln R(T_n)^2-\sum_{i=0}^{m}\frac{(a+\delta_n)(qB_q+\varepsilon_{n,i})}{(1+T_n+iq)\sin\pi k}\nonumber\\
\leq &\ln R(T_n)^2-\frac{(a+\delta_n)(B_q+\tilde{\varepsilon}_n)}{\sin\pi k}\ln \frac{T_{n}+(m+1)q}{T_n},\label{fllmad}
\end{align}
where $\tilde{\varepsilon}_n=\frac{O(1)}{1+n}$. Therefore, by \eqref{oddaBq} there exist $\{\delta_n\}_{n=1}^\infty\subset\R$ with $\delta_n=\frac{O(1)}{1+n}$ such that
\begin{align}
\frac{(a+\delta_n)(B_q+\tilde{\varepsilon}_n)}{\sin\pi k}\geq 1+\frac{1}{n}.\nonumber
\end{align}
By \eqref{oddO1}, \eqref{fllmad} and applying Lemma \ref{lemmabasiccritical} one can obtain that $R(\cdot)\in\ell^2(\mathbb{N})$. 
		We finish the proof for odd $q.$
		
	\end{proof}	
	
	Now let us consider the case $k=\frac{p}{q}$ with even $q\geq2$. 
	
	\begin{lemma}\cite[Lemma 3.3]{LiuDiscrete}
\label{l44} For even $q\geq2$, we have
		\begin{align}
			A_q&=\frac{1}{q}\max_{\varphi\in \left[0,\frac{1}{q}\right]}\sum_{j=0}^{q-1}\abs{\sin\left(\frac{2\pi}{q}j+2\pi\varphi\right)}=\frac{1}{q}\sum_{j=0}^{q-1}\abs{\sin\left(\frac{2\pi}{q}j+\frac{\pi}{q}\right)},\nonumber
		\end{align}
where $A_q$ is defined by \eqref{aqa}. 
	\end{lemma}
Recall that $(l_0,l_1,\cdots, l_{q-1})$ is a permutation of $(0,1,\cdots,q-1)$ such that for any $j=0,1,\cdots,q-1,$ 
\begin{align}
	\left[\frac{p}{q}l_j\right]=\frac{j}{q}.\nonumber
\end{align} For even $q\geq 2$,
we have
\begin{align}
\sin\left(\frac{\pi}{q}+\frac{2\pi}{q}pl_j\right)>0,\ j\in\left\{0,1,\cdots,\frac{q}{2}-1\right\},\label{sgz}
\end{align}
and
\begin{align}
\sin\left(\frac{\pi}{q}+\frac{2\pi}{q}pl_j\right)<0,\ j\in\left\{\frac{q}{2},\frac{q}{2}+1,\cdots,q-1\right\}.\label{slz}
\end{align}

	\begin{proof}[\textbf{Proof of Theorem \ref{main2} for even $q\geq2$.}]
Let 
	\begin{align}
		a=\frac{\sin\pi k}{A_q}\label{sinpikaqa}	
		\end{align}	
and $\{T_n\}_{n=1}^{\infty}$ be defined by \eqref{deTn} with $T_1=1$,
$N_n\equiv q$ and 
\begin{align}
m_n=\left\lfloor\frac{e^{(n+1)^3}-e^{n^3}}{q}\right\rfloor,\ n\geq 1.\nonumber
\end{align}
Then one directly obtains \eqref{sumnNnTnleqinfty}.

		Let $n_0$ be a large fixed positive integer. Define $V(n)=0$ for any $0\leq n\leq T_{n_0}-2$. Let $V(T_{n_0}-1)$ be such that 
\begin{align}
			2\pi\theta(T_{n_0})=\frac{\pi}{q} \mod 2\pi.\nonumber
		\end{align}
Define $V(n)$ on $[T_{n}+mq,T_{n}+(m+1)q)$ for any $n\geq n_0,0\leq m\leq m_n-1$, by
\begin{align}
V(T_n+mq+j)=\begin{cases}
				\frac{a+\delta_{n}}{1+{T}_{n}+mq+j}{\rm{sgn}}(\sin2\pi\theta({T}_{n}+mq+j)), & j\in\{l_0,l_1,\cdots,l_{\frac{q}{2}-1}\},\\
				\frac{a+\delta_{n}+\varepsilon_{n,m}}{1+{T}_{n}+mq+j}{\rm{sgn}}(\sin2\pi\theta({T}_{n}+mq+j)), & j\in\{l_{\frac{q}{2}},\cdots,l_{q-1}\},\label{deVeq}
			\end{cases}
\end{align}
where $\{\delta_n\}_{n=1}^{\infty}\subset\R$ with $\lim_{n\to\infty}\delta_n=0$, and $\{\varepsilon_{n,m}\}_{n\geq n_0;0\leq m\leq m_n-1}\subset\R$  with 
\begin{align}
	\varepsilon_{n,m}=\frac{O(1)}{1+T_n+m}\label{varepsilonnm}
\end{align}
will be defined later. Therefore, one has
$$\limsup_{n\to \infty}\abs{nV(n)}=\frac{\sin\pi k}{A_q}.$$
Then by \eqref{lr}, one has for any $1\leq j\leq q$,
\begin{align}
	\frac{R(n+j)}{R(n)}=O(1)\label{evenO1}
\end{align}
as $n\to\infty$.

We claim that for arbitrary $\{\delta_n\}_{n=1}^{\infty}\subset\R$ with $\lim_{n\to\infty}\delta_n=0$, we can always choose $\{\varepsilon_{n,m}\}_{n\geq n_0;0\leq m\leq m_n-1}\subset\R$ which satisfy \eqref{varepsilonnm} 
such that for any $n\geq n_0,\ 0\leq m\leq m_n-1$,
\begin{align}
2\pi\theta(T_{n}+mq)=\frac{\pi}{q}\mod2\pi.\label{tpm2pf}
\end{align}
Indeed, without loss of generality, suppose that 
\begin{align}
	2\pi \theta(T_n+mq)=\frac{\pi}{q}\label{tptef}
\end{align}
for some $n\geq n_0$ and $0\leq m\leq m_n-1$. 
Then by \eqref{bhgx}, there exists some constant $C>0$ such that for any $0\leq j\leq q-1$,
\begin{align}
	\abs{2\pi\theta(T_n+mq+j)-\left(\frac{\pi}{q}+\frac{2\pi}{q}pj\right)}\leq \frac{C}{1+T_n+m}\label{fi2pt}
\end{align}
Since $T_n\geq n_0q$ is large and $\sin \left(\frac{\pi}{q}+\frac{2\pi}{q}pj\right)\neq 0$ for any $0\leq j\leq q-1$, by \eqref{fi2pt} one has $$\mathrm{sgn}(\sin2\pi\theta(T_n+mq+j))=\mathrm{sgn}\left(\sin\left(\frac{\pi}{q}+\frac{2\pi}{q}pj\right)\right).$$ 
Then by \eqref{sgz}, \eqref{slz} and \eqref{deVeq},
\begin{align}
V(T_n+mq+j)=\begin{cases}
				\frac{a+\delta_{n}}{1+{T}_{n}+mq+j}, & j\in\{l_0,l_1,\cdots,l_{\frac{q}{2}-1}\},\\
				-\frac{a+\delta_{n}+\varepsilon_{n,m}}{1+{T}_{n}+mq+j}, & j\in\{l_{\frac{q}{2}},\cdots,l_{q-1}\}.\label{devcrieq}
			\end{cases}
\end{align}
Then by \eqref{n1n} and \eqref{bhgx}, one can obtain
\begin{align}
\theta(&T_n+(m+1)q)\nonumber\\
=&\theta(T_n+mq)+ kq+\sum_{j\in\left\{l_0,l_1,\cdots,l_{\frac{q}{2}-1}\right\}}\sin^2\pi\left(\theta(T_n+mq)+\frac{p}{q}j\right)\frac{a+\delta_n}{\pi\sin\pi k(1+T_n+mq)}\nonumber\\
&-\sum_{j\in\left\{l_{\frac{q}{2}},\cdots,l_{q-1}\right\}}\sin^2\pi\left(\theta(T_n+mq)+\frac{p}{q}j\right)\frac{a+\delta_n+\varepsilon_{n,m}}{\pi\sin\pi k(1+T_n+mq)}+\frac{O(1)}{1+T_n^2+m^2}.\nonumber
\end{align}
Applying \eqref{permupqlj} and \eqref{tptef} we can conclude that
\begin{align}
2\pi&\theta(T_n+(m+1)q)\nonumber\\
=&\frac{\pi}{q}+2\pi kq+2\sum_{j=0}^{\frac{q}{2}-1}\sin^2\pi\left(\frac{1}{2q}+\frac{j}{q}\right)\frac{a+\delta_n}{\sin\pi k(1+T_n+mq)}\nonumber\\
&-2\sum_{j=\frac{q}{2}}^{q-1}\sin^2\pi\left(\frac{1}{2q}+\frac{j}{q}\right)\frac{a+\delta_n+\varepsilon_{n,m}}{\sin\pi k(1+T_n+mq)}+\frac{O(1)}{1+T_n^2+m^2}.\nonumber
\end{align}
By the fact that
%
		\begin{align}
			\sum_{j=0}^{\frac{q}{2}-1}\sin^2\pi\left(\frac{1}{2q}+\frac{j}{q}\right)=\sum_{j=\frac{q}{2}}^{q-1}\sin^2\pi\left(\frac{1}{2q}+\frac{j}{q}\right)=\frac{q}{4}\nonumber
		\end{align}
there exists a family $\{\varepsilon_{n,m}\}_{n\geq n_0;0\leq m\leq m_n-1}\subset\R$ satisfying \eqref{varepsilonnm}, such that we have \eqref{tpm2pf}.  

By \eqref{lr}, \eqref{bhgx}, \eqref{tptef} and \eqref{devcrieq}, for any $n\geq n_0$ and $0\leq m\leq m_n-1$, one has
		\begin{align}
			\ln &R({T}_{n}+(m+1)q)^2\nonumber\\
=&\ln R({T}_{n}+mq)^2-\sum_{j=0}^{q-1}\frac{V({T}_{n}+mq+j)}{\sin\pi k}\sin\left(\frac{\pi}{q}+\frac{2\pi}{q}pj\right)+\frac{O(1)}{1+{T}_{n}^2+m^2}\nonumber\\
=&\ln R({T}_{n}+mq)^2-\sum_{j\in\left\{l_0,l_1,\cdots,l_{\frac{q}{2}-1}\right\}}\frac{(a+\delta_n)\sin\left(\frac{\pi}{q}+\frac{2\pi}{q}pj\right)}{(1+T_n+mq)\sin\pi k}\nonumber\\
&+\sum_{j\in\left\{l_{\frac{q}{2}},\cdots,l_{q-1}\right\}}\frac{(a+\delta_n+\varepsilon_{n,m})\sin\left(\frac{\pi}{q}+\frac{2\pi}{q}pj\right)}{(1+T_n+mq)\sin\pi k}+\frac{O(1)}{1+{T}_{n}^2+m^2}.\nonumber
		\end{align}
Then by \eqref{sgz}, \eqref{slz} and Lemma \ref{l44}, there exist $\{\tilde{\varepsilon}_{n,m}\}_{n\geq n_0;0\leq m\leq m_n-1}\subset \R$ with
$$\tilde{\varepsilon}_{n,m}=\frac{O(1)}{1+{T}_{n}+m},$$
such that
\begin{align}
	\ln R({T}_{n}+(m+1)q)^2=&\ln R({T}_{n}+mq)^2-\frac{q(a+\delta_n)(A_q+\tilde{\varepsilon}_{n,m})}{(1+T_n+mq)\sin\pi k}.\nonumber
\end{align}
Therefore,
\begin{align}
\ln R(T_n+(m+1)q)^2\leq& \ln R(T_n)^2-\sum_{i=0}^{m}\frac{q(a+\delta_n)(A_q+\tilde{\varepsilon}_{n,i})}{(1+T_n+iq)\sin\pi k}\nonumber\\
\leq &\ln R(T_n)^2-\frac{(a+\delta_n)(A_q+\tilde{\varepsilon}_n)}{\sin\pi k}\ln \frac{T_{n}+(m+1)q}{T_n},\nonumber
\end{align}
where $\tilde{\varepsilon}_n=\frac{O(1)}{1+n}$. Therefore, by \eqref{sinpikaqa} there exist $\{\delta_n\}_{n=1}^\infty\subset\R$ with $\delta_n=\frac{O(1)}{1+n}$ such that
\begin{align}
\frac{(a+\delta_n)(A_q+\tilde{\varepsilon}_n)}{\sin\pi k}\geq 1+\frac{1}{n}.\nonumber
\end{align}
By \eqref{evenO1} and applying Lemma \ref{lemmabasiccritical} one can obtain that $R(\cdot)\in\ell^2(\mathbb{N})$. 
		We finish the proof for even $q$.
	\end{proof}

	Next let us consider the case with irrational $k$. 
	\begin{lemma}\cite[{Lemma 3.4}]{LiuDiscrete} \label{l43}
		Let $\theta(n)$ be determined by \eqref{ctt} with irrational $k$. Then for any positive integer $n$, there exists $N_{n}\in\mathbb{N}$ such that for any $e_{n}\in\mathbb{N}$ large enough,
		\begin{align}
			\abs{\left(\frac{1}{N_{n}}\sum_{j=e_n}^{e_n+N_{n}-1}\abs{\sin2\pi\theta(j)}\right)-A_0}\leq\frac{1}{{n}},
		\end{align}
		where $A_0$ is defined by \eqref{da0}.	
		\end{lemma}

	\begin{proof}[\textbf{Proof of Theorem \ref{main2} for irrational $k$ ($q=0$).}]
Let
	\begin{align}
			a=\frac{\sin\pi k}{A_0}\label{aA0sinpik}
		\end{align}
and $\{N_n\}_{n=1}^{\infty}\subset \mathbb{N}$, $\{e_n\}_{n=1}^\infty\subset\mathbb{N}$ be defined as in Lemma \ref{l43}.  Let $V(n)=0$ for any $n\in[0,T_1)$. Define $V$ on  $[T_{n},T_{n+1})$ for any $n\geq 1$ by
		\begin{align}
			V(T_n+j)=\frac{a+\delta_{n}}{1+T_n+j}{\rm{sgn}}(\sin2\pi\theta(T_n+j)),\ \ 0\leq j<T_{n+1}-T_n,\label{defvcriirra}
		\end{align}
		where $\{\delta_n\}_{n=1}^\infty\subset\R$ with $\lim_{n\to\infty}\delta_n=0$ will be defined later. Therefore, one has
		$$\limsup_{n\to \infty}\abs{nV(n)}=\frac{\sin\pi k}{A_0}.$$ 
		Here  $\{T_{n}\}_{n=1}^{\infty}\subset \mathbb{N}$ is defined by
\eqref{deTn}
with $T_1\geq e_1, T_1\gg N_1$, $\{m_n\}_{n=1}^\infty\subset\mathbb{N}$ satisfy \eqref{sumnNnTnleqinfty} and
for any $n>1$,
		\begin{align}\label{Tomeganomega}
			T_{n}\geq e_{n}, 
		\end{align}
and
		\begin{align}\label{NomegaTomega}
			\frac{N_{n}}{T_{n}}\leq \frac{1}{n^3}.
		\end{align}
By \eqref{lr} and \eqref{defvcriirra}, for any $n\geq 1,$ and any $0\leq m\leq m_n-1$, $1\leq j\leq N_n$, one has
\begin{align}
\ln R(T_n+mN_n+j)^2=&\ln R(T_n+mN_n)^2+\sum_{i=0}^{j-1}\frac{O(1)}{1+T_n+mN_n}\nonumber\\
=&\ln R(T_n+mN_n)^2+O\left(\frac{N_n}{T_n}\right).\nonumber
\end{align}
Hence, by \eqref{NomegaTomega} one obtains that as $n\to\infty$,
\begin{align}
\frac{R(T_n+mN_n+j)}{R(T_n+mN_n)}=O(1)\label{irraO1}
\end{align} 
uniformly in $j\in [1,N_n]$ and $m\in [0,m_n-1]$.
By \eqref{lr} and \eqref{defvcriirra} again, for any $0\leq m\leq m_{n}-1$, one has
		\begin{align}
			\ln R(T_{n}+(m+1)N_{n})^2=& \ln R(T_{n}+mN_{n})^2-\sum_{j=0}^{N_{n}-1}\frac{(a+\delta_{n})\abs{\sin2\pi\theta(T_{n}+mN_{n}+j)}}{(1+T_{n}+mN_{n}+j)\sin\pi k}\nonumber\\
&+\sum_{j=0}^{N_{n}-1}\frac{O(1)}{1+T_{n}^2+m^2N_n^2+j^2}.\nonumber
		\end{align}
Then by \eqref{Tomeganomega}, \eqref{NomegaTomega} and Lemma \ref{l43}, there exist $\{\varepsilon_{n}\}_{n=1}^\infty
\subset\R$ with $\varepsilon_{n}=\frac{O(1)}{1+{n}}$ such that  for any $0\leq m\leq m_{n}-1$,
\begin{align}
\ln R(T_{n}+(m+1)N_{n})^2\leq \ln R(T_{n}+mN_{n})^2-\frac{N_{n}(a+\delta_{n})(A_0+\varepsilon_{n})}{(1+T_{n}+mN_{n})\sin\pi k}.\nonumber
\end{align}
Thus, we can conclude that
		\begin{align}
			\ln R(T_{n}+(m+1)N_{n})^2\leq&\ln R(T_{n})^2-\sum_{i=0}^{m}\frac{N_{n}(a+\delta_{n})(A_0+\varepsilon_{n})}{(1+T_{n}+iN_{n})\sin\pi k}\nonumber\\
			\leq &\ln R(T_{n})^2-\frac{(a+\delta_{n})(A_0+\tilde{\varepsilon}_{n})}{\sin\pi k}\ln\frac{T_{n}+(m+1)N_{n}}{T_{n}},\nonumber
		\end{align}
where $\tilde{\varepsilon}_n=\frac{O(1)}{1+n}$. By \eqref{aA0sinpik}, there exist $\{\delta_{n}\}_{n=1}^{\infty}\subset\R$ with  $\delta_n=\frac{O(1)}{1+n}$ such that 
		\begin{align}
			\frac{(a+\delta_{n})(A_0+\tilde{\varepsilon}_{n})}{\sin\pi k}\geq1+\frac{1}{n}.\nonumber
		\end{align}
	By \eqref{irraO1} and applying Lemma \ref{lemmabasiccritical} one can obtain that $R(\cdot)\in\ell^2(\mathbb{N})$. 
		We finish the proof for irrational $k$.
	\end{proof}

\appendix
\section{Proof of the Continuous Analogue}\label{appena}

\begin{proof}[\textbf{Proof of Theorem \ref{thmcontinuous}}]
	For simplicity, assume that $0<b<c<1$ and $\delta\ll \min\{b,1-c\}.$ 
	Since 
	\begin{align}
	b=\liminf_{t\to\infty}f(t),\ \ c=\limsup_{t\to\infty}f(t),\nonumber
	\end{align}
we can assume that $t_\delta\gg\frac{1}{\delta}$ with $f(t_\delta)\in (c-\delta,c+\delta)$, and for any $t\geq t_\delta$,
	\begin{align}
		f(t)\in (b-\delta,c+\delta).\nonumber
	\end{align}
Let $M$ be sufficiently large, to be specified later, and define
\begin{align}
	\mathcal{I}_i=\{t_\delta\leq t\leq& t_\delta+M: c-i\delta\leq f(t)<c+\delta\},\ i=1,2,\cdots,K_1-2.\label{ap1}
\end{align}Clearly, we have
\begin{align}
	\mathcal{I}_1\subset \mathcal{I}_2\subset\cdots\subset \mathcal{I}_{K_1-2}.\label{apsub}
\end{align}
Then we can prove the result by letting
\begin{align}
	&I_1=\mathcal{I}_1,\label{ap2}\\
	I_i=\mathcal{I}_i\setminus &\mathcal{I}_{i-1},\ i=2,3,\cdots,K_1-2,\label{ap3}
\end{align}
and
\begin{align}
	I_{K_1-1}=[t_\delta,t_\delta+M]\setminus \mathcal{I}_{K_1-2}.\label{ap4}
\end{align}

Indeed, by \eqref{deofVi} and \eqref{ap1}-\eqref{ap4} one can obtain \eqref{appeniinvia}. Since $\int_{x}^yf'(t)dt=f(y)-f(x)$, one concludes that
\begin{align}
	\int_{\mathcal{I}_i}f'(t)dt=f(t_i)-f(t_\delta),\nonumber\ i=1,2,\cdots,K_1-2,
\end{align}
where $t_i=\sup\mathcal{I}_i$. It follows that
\begin{align}
	\abs{\int_{\mathcal{I}_i}f'(t)dt}\leq 1.\nonumber
\end{align}
Hence, by \eqref{apsub} and \eqref{ap3} we can obtain \eqref{appenee} for $i=1,2,\cdots,K_1-2$. Similarly, one can prove \eqref{appenee} for $i=K_1-1$.

By the definition, we know that
\begin{align}
	\{t_\delta\leq t\leq t_\delta+M:c-i\delta \leq f(t)<c-(i-1)\delta\}\subset I_i,\ i=1,2,\cdots,K_1-1.\label{apfi}
\end{align}
Therefore, by \eqref{apfprO}, each $I_i$ contains many intervals $I$ such that
\begin{align}
	\delta=\abs{\int_If'(t)dt}\leq O(1)\int_I\frac{1}{1+{t}}dt.\nonumber 
\end{align}
Then by letting $M$ be large so that each $I_i$ contains enough such intervals, one can prove \eqref{appenIifractionlargecasec}.
\end{proof} 

\section{Behavior Under Sign-Type Potentials}
Assume that $k=\frac{p}{q}$ with odd $q\geq 3$. We prove that under the sign-type potential, the Pr\"ufer angles $[\theta(nq)]$ will be close to $\frac{j}{2q}$ for some $j\in\{0,\cdots,2q\}$ so that $P(\theta(nq))$ attains its minimum \eqref{Bq}, from which we understand why the classical sign-type potential cannot yield a better bound.
\begin{example}
	Let $a>0$ be fixed and $\theta(n)$ be determined by \eqref{ctt} with
	\begin{align}
		V(n)=V_0(n)=\frac{a}{1+n}\mathrm{sgn}(\sin2\pi\theta(n)).\label{apdev}
	\end{align}
	Then either there exist two disjoint subsets $N_1,N_2$ with $\mathbb{N}=N_1\cup N_2$ such that 
	\begin{align}
		\lim_{N_1\ni n\to\infty}[\theta(nq)]=0,\ \lim_{N_2\ni n\to\infty}[\theta(nq)]=1,\nonumber
	\end{align}
	or $\lim_{n\to\infty}[\theta(nq)]$ exists and 
	\begin{align}
		\lim_{n\to\infty}[\theta(nq)]\in\left\{\frac{j}{2q}\right\}_{j=0}^{2q}.
	\end{align}
	In particular, by \eqref{perioPvarphi} and \eqref{Bq}, $\lim_{n\to\infty}P([\theta(nq)])=qB_q$.
\end{example}
\begin{proof} We argue by contradiction. Suppose that the conclusion fails.
	By Corollary \ref{corangle}, there exists $\varphi_0\in (0,1)\setminus\left\{\frac{j}{2q}\right\}_{j=0}^{2q}$, $\{n_i\}_{i=1}^\infty\subset \mathbb{N}$ with $\lim_{i\to\infty}n_i=\infty$, and 
	\begin{align}
		\delta_0=\frac{1}{10}\min\left\{\abs{\varphi_0-\frac{j}{2q}}:j=0,1,\cdots,2q\right\},\label{apd0v0}
	\end{align}
	 such that for any $i\in\mathbb{N}$,
	\begin{align}
		[\theta(n_iq)]\in (\varphi_0-\delta_0,\varphi_0+\delta_0).\label{apthetain}
	\end{align}
	Suppose that  
	\begin{align}
		 \varphi_0\in \left(\frac{j_0}{q},\frac{j_0}{q}+\frac{1}{2q}\right)\cup \left(\frac{j_0}{q}+\frac{1}{2q},\frac{j_0+1}{q}\right)\label{apvarphi}
	\end{align}for some $j_0\in\{0,1,\cdots,q-1\}$. 
	Let  $m$ be large enough with 
	\begin{align}
		[\theta(mq)]\in \left(\frac{j_0}{q}+\delta_0,\frac{j_0}{q}+\frac{1}{2q}-\delta_0\right)\cup\left(\frac{j_0}{q}+\frac{1}{2q}+\delta_0,\frac{j_0+1}{q}-\delta_0\right).\label{aptmv}
	\end{align}
	By \eqref{thetan+jn} and \eqref{aptmv}, we have for each $j=0,1,\cdots,q-1$, that
	\begin{align}
	\mathrm{sgn}(\sin2\pi\theta(mq+j))=\mathrm{sgn}\left(\sin2\pi\left(\theta(mq)+kj\right)\right).\label{apsg}
	\end{align}
By \eqref{n1n}, \eqref{thetan+jn}, \eqref{apdev} and \eqref{apsg},
\begin{align}
	\theta((m+1)q)=&\theta(mq)+kq+\sum_{j=0}^{q-1}\sin^2\pi\left(\theta(mq)+kj\right)\frac{V(mq+j)}{\pi\sin\pi k}+\frac{O(1)}{1+m^2}\nonumber\\
	=&\theta(mq)+kq+a\sum_{j=0}^{q-1}\sin^2\pi\left(\theta(mq)+kj\right)\frac{\mathrm{sgn}(\sin2\pi\theta(mq+j))}{(1+mq)\pi\sin\pi k}+\frac{O(1)}{1+m^2}\nonumber\\
\label{aptmm}=&\theta(mq)+kq+a\sum_{j=0}^{q-1}\sin^2\pi\left(\theta(mq)+kj\right)\frac{\mathrm{sgn}(\sin2\pi(\theta(mq)+kj))}{(1+mq)\pi\sin\pi k}+\frac{O(1)}{1+m^2}.
\end{align}
Recall that  $(l_0,\cdots, l_{q-1})$ is a permutation of $(0,1,\cdots,q-1)$ so that for any $j=0,\cdots,q-1,$ we have
\begin{align}
	\left[\frac{p}{q}l_j\right]=\frac{j}{q}.\nonumber
\end{align}
Then by \eqref{aptmm}, one has
\begin{align}
	\theta((m+1)q)=&\theta(mq)+kq+a\sum_{j=0}^{q-1}\sin^2\pi\left(\theta(mq)+\frac{j}{q}\right)\frac{\mathrm{sgn}(\sin2\pi(\theta(mq)+\frac{j}{q}))}{(1+mq)\pi\sin\pi k}\nonumber\\
	&+\frac{O(1)}{1+m^2}\nonumber\\
	=&\theta(mq)+kq+\frac{aT(\theta(mq))}{(1+mq)\pi\sin\pi k}+\frac{O(1)}{1+m^2},\label{appttmmt}
\end{align}
where $T(\cdot)$ is defined by \eqref{definitionofT_qtheta}. 

By \eqref{bhgx}, there exists $\delta_1>0$ such that for any  $[\theta(nq)]\in \left(\frac{j_0}{q}+\delta_0,\frac{j_0+1}{q}-\delta_0\right)$ with large enough $n$,
\begin{align}
	\abs{ [\theta((n+1)q)]-[\theta(nq)]}\leq \frac{\delta_1}{1+nq}.\label{apptnqd1}
\end{align}

By \eqref{perioTvarphi}, \eqref{ssccs} and \eqref{tqpdy0}, there exists $\delta_2>0$ such that for any $[\theta(nq)]\in \left(\frac{j_0}{q}+\delta_0,\frac{j_0}{q}+\frac{1}{2q}-\delta_0\right)$ with large enough $n$,
\begin{align}
	T(\theta(nq))\geq \delta_2,\label{apptnd1}
\end{align}
and for any $[\theta(nq)]\in\left(\frac{j_0}{q}+\frac{1}{2q}+\delta_0,\frac{j_0+1}{q}-\delta_0\right)$ with large enough $n$,
\begin{align}
	T(\theta(nq))\leq-\delta_2.\label{apptnd2}
\end{align}

By 
\eqref{appttmmt}, \eqref{apptnd1} and \eqref{apptnd2}, there exists $\delta_3>0$, such that if 
$$[\theta(mq)]\in \left(\frac{j_0}{q}+\delta_0,\frac{j_0}{q}+\frac{1}{2q}-\delta_0\right),$$ 
then for large enough $m$,
\begin{align}
	[\theta((m+1)q)]-[\theta(mq)]\geq \frac{\delta_3}{1+mq},\label{apfitd2}
\end{align}
and if
$$[\theta(mq)]\in \left(\frac{j_0}{q}+\frac{1}{2q}+\delta_0,\frac{j_0+1}{q}-\delta_0\right),$$
then for large enough $m$,
\begin{align}
	[\theta((m+1)q)]-[\theta(mq)]\leq-\frac{\delta_3}{1+mq}.\label{apfitd3}
\end{align}

Combining with \eqref{apptnqd1}, \eqref{apfitd2} and \eqref{apfitd3}, we know that for any large $m$ with \eqref{aptmv}, there exists $M\in\mathbb{N}$ such that for any $i>M$,
\begin{align}
	[\theta((m+i)q)]\in \left[\frac{j_0}{q}+\frac{1}{2q}-\delta_0,\frac{j_0}{q}+\frac{1}{2q}+\delta_0\right],\nonumber
\end{align}
which contradicts \eqref{apthetain}. This completes the proof.

\vspace{2mm}\noindent \textbf{Data availability}
No data was used for the research described in the article.

\vspace{2mm}\noindent \textbf{Conflict of interest} The authors report no conflict of interest.

\end{proof}

\end{document}